\documentclass{article}
\usepackage{float}
\usepackage[flushleft]{threeparttable}%
\usepackage{amsmath}
\usepackage{geometry}
\usepackage[utf8]{inputenc}
\usepackage[nottoc]{tocbibind}
\usepackage{bm}
\usepackage{enumitem}
\usepackage{amssymb}
\usepackage{cite}
\usepackage[backref]{hyperref} %
\usepackage{mathrsfs}
\usepackage{mathtools}
\usepackage{caption}
\usepackage{indentfirst}
\usepackage{array}
\usepackage{amsfonts}
\usepackage{amsthm}
\usepackage{listings}
\usepackage{graphicx}
\usepackage{centernot}
\usepackage{xcolor}
\usepackage{setspace}
\usepackage{natbib}
\usepackage{booktabs}
\usepackage{subcaption}
\usepackage{pdflscape}
\usepackage[title]{appendix}
\newtheorem{condition}{Condition}[section]

\newtheorem{theorem}{Theorem}
\newtheorem{remark}{Remark}
\newtheorem{lemma}{Lemma}

\definecolor{tianedit}{RGB}{180,0,0} %

\newcommand{\normm}[1]{\lVert#1\rVert}
\newcommand{\E}{\mathbb{E}}

\newcommand{\R}{\mathbb{R}}

\newtheorem{example}{Example}[section]

\title{Conditional-Moment Estimation and Inference in the BLP Model}
\author{Hua Jin\thanks{Department of Economics, University College London. Email: \href{mailto:uctphji@ucl.ac.uk}{uctphji@ucl.ac.uk}} 
\and Tian Xie\thanks{Department of Economics, University of Melbourne. Email: \href{mailto:tian.xie@unimelb.edu.au}{tian.xie@unimelb.edu.au}}}
\date{This draft: \today}

\begin{document}

\maketitle
\doublespacing

\begin{abstract}
The random-coefficient demand model of Berry, Levinsohn, and Pakes (1995) is commonly estimated by the generalized method of moments (GMM), using an unconditional moment restriction with a fixed set of instruments. Identification of the model, however, rests on a conditional moment restriction. The two are not equivalent: the unconditional restriction may admit additional parameter values. We construct a counterexample in which the model is identified by the conditional restriction yet standard GMM is not, even with the optimal instrument. Building directly on the identifying restriction, we propose a two-step estimator, following \cite{ai2003efficient}, that first estimates the relevant conditional expectations nonparametrically and then selects the structural parameters by a conditional-variance-weighted minimum-distance criterion; standard GMM is recovered as the special case of a linear projection onto finitely many instruments. We establish $\sqrt T$-asymptotic normality for the proposed estimator, and we develop the theory for both kernel and series implementations of the first stage. The two implementations share a common limiting distribution, attaining the semiparametric efficiency bound. Simulation evidence illustrates the consequences of the identification gap and demonstrates that the proposed estimator outperforms standard GMM in finite samples.

\end{abstract}

\thispagestyle{empty}
\clearpage

\setcounter{page}{1} 

\section{Introduction}
{
In demand estimation for differentiated products, the \citet*{berry1995automobile} or BLP approach has become the workhorse framework of empirical industrial organization (IO). \cite{nevo2000mergers,nevo2001measuring}, for instance, examines competition and merger effects in the ready-to-eat cereal market, \cite{petrin2002quantifying} measures consumer welfare generated by the minivan's entry, \cite{berry1999voluntary} assess how voluntary export restraints reshaped the U.S. auto market, and \cite{goldberg2001evolution} document price dispersion across European car markets, among many other applications.

However, estimation and identification in this literature rest on different moment conditions. The standard BLP estimation approach uses the generalized method of moments (GMM), which imposes an unconditional moment restriction on a fixed set of instruments. The unconditional moment using optimal-instrument constructions or feasible approximations to them is commonly applied in practice (see, e.g., \cite{berry2021foundations}). By contrast, identification is assumed on the conditional moments. The conditional moments provide basis for the optimal instrument and imply the unconditional moments, but estimation based on the unconditional moments only exploits a necessary condition. It has been pointed out by the seminal work of \citet{dominguez2004consistent} that this gap can be consequential, leading to potential failure of identification by unconditional moments even when the instrument is optimal. Whether it is consequential in the BLP model, however, has not been examined, and the tension is not widely recognized in the applied literature.

In this paper, we directly show that this identification gap can arise in the BLP model. Specifically, we construct an example in which the conditional moment restriction point-identifies the structural parameters, whereas the population GMM unconditional moments with the optimal instrument admit an additional solution. As such, it raises caution for the common BLP practice of using GMM with the unconditional moment restriction, even with the optimal instrument. Motivated by this finding, we propose a semiparametric estimator that directly implements the conditional moment restriction.

The answer is not a generic feature of all conditional-moment models. It is well known that the conditional restriction implies the unconditional one, so GMM imposes a necessary condition for identification; what is not immediate is whether that necessary condition is also sufficient in the BLP model. In a linear model, a full-rank collection of linearly independent unconditional moments with dimension equal to the parameters delivers identification \citep{dominguez2004consistent}. This paper's contribution is twofold. First, it provides an explicit example showing that the BLP model can fail to be identified by the unconditional moments. The example demonstrates how the nonlinearity of the BLP inverse can generate additional solutions to the population GMM moments. Second, the paper establishes properties of estimation and inference by directly building on the conditional moments under BLP-specific conditions.

To make the point concrete, we consider a simplified BLP model with $J=2$ products, where the market share of product $j$ in market $t$ is given by
\begin{align*}
    s_{jt} = \int_{-\infty}^{+\infty} \frac{\exp\left(p_{jt}\alpha_0 + p_{jt} \sigma_0 v  \right)}{1 + \sum_{i=1}^2 \exp\left( p_{it}\alpha_0 + p_{it} \sigma_0 v \right) } \phi (v) dv.
\end{align*}
Here, $p_{jt}$ is the price of product $j$ in market $t$, $\sigma_0$ is the standard deviation of the random coefficient, $\alpha_0$ is the price coefficient, $v$ captures consumer heterogeneity, and $\phi(\cdot)$ is the probability density function of the standard normal distribution. The parameters of interest are $(\sigma_0,\alpha_0)$, and an instrument $z_{jt}$ is also observed alongside $s_{jt}$ and $p_{jt}$.  

For any candidate $\sigma$, not necessarily equal to $\sigma_0$, the share system can be inverted for the mean utilities, which delivers the conditional moment restriction $\E\left[y_j \left(\sigma, p_{1t}, p_{2t}\right) - \alpha p_{jt} \mid z_{jt} \right] =0$. Two features of the notation are worth flagging. The inverse depends on $\sigma$ but not on $\alpha$. And because there is no unobserved characteristic here, the shares are a deterministic function of the two prices at the true parameter, $s_t=s(p_{1t},p_{2t})$; substituting them out gives the inverse as $y_j(\sigma,p_{1t},p_{2t})$, suppressing the shares that it in fact depends on. The following example shows that the unconditional moment built on the optimal instrument can fail to identify $(\sigma_0,\alpha_0)$ even though the conditional restriction identifies them.
\begin{example}\label{eg:counter}
Let $J=2$. Set $\sigma_0 =1$, $\alpha_0 =-2$. Assume $p_{1t}$ and $p_{2t}$ follow the same discrete distribution: $\Pr\left(p_{jt} = -4\right) = 0.6$, $\Pr\left(p_{jt}=2\right) = 0.3$, $\Pr\left(p_{jt}=3\right) = 0.1$. Further assume $z_{jt} = p_{jt}$, i.e., the price is exogenous. The optimal instrument based on $z_{jt}$ is \(z_{jt}^\ast = \left(p_{jt}, \E\left[\frac{\partial y_j}{\partial \sigma }\left(\sigma_0, p_{1t}, p_{2t}\right) \mid p_{jt} \right] \right)'\).
The unconditional moment $\E\left[\left(y_j \left(\sigma, p_{1t}, p_{2t}\right) - \alpha p_{jt}\right) z_{jt}^\ast\right]=0$ obtains identification of $(\sigma_0, \alpha_0)$ iff $g_2\left(\sigma\right) = 0$ implies $\sigma = \sigma_0$, where
\begin{align*}
    g_2 \left(\sigma \right)&= \E\left[y_j\left(\sigma, p_{1t}, p_{2t}\right)  \E\left[\frac{\partial y_j}{\partial \sigma }\left(\sigma_0, p_{1t}, p_{2t}\right) \mid p_{jt} \right] \right] \\
    &\quad - \frac{\E\left[ y_j\left(\sigma, p_{1t}, p_{2t}\right) p_{jt}  \right]}{ \E\left[p_{jt}^2\right]} \E\left[ p_{jt} \E\left[\frac{\partial y_j}{\partial \sigma }\left(\sigma_0, p_{1t}, p_{2t}\right) \mid p_{jt} \right]\right].
\end{align*}
Although the close-form expression of $g_2\left(\sigma\right)$ is not available, Figure \ref{fig:example:optimal_IV} plots the numerical evaluation.
\begin{figure}[t]
    \begin{subfigure}{0.5\textwidth}
        \centering
        \includegraphics[width=0.7\linewidth]{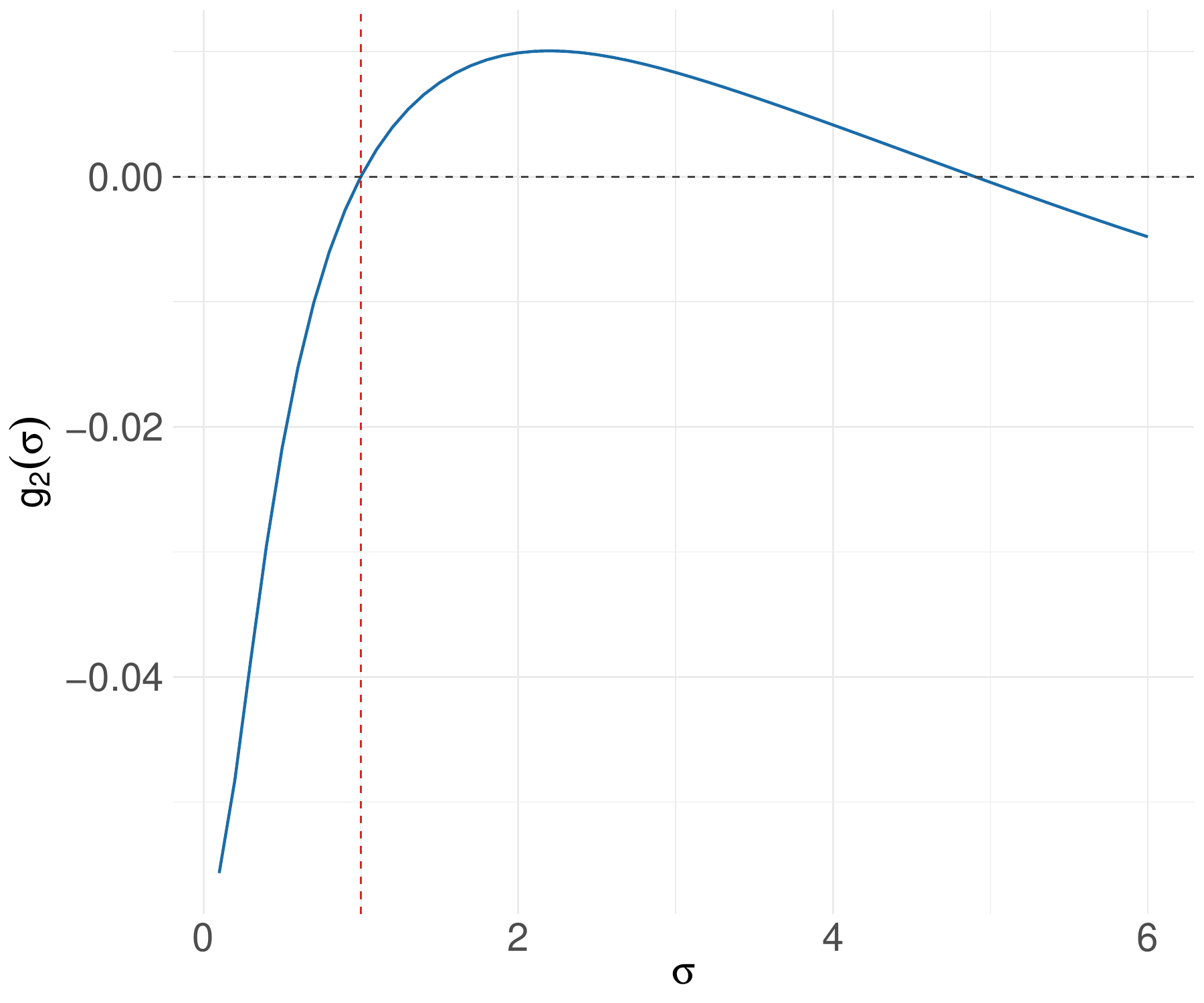}
        \caption{$g_2\left(\sigma\right)$ with Optimal IV}
        \label{fig:example:optimal_IV}
    \end{subfigure}%
    \begin{subfigure}{0.5\textwidth}
        \centering
        \includegraphics[width=0.7\linewidth]{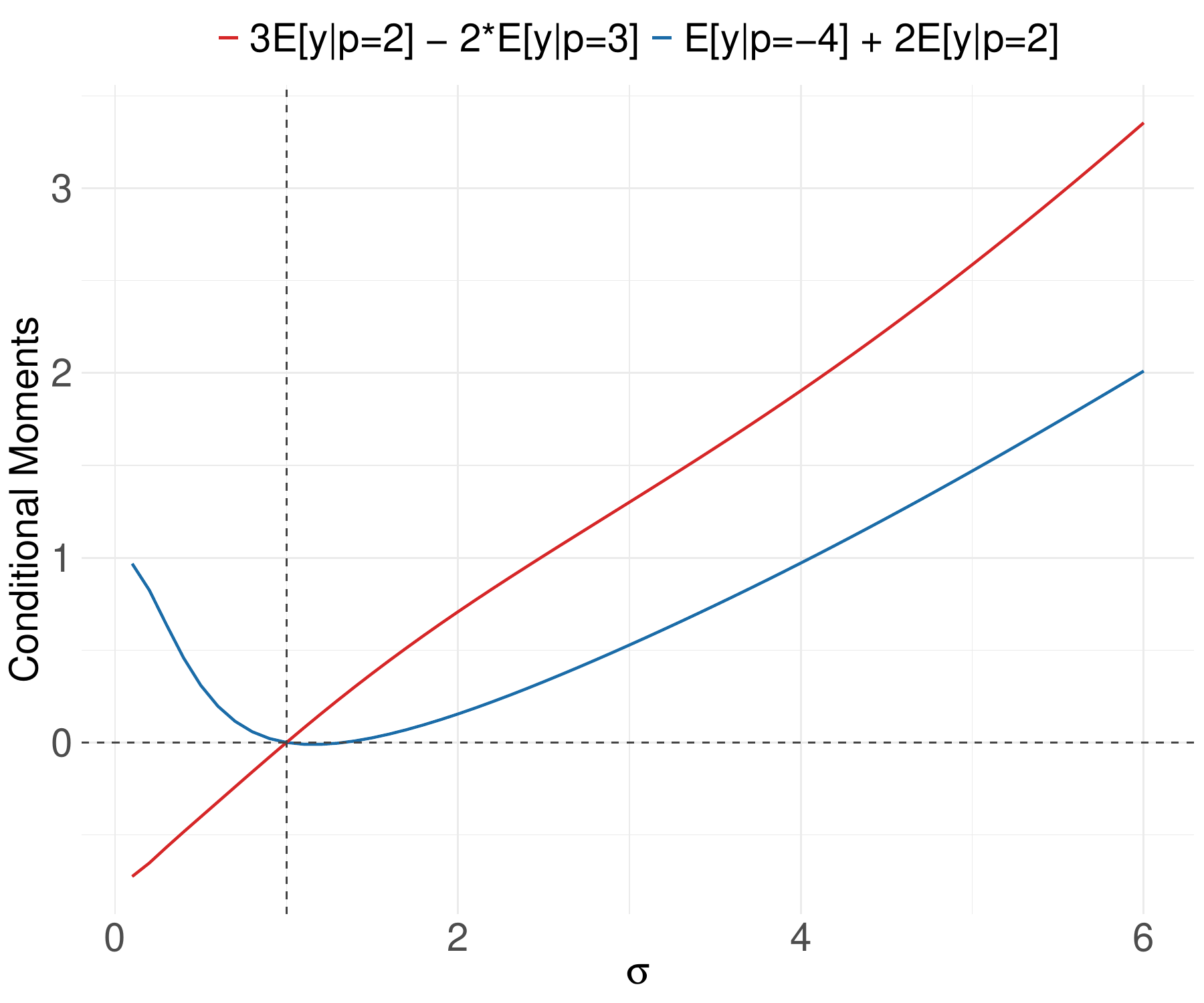}
        \caption{Two Functions of $\sigma$ from the Conditional Moments}
        \label{fig:example:conditional}
    \end{subfigure}
\end{figure}
The function has a second zero for $g_2\left(\sigma\right)$ in addition to $\sigma_0 =1$, so idenfication based on the unconditional moment fails.

By contrast, the conditional moment $\E\left[y_j \left(\sigma, p_{1t}, p_{2t}\right) - \alpha p_{jt} \mid z_{jt} \right] =0$ identifies $(\sigma_0, \alpha_0)$. Indeed, $(\sigma_0, \alpha_0)$ is identified if an only if $\sigma_0$ is the unique zero for the following two functions jointly:
\begin{align*}
    \E\left[ y_1 \left(\sigma, p_{1t}, p_{2t} \right) \mid p_{1t} = -4 \right] + 2 \E\left[ y_1 \left(\sigma, p_{1t}, p_{2t} \right) \mid p_{1t} = 2 \right] &=0, \; \text{and}\\
    3\E\left[ y_1 \left(\sigma, p_{1t}, p_{2t} \right) \mid p_{1t} = 2 \right] - 2 \E\left[ y_1 \left(\sigma, p_{1t}, p_{2t} \right) \mid p_{1t} = 3 \right] &=0.
\end{align*}
Figure \ref{fig:example:conditional} plots their numerical evaluations. Their only common zero is $\sigma_0 =1$, so the conditional moment restriction identifies the true parameters.
\end{example}

This paper proposes a semiparametric approach, built directly on the identification assumption, for estimation and inference in the BLP model. The approach estimates the conditional expectations and conditional variance-covariance matrix nonparametrically, and then chooses the parameters from a weighted minimum distance method. It has two advantages. First, estimation and identification rest on the same condition: it imposes no auxiliary moment restriction distinct from the identifying one, thereby avoiding the failure exposed by Example \ref{eg:counter}. 
Second, it is flexible in the first stage. Any nonparametric estimator satisfying suitable regularity conditions can be used to estimate the conditional expectation function, including machine-learning methods. 
The standard GMM approach is nested as a special case of this nonparametric procedure.
When the conditional expectation is estimated by a linear projection onto a fixed, finite set of instruments, the method reduces to standard GMM, as noted by \cite{ai2003efficient}, while sieve estimation of the conditional expectation corresponds to GMM with a growing set of instruments.

We establish consistency and asymptotic normality for the proposed estimator, which delivers valid inference.
We consider two nonparametric implementations, sieve estimation and kernel estimation, and present simulation evidence that the estimator outperforms the standard GMM approach in finite samples when the identification of the unconditional moments is unclear. As a practical takeaway, when the conditional moment restriction is assumed, we recommend estimation directly from the that restriction for the BLP model, rather than using it to construct optimal instruments and then estimating from the unconditional moments those instruments generate.
}

\paragraph{Related literature.}
This paper connects to several strands of the literature on the identification and estimation of differentiated-products demand.

{Our starting point is the literature on the nonparametric identification of demand from aggregate data. \citet{berry2014identification} use a completeness-type condition on the instruments---of the form $\E[f(s_t,p_t)\mid z_t,x_t]=0$ if and only if $f=0$---to identify demand from a conditional moment restriction. We retain a parametric random-coefficients specification and assume that the conditional restriction identifies the structural parameters. Our concern is whether the fixed, finite collection of unconditional moments used in conventional GMM preserves that identification. We show through a counterexample that it need not, even when the moments are formed with the optimal instrument, and develop an estimator that directly implements the identifying conditional restriction.}

A large literature studies the estimation of BLP-type models by the generalized method of moments and, in particular, the choice of instruments used to form the unconditional moments. \cite{chamberlain1987asymptotic} characterizes the efficient instruments for a conditional moment model, and a subsequent literature develops feasible approximations to them: \cite{berry1999voluntary} and \cite{reynaert2014improving} document the gains from approximating the optimal instruments. \cite{gandhi2019measuring} propose differentiation instruments that can improve the approximating instruments. \cite{conlon2020best,conlon2023incorporating} compare and implement optimal instrument sets in practice, and \cite{donald2009choosing} study how many approximating instruments to include. Our contribution is complementary but distinct: this literature seeks better \emph{finite} instrument sets within the unconditional-moment GMM framework, whereas we show that the selected unconditional moments need not preserve identification by the conditional moment restriction, even when the selected instruments are optimal in the sense of efficiency.

Methodologically, we build on the literature on sieve estimation of conditional moment restrictions, in particular the sieve minimum-distance framework of \cite{ai2003efficient}. Although that framework is not developed for the BLP model, it applies directly once demand is written as a conditional moment restriction in $\theta$, and it nests standard GMM as the special case in which the conditioning information is projected onto a fixed, finite set of instruments. Our estimator can be viewed as an instance of this framework tailored to the BLP inverse, with the number of approximating terms allowed to grow with the sample.

{This paper's approach is related to but differs from the literature using the integrated conditional moment method; see, e.g., \citet{bierens1982consistent,lavergne2013smooth,antoine2014conditional,escanciano2018simple}. These papers replace the conditional moment restriction with an equivalent continuum of unconditional moment restrictions, and under suitable conditions the same construction could conceptually be applied to the BLP model. Yet to our knowledge it has not been discussed in the previous literature, and is beyond the scope of this paper. Establishing the corresponding theory for integrated conditional-moment estimators in the BLP setting and comparing their statistical and computational properties with our approach are directions for future research.}

Our approach is distinct from the recent literature that estimates BLP-type demand allowing a nonparametric distribution of consumer heterogeneity. \cite{wang2023sieve} proposes a sieve BLP estimator that nonparametrically recovers the distribution of random coefficients, and \cite{fox2012random} and \cite{fox2016nonparametric} establish nonparametric identification of that distribution under support and regularity conditions that differ from the standard BLP environment. \cite{lu2023semi} develop a semiparametric estimator and recover the unknown distribution of random coefficients but are under the large-$J$ asymptotic framework. We instead retain the parametric random-coefficient specification and target $(\sigma,\alpha)$ directly with $J$ fixed; our object of interest is the estimation--identification gap rather than the flexibility of the heterogeneity distribution. Finally, our contribution differs from work that accelerates the computation of the standard BLP estimator, such as the MPEC reformulation of \cite{dube2012improving} and the accelerated iteration of \cite{lee2015computationally}: those papers compute a given estimator more efficiently, whereas we study a different population criterion and estimation strategy. \cite{salanie2022fast} construct tractable approximations based on low-order Taylor expansions of the demand inverse. Our approach instead retains the original demand inverse and estimates the relevant conditional expectations nonparametrically.

\paragraph{Organization.} 
The remainder of this paper is structured as follows. 
Section 2 discusses the framework and estimation. 
Section 3 presents the general theory for estimation and inference.
Section 4 summarizes the Monte Carlo simulation results, and Section 5 concludes.
All proofs and lemmas are collected in the Appendix.

\paragraph{Notations.} 
$\normm{A}:=\sqrt{tr(A'A)}$.
$\E_T:= \frac{1}{T}\sum_{t=1}^T$ and $\E_J:= \frac{1}{J}\sum_{j=1}^J$.
$a_n\asymp b_n$ means $a_n=O(b_n)$ and $b_n=O(a_n)$.
$a_n\lesssim b_n$ means $a_n=O(b_n)$.
$\normm{a} = (a'a)^{1/2}$ as the Euclidean norm of a vector $a$.
$\normm{a}_\infty=\max_i |a_i|$ as the sup-norm of a vector $a$.
$\normm{A}_{op}=\sup_{\normm{b}=1}\normm{Ab}=\max\{|\lambda_{\max}(A)|,|\lambda_{\max}(-A)|\}$ as the operator norm of a matrix $A$.
$\normm{f(z)}_{L^2}=\E[f(z)^2]^{1/2}$ as the $L^2$ norm of a function $f(z)$.

\section{Model and Estimation}
\subsection{Model}\label{sec:model}
We adopt the random-coefficients discrete-choice model of \cite{berry1995automobile}.
Consider a sequence of markets indexed by $t=1,\dots,T$.\footnote{The definition of a market is application-specific: with cross-sectional data markets are typically geographic areas, with panel data they are time periods, and in many cases a combination of the two.}
In each market a unit mass of consumers chooses among $J$ inside products, $j=1,\dots,J$, or an outside option $j=0$.
Product $j$ in market $t$ is summarized by the triple $(x_{jt}',p_{jt},\xi_{jt})$: a vector of $d_x$ exogenous characteristics $x_{jt}\in\R^{d_x}$, a scalar price $p_{jt}\in\R$, and a scalar characteristic $\xi_{jt}\in\R$ that is observed by consumers and firms but not by the econometrician.
Because prices are determined in equilibrium, they are correlated with $\xi_{jt}$ and are treated as endogenous.
Identification therefore relies on a vector of instruments $z_{jt}\in\R^{d_z}$ satisfying the conditional moment restriction
\begin{equation}\label{eq:exog}
    \E[\xi_{jt}\mid z_{jt}]=0,
    \qquad
    \sup_{1\le j\le J}\E[\xi_{jt}^{2}\mid z_{jt}]<C<\infty,
\end{equation}
where the exogenous characteristics $x_{jt}$ are themselves included among the instruments $z_{jt}$.

Given this market structure, we turn to the specification of preferences.
The indirect utility that consumer $i$ derives from product $j$ in market $t$ is
\begin{equation}\label{eq:utility}
    u_{ijt}=\underbrace{x_{jt}'\beta_0+p_{jt}\alpha_0+\xi_{jt}}_{\textstyle y_{jt}}+x_{jt}^{(1)\prime}\beta_i+\epsilon_{ijt},
\end{equation}
and the utility of the outside option is normalized to $u_{i0t}=\epsilon_{i0t}$.
The idiosyncratic terms $\epsilon_{ijt}$ are independent and identically distributed across $i$, $j$, and $t$, following the Type-I extreme-value distribution.
The taste for the characteristic $x_{jt}^{(1)}$ is consumer-specific and $\beta_i$ is a random variable drawn from a distribution $F(\,\cdot\,;\sigma_0)$ known up to a finite-dimensional parameter $\sigma_0\in\R^{d_\sigma}$.
The bracketed term $y_{jt}=x_{jt}'\beta_0+p_{jt}\alpha_0+\xi_{jt}$ is the \emph{mean utility}: the linear term of product $j$ in market $t$.

We collect the parameters in $\theta_0=(\sigma_0,\alpha_0,\beta_0')'\in\Theta$, where $\Theta$ is compact and $\theta_0$ lies in its interior.
Throughout, $\beta_i$, $\epsilon_{ijt}$, and $(\xi_{jt},x_{jt},p_{jt})$ are mutually independent, and $(\beta_i,\epsilon_{ijt})$ are i.i.d.\ across consumers.

Aggregating these preferences over consumers delivers the market shares predicted by the model.
Consumer $i$ purchases the product that maximizes her utility, so her choice of product $j$ is described by
\begin{equation*}
    d_{ijt}=
    \begin{cases}
        1 & \text{if } u_{ijt}>u_{ikt}\text{ for all }k\neq j,\\
        0 & \text{otherwise.}
    \end{cases}
\end{equation*}
Integrating the individual choices over the extreme-value shocks and the distribution of $\beta_i$ yields the market share of product $j$,
\begin{align}
    s_{jt}
    &=f_{js}(x^{(1)}_t,y_t,\sigma_0)
    =\int d_{ijt}\,\mathrm{d}F(\epsilon_{i0t},\dots,\epsilon_{iJt})\,\mathrm{d}F(\beta_i;\sigma_0)\notag\\
    &=\int\frac{\exp(y_{jt}+x_{jt}^{(1)\prime}\beta_i)}{1+\sum_{k=1}^{J}\exp(y_{kt}+x_{kt}^{(1)\prime}\beta_i)}\,\mathrm{d}F(\beta_i;\sigma_0),
    \label{eq:share}
\end{align}
where the second line follows from integrating the Type-I extreme-value shocks analytically to obtain the familiar logit kernel and then averaging over the random coefficient $\beta_i$.
Stacking the products, the vector of shares in market $t$ is
\begin{equation}\label{eq:share-vec}
    s_t=f_s(x^{(1)}_t,y_t,\sigma_0)
    :=\bigl(f_{1s}(x^{(1)}_t,y_t,\sigma_0),\dots,f_{Js}(x^{(1)}_t,y_t,\sigma_0)\bigr)',
\end{equation}
with market-level vectors $y_t=(y_{1t},\dots,y_{Jt})'$, $x^{(1)}_t=(x_{1t}^{(1)},\dots,x_{Jt}^{(1)})'$, and $\xi_t=(\xi_{1t},\dots,\xi_{Jt})'$ defined analogously.
Equation~\eqref{eq:share} is the random-coefficients logit share, and \eqref{eq:share-vec} defines the map from parameters and characteristics to shares.

The share system just derived can be inverted to recover the mean utilities, which is the property underlying our estimator.
\cite{berry1994estimating} and \cite{berry1995automobile} establish that this map can be inverted in the mean utilities.
For any fixed $\sigma$, any characteristics vector $x^{(1)}_t$, and any share vector $s_t$ in the interior of the unit simplex, there exists a unique $y\in\R^{J}$ solving $s_t=f_s(x^{(1)}_t,y,\sigma)$; that is, $f_s(x^{(1)}_t,\cdot\,,\sigma)$ is invertible.\footnote{Equivalently, \cite{berry1994estimating} shows that for given $x^{(1)}_t$, $s_t$, $x_t=(x_{1t}',\dots,x_{Jt}')'$, and $(\beta,\alpha,\sigma)$ there is a unique $\xi\in\R^{J}$ solving $s_t=f_s(x^{(1)}_t,x_t\beta+p_t\alpha+\xi,\sigma)$. Uniqueness of $y$ is the special case obtained by setting $\beta=0$ and $\alpha=0$.}
Denote the inverse by $f_s^{-1}(s_t,x^{(1)}_t,\sigma)\in\R^{J}$, and define the imputed mean utilities and demand shocks
\begin{align}
    y(\sigma,w_t)&:=f_s^{-1}(s_t,x^{(1)}_t,\sigma),\label{eq:inv}\\
    \xi(\theta,w_t)&:=y(\sigma,w_t)-x_t'\beta-p_t\alpha,\label{eq:xi}
\end{align}
where $w_t$ denotes all the data in market $t$, $\theta=(\sigma,\alpha,\beta')'$ and $x_t'\beta$ denotes the vector with $j$-th element $x_{jt}'\beta$.
Let $y_{j}(\sigma,w_t)$ and $\xi_{j}(\theta,w_t)$ be the $j$-th elements of $y(\sigma,w_t)$ and $\xi(\theta,w_t)$.
By construction these imputations recover the truth at $\theta_0$: $y_{j}(\sigma_0,w_t)=y_{j t}$ and $\xi_{j}(\theta_0,w_t)=\xi_{j t}$.

The identification assumption (the conditional moment restriction) is
\begin{align}
    \mu_{y,j}(z_{jt};\sigma)-\alpha\,\mu_{p,j}(z_{jt})-x_{jt}'\beta=0 \quad &\text{for all } j\; a.s.\text{ iff}\quad \theta=\theta_0 \label{eq:id1}\\
        \qquad\Longleftrightarrow\qquad& \nonumber\\
    \E_J\E\Bigl[\mu_{y,j}(z_{jt};\sigma)-\alpha\,\mu_{p,j}(z_{jt})-x_{jt}'\beta\Bigr]^2=0 \quad &\text{iff}\quad \theta=\theta_0.\label{eq:id2}
\end{align}
Restriction~\eqref{eq:id2} is the population relationship on which our estimator is built.

\subsection{Estimation}

The estimation proceeds in two steps.
\begin{enumerate}
    \item[Step 1:] For each candidate $\sigma$, estimate the conditional expectations $\mu_{y,j}(z_{jt};\sigma)=\E[y_{j}(\sigma,w_t)\mid z_{jt}]$ and $\mu_{p,j}(z_{jt})=\E[p_{jt}\mid z_{jt}]$ nonparametrically, yielding $\hat\mu_{y,j}(z_{jt};\sigma)$ and $\hat\mu_{p,j}(z_{jt})$.
    \item[Step 2:] Choose $\theta=(\sigma,\alpha,\beta)$ to minimize the sample criterion
    \begin{equation}\label{eq:est}
        \hat{\theta}=\arg\min_{\theta\in\Theta}\sum_{j,t}\big(\hat\mu_{y,j}(z_{jt};\sigma)-\alpha\,\hat\mu_{p,j}(z_{jt})-x_{jt}'\beta\big)^2\,\hat\Lambda_j(z_{jt}),
    \end{equation}
    where $\hat\Lambda_j$ is an estimated weight, constructed below from a nonparametric estimate of the conditional variance $\Sigma_{j}(z_{jt})=\E[\xi_{j}^2(\theta_0,w_t)\mid z_{jt}]$ and a trimming factor that discards instruments lying close to the boundary of their support.
\end{enumerate}
The two steps mirror the population identifying restriction: Step~1 replaces the unknown conditional expectations $\mu_{y,j}$ and $\mu_{p,j}$ with nonparametric estimates, and Step~2 chooses $\theta$ so that the resulting residual $\hat\mu_{y,j}(z_{jt};\sigma)-\alpha\,\hat\mu_{p,j}(z_{jt})-x_{jt}'\beta$---the sample analogue of $\E[\xi_{j}(\theta,w_t)\mid z_{jt}]$---is as close to zero as possible.
The population weight targeted by $\hat\Lambda_j$ is $\Lambda_j(z_{jt})=\Sigma_j(z_{jt})^{-1}$, the conditional-variance weighting familiar from generalized least squares: it downweights observations with noisier unobserved product characteristics and, as we show in Section~\ref{sec:theory}, delivers the efficient estimator within this class. The theory there in fact accommodates any weight bounded away from zero and infinity; the unweighted criterion used in the preliminary step below is the case $\Lambda_j\equiv1$.
We now describe each step in detail.

\paragraph*{Step 1.}
The conditional expectations can be estimated by any nonparametric smoother; we consider two standard choices, kernel and series estimation.
A useful feature common to both is that the smoothing weights depend only on the instruments $z_{jt}$ and not on $\sigma$, so that only the regressand $y_{j}(\sigma,w_t)$ changes as $\sigma$ varies.
Re-estimating $\hat\mu_{y,j}(z_{jt};\sigma)$ across values of $\sigma$ therefore requires no re-computation of the weights, which makes the outer optimization in \eqref{eq:est} computationally light.

\emph{Kernel estimation.}
Using a kernel $K_h(\cdot)$ with bandwidth $h$, the local-constant (Nadaraya--Watson) estimators are
\begin{subequations}\label{est:kernel}
\begin{align}
    \hat\mu_{y,j}(z_{jt};\sigma)&=\frac{\sum_{l=1}^{T}K_h(z_{jl}-z_{jt})\,y_{j}(\sigma,w_{l})}{\sum_{l=1}^{T}K_h(z_{jl}-z_{jt})},\\[4pt]
    \hat\mu_{p,j}(z_{jt})&=\frac{\sum_{l=1}^{T}K_h(z_{jl}-z_{jt})\,p_{jl}}{\sum_{l=1}^{T}K_h(z_{jl}-z_{jt})},
\end{align}
\end{subequations}

where the sums run over markets $l=1,\dots,T$ and the estimators are computed product by product.
If the conditional distributions of $y_{j}(\sigma,w_t)$ and $p_{jt}$ given $z_{jt}$ do not vary across products---for instance when products are exchangeable---the regression functions $\mu_{y,j}(\cdot\,;\sigma)$ and $\mu_{p,j}(\cdot)$ are common across $j$, and pooling over products yields a more efficient estimator by enlarging the effective sample:
\begin{align*}
    \hat\mu_{y,j}(z_{jt};\sigma)&=\frac{\sum_{l=1}^{T}\sum_{j'=1}^{J}K_h(z_{j'l}-z_{jt})\,y_{j'}(\sigma,w_l)}{\sum_{l=1}^{T}\sum_{j'=1}^{J}K_h(z_{j'l}-z_{jt})},\\[4pt]
    \hat\mu_{p,j}(z_{jt})&=\frac{\sum_{l=1}^{T}\sum_{j'=1}^{J}K_h(z_{j'l}-z_{jt})\,p_{j'l}}{\sum_{l=1}^{T}\sum_{j'=1}^{J}K_h(z_{j'l}-z_{jt})}.
\end{align*}

\emph{Series estimation.}
Let $\{\phi_k(z)\}_{k=1}^{K}$ be a set of basis functions and collect them in $\phi^{K}(z_{jt})=(\phi_1(z_{jt}),\dots,\phi_{K}(z_{jt}))'$, with the number of terms $K=K_n$ growing with the sample size.
The series estimators are the fitted values of the least-squares projections onto these bases,
\begin{subequations}\label{est:series}
\begin{align}
    \hat\mu_{y,j}(z_{jt};\sigma)&=\phi^{K}(z_{jt})'\hat{\Pi}_{y,j}(\sigma),\\
    \hat\mu_{p,j}(z_{jt})&=\phi^{K}(z_{jt})'\hat{\Pi}_{p,j},
\end{align}
\end{subequations}

where $\hat{\Pi}_{y,j}(\sigma)$ and $\hat{\Pi}_{p,j}$ are the coefficient vectors from regressing $y_{j}(\sigma,w_t)$ and $p_{jt}$, respectively, on $\phi^{K}(z_{jt})$.
Because the projection is linear in the regressand, the estimated coefficients are
\begin{align*}
    \hat{\Pi}_{y,j}(\sigma)&=\big(\E_T\phi^{K}(z_{jt})\phi^{K}(z_{jt})'\big)^{-1}\E_T\phi^{K}(z_{jt})\,y_{j}(\sigma,w_t),\\
    \hat{\Pi}_{p,j}&=\big(\E_T\phi^{K}(z_{jt})\phi^{K}(z_{jt})'\big)^{-1}\E_T\phi^{K}(z_{jt})\,p_{jt}.
\end{align*}
So the design matrix is inverted once and only the second factor is updated as $\sigma$ changes---the same computational saving noted for the kernel estimator.
As with the kernel case, if the regression functions are common across products---so that $\mu_y(\cdot\,;\sigma)$ and $\mu_p(\cdot)$ do not depend on $j$---the coefficients may be estimated by stacking all products in a single least-squares regression:
\begin{align*}
    \hat{\Pi}_{y,j}(\sigma)&=\Big(\E_J\E_T\phi^{K}(z_{jt})\phi^{K}(z_{jt})'\Big)^{-1}\E_J\E_T\phi^{K}(z_{jt})\,y_{j}(\sigma,w_t),\\[4pt]
    \hat{\Pi}_{p,j}&=\Big(\E_J\E_T\phi^{K}(z_{jt})\phi^{K}(z_{jt})'\Big)^{-1}\E_J\E_T\phi^{K}(z_{jt})\,p_{jt},
\end{align*}
and the fitted values $\hat\mu_{y,j}(z_{jt};\sigma)=\phi^{K}(z_{jt})'\hat{\Pi}_{y,j}(\sigma)$ and $\hat\mu_{p,j}(z_{jt})=\phi^{K}(z_{jt})'\hat{\Pi}_{p,j}$ are formed as before.
By contrast, the product-specific estimator restricts each sum to a single $j$, replacing $\sum_{j=1}^{J}\sum_{t=1}^{T}$ with $\sum_{t=1}^{T}$ and yielding a separate coefficient vector $\hat{\Pi}_{y,j}(\sigma)$ for each product.

\paragraph*{Step 2.}
Given the first-step estimates, the criterion in \eqref{eq:est} is minimized over $\theta=(\sigma,\alpha,\beta)$.
The residual is linear in the finite-dimensional parameters $(\alpha,\beta)$ and enters $\sigma$ only through $\hat\mu_{y,j}(\cdot\,;\sigma)$.
Hence, for any fixed $\sigma$, the inner minimization over $(\alpha,\beta)$ is a weighted least-squares problem with the closed-form solution
\begin{equation*}
    \big(\hat\alpha(\sigma),\hat\beta(\sigma)\big)
    =\Big(\textstyle\sum_{j,t}\tilde{x}_{jt}\tilde{x}_{jt}'\,\hat\Lambda_j(z_{jt})\Big)^{-1}
    \sum_{j,t}\tilde{x}_{jt}\,\hat\mu_{y,j}(z_{jt};\sigma)\,\hat\Lambda_j(z_{jt}),
    \qquad
    \tilde{x}_{jt}=\big(\hat\mu_{p,j}(z_{jt}),\,x_{jt}'\big)',
\end{equation*}
so the optimization reduces to a one-dimensional search over $\sigma$ (or $d_\sigma$-dimensional when $\sigma$ is a vector), with $(\alpha,\beta)$ concentrated out.
This profiling substantially lightens the numerical burden relative to a joint search over all of $\theta$.

\paragraph*{The estimated weight.}
The weight factors into a variance estimate and a trimming indicator,
\begin{equation}\label{eq:weight-factor}
    \hat\Lambda_j(z)=\hat\Sigma_j(z)^{-1}\,\tau_T(z),
\end{equation}
and the two factors do different jobs. The first is obtained in a preliminary estimation: one solves \eqref{eq:est} with $\hat\Lambda_j(z_{jt})= \tau_T(z_{jt})$ (unweighted), forms the residuals $\hat\xi_{jt}=y_j(\tilde\sigma,w_t)-\tilde\alpha \ p_{jt}-x_{jt}'\tilde\beta$, where $\tilde\sigma,\ \tilde\alpha,\ \tilde\beta$ are the preliminary estimates, and estimates $\Sigma_j$ by smoothing $\hat\xi_{jt}^2$ on $z_{jt}$ with the same nonparametric method used in Step~1, denoted by $\hat\Sigma_j(z)$. 
The criterion \eqref{eq:est} is then re-minimized with the estimated weights. The second factor, $\tau_T$, takes the value zero for instruments close to the boundary of their support and one elsewhere for kernel estimation, while for series estimation it takes the value one for all $z_{jt}$. It is there because $\hat\Sigma_j$ and the first-step estimates are nonparametric fits: near the boundary they are computed from few effective observations, and dropping those observations is what allows the uniform control required for inference to be imposed on the interior region alone, which is crucial in the kernel estimation. The discarded region shrinks with the sample, so the trimming does not affect the limiting distribution. Both factors depend on the smoother, and we take the two in turn. Throughout, $\mathcal Z$ denotes the support of $z_{jt}$ and
\[
    d(z,\partial\mathcal Z)=\inf_{z'\in\partial\mathcal Z}\|z-z'\|_\infty
\]
is the sup-norm distance from $z$ to the boundary of $\mathcal Z$, which for $z$ in the interior of $\mathcal Z$ is the distance to the boundary $\partial\mathcal Z$.
 
\emph{Kernel.} The conditional variance is the Nadaraya--Watson fit of the squared preliminary residuals,
\[
    \hat\Sigma_j(z)=\frac{\sum_{l=1}^{T}K_h(z_{jl}-z)\,\hat\xi_{jl}^{\,2}}{\sum_{l=1}^{T}K_h(z_{jl}-z)} .
\]
Let $R<\infty$ be the radius of the support of the kernel, so that $K(u)=0$ whenever $\|u\|_\infty>R$ and $K_h(\cdot-z)$ is supported on the window $\{z':\|z'-z\|_\infty\le Rh\}$. The trimming is then
\[
    \tau_T(z)=\mathbf{1}\{d(z,\partial\mathcal Z)\ge Rh\},
\]
which retains exactly those points whose smoothing window lies inside $\mathcal Z$. Since the discarded collar has width $Rh\to0$, the fraction of observations it removes vanishes.
 
\emph{Series.} The conditional variance is the least-squares projection of the squared preliminary residuals on the same dictionary,
\[
    \hat\Sigma_j(z)=\phi^{K}(z)'\hat\Pi_{\xi^2,j},
    \qquad
    \hat\Pi_{\xi^2,j}=\big(\E_T\phi^{K}(z_{jt})\phi^{K}(z_{jt})'\big)^{-1}\E_T\phi^{K}(z_{jt})\,\hat\xi_{jt}^{\,2},
\]
formed with the same design matrix as in Step~1, so no additional inversion is required. Here no trimming is needed and one sets $\tau_T\equiv1$, so that $\hat\Lambda_j=\hat\Sigma_j^{-1}$: a series fit is a global projection rather than a local average, so it has no boundary window to truncate.
 
Finally, the smoothing parameters---the bandwidth $h$ for the kernel estimator and the number of basis terms $K$ for the series estimator---govern the bias--variance trade-off of the first step; the rate conditions they must satisfy for valid inference on $\theta$ are given in Section~\ref{sec:theory}.

\section{Asymptotic Theory}\label{sec:theory}

\subsection{General Theory}\label{sec:theory-general}

We first fix notation common to both results.  Consider a compact parameter
space $\Theta$ and support set of $z_{jt}$ as $\mathcal Z$.  We maintain the condition that the linear characteristics
are their own projection onto the instruments,
$\mathbb E[x_{jt}\mid z_{jt}]=x_{jt}$.  Under this condition the identifying
restriction takes the form
\begin{equation}\label{eq:theory-id}
  \mu_{y,j}(z_{jt};\sigma_0)=\alpha_0\,\mu_{p,j}(z_{jt})+x_{jt}'\beta_0.
\end{equation}
 
The asymptotic analysis requires the derivatives of $\mu_{y,j}$ in the
nonlinear parameter.  We denote the first and second derivatives by
\[
  \dot\mu_{y,j}(z;\sigma)=\tfrac{\partial}{\partial\sigma}\mu_{y,j}(z;\sigma)
  =\mathbb E\Bigl[\tfrac{\partial}{\partial\sigma}y_j(\sigma,w)\Bigm| z\Bigr],
  \qquad
  \ddot\mu_{y,j}(z;\sigma)=\tfrac{\partial^2}{\partial\sigma^2}\mu_{y,j}(z;\sigma)
  =\mathbb E\Bigl[\tfrac{\partial^2}{\partial\sigma^2}y_j(\sigma,w)\Bigm| z\Bigr].
\]
The corresponding nonparametric estimators of $\mu_{y,j}(z;\sigma)$,
$\dot\mu_{y,j}(z;\sigma)$ and $\ddot\mu_{y,j}(z;\sigma)$ are written
$\hat\mu_{y,j}(z;\sigma)$, $\hat{\dot\mu}_{y,j}(z;\sigma)$ and
$\hat{\ddot\mu}_{y,j}(z;\sigma)$, respectively.  It is convenient to
abbreviate the direction,
\[
  \Gamma_{jt}\;:=\;\bigl(\dot\mu_{y,j}(z_{jt};\sigma_0),\,-\mu_{p,j}(z_{jt}),\,-x_{jt}\bigr)'
  \;=\;\nabla_\theta\, r_{jt}(\theta_0),
\]
with $r_{jt}(\theta)$ the residual term, defined as
$r_{jt}(\theta)=\mu_{y,j}(z_{jt};\sigma)-\alpha\,\mu_{p,j}(z_{jt})-x_{jt}'\beta$ for the population residual; under \eqref{eq:theory-id}, $r_{jt}(\theta_0)=0$ for every $(j,t)$.

\paragraph{Weighting and trimming.} The criterion (\ref{eq:est}) is built from two
further objects.  The first is a population weight
$\omega_j:\mathcal Z\to(0,\infty)$, $j=1,\dots,J$, together with an estimator
$\hat\omega_j$ of it.  We do \emph{not} require $\omega_j$ to be the
conditional variance $\Sigma_j$: any weight bounded away from zero and
infinity is admissible, $\omega_j\equiv1$ giving the unweighted estimator and
$\omega_j=\Sigma_j$ the efficient one (under suitable conditions).  The second is a trimming indicator.
Letting $\varepsilon_T\to0$ be a nonrandom sequence, put
\[
  \mathcal Z_T\;:=\;\bigl\{z\in\mathcal Z:\ d(z,\partial\mathcal Z)\ \ge\ \varepsilon_T\bigr\},
  \qquad
  \tau_T(z)\;:=\;\mathbf 1\bigl\{d(z,\partial\mathcal Z)\ \ge\ \varepsilon_T\bigr\},
\]
so that $\tau_T(z)=0$ for $z$ within $\varepsilon_T$ of the boundary and
$\tau_T(z)=1$ for $z$ in the interior region $\mathcal Z_T$, and write
\[
  p_T\;:=\;\max_{1\le j\le J}\Pr\bigl(\tau_T(z_{jt})=0\bigr)
\]
for the probability mass discarded.  The leading case, used for the kernel
implementation in Section \ref{sec:theory-kernel}, is $\varepsilon_T=Rh$, where $h$ is the
bandwidth and $R$ is the radius of the support of the kernel, so that for
every retained $z$ the smoothing window $\{z':\|z'-z\|_{\infty}\le Rh\}$ lies entirely
inside $\mathcal Z$; the series implementation in Section \ref{sec:theory-series} needs no
trimming and takes $\varepsilon_T=0$, i.e.\ $\tau_T\equiv1$ and $p_T=0$.
Finally, collect the weights as they enter the criterion,
\[
  \Lambda_j(z)\;:=\;\omega_j(z)^{-1},
  \qquad
  \hat\Lambda_j(z)\;:=\;\hat\omega_j(z)^{-1}\,\tau_T(z),
  \qquad
  a_{jt}\;:=\;\Lambda_j(z_{jt})\,\Gamma_{jt},
\]
so that $\Lambda_j$ is the infeasible target, $\hat\Lambda_j$ the feasible
trimmed weight, and $\hat\Lambda_j(z)=0$ whenever $z\notin\mathcal Z_T$.
 
It is convenient to treat these conditional-mean functions as an
infinite-dimensional nuisance parameter, which we collect in
$\eta=(\eta_p,\eta_y,\dot\eta_y,\eta_\Lambda)$, with true value
$\eta_0=(\mu_p,\mu_y,\dot\mu_y,\Lambda)$, where $\mu_p=(\mu_{p,1},\dots,\mu_{p,J})$,
$\mu_y=(\mu_{y,1},\dots,\mu_{y,J})$, $\dot\mu_y=(\dot\mu_{y,1},\dots,\dot\mu_{y,J})$
and $\Lambda=(\Lambda_1,\dots,\Lambda_J)$; the last coordinate is the weight
as it enters the criterion, so the trimming is carried inside the nuisance
parameter.  Define the moment function
\begin{equation}\label{eq:moment-fn}
  g(w_t,\theta,\eta)=\mathbb E_J\left[
    \bigl(\eta_{y,j}(z_{jt};\sigma)-\alpha\,\eta_{p,j}(z_{jt})-x_{jt}'\beta\bigr)\,
    \eta_{\Lambda,j}(z_{jt})
    \begin{pmatrix}\dot\eta_{y,j}(z_{jt};\sigma)\\[2pt] -\eta_{p,j}(z_{jt})\\[2pt] -x_{jt}\end{pmatrix}
  \right],
\end{equation}
The estimator in (\ref{eq:est}) is then a solution to the sample first-order
condition
\[
  \mathbb E_T\,g(w_t,\hat\theta,\hat\eta)=0,
  \qquad
  \hat\eta=(\hat\mu_p,\hat\mu_y,\hat{\dot\mu}_y,\hat\Lambda).
\]

\begin{condition}\label{cond:regularity}
\begin{enumerate}
\item\label{cond:independence} (Independence) The markets $t=1,\dots,T$ are independent and
identically distributed draws of
$w_t=(z_{jt},x_{jt},p_{jt},s_{jt},\xi_{jt})_{j=1}^J$, with $z_{jt}\supset x_{jt}$.
\item\label{cond:cons-reg} (Regressors and parameter space) $\|\xi_{jt}\|_\infty$,
$\|x_{jt}\|_\infty$, $\|z_{jt}\|_\infty\le C$ and $|p_{jt}|\le C$ almost
surely; $\Theta$ is compact and $\theta_0$ lies in its interior.
\item\label{cond:cons-weight} (Weighting) There exist constants $0<c\le C<\infty$ such that
$c\le\Lambda_j(z)=\omega_j(z)^{-1}\le C$ for all $z$ and $j$.
\item\label{cond:cons-id} (Identification) $\theta_0$ uniquely minimizes the population criterion
$Q(\theta)=\mathbb E_J\mathbb E\bigl[r_{jt}(\theta)^2\,\Lambda_j(z_{jt})\bigr]$.
\item\label{cond:cons-bound} (Smoothness) $\mu_{y,j},\mu_{p,j}$ are bounded uniformly and continuous
in $(\sigma,z)$ and $z$ respectively for each $j$.
\end{enumerate}
\end{condition}

\begin{condition}\label{cond:consistency}
\begin{enumerate}
\item\label{cond:cons-rate} (First-stage consistency) Uniformly over $\sigma\in\Theta_\sigma$,
\[
  \sup_{\sigma\in\Theta_\sigma}\E_J\E_T\Big[\big|\hat\mu_{y,j}(z_{jt};\sigma)-\mu_{y,j}(z_{jt};\sigma)\big|^2\tau_T(z_{jt})\Big]
  +\E_J\E_T\Big[\big|\hat\mu_{p,j}(z_{jt})-\mu_{p,j}(z_{jt})\big|^2\tau_T(z_{jt})\Big]=o_p(1).
\]
\item\label{cond:cons-uniform} (Uniform convergence) Uniformly over $z\in\mathcal Z_T$ for each $j$:
\[
  \sup_{z\in\mathcal Z_T}\bigl|\hat\omega_j(z)^{-1}-\omega_j(z)^{-1}\bigr|=o_p(1),
  \qquad\text{equivalently}\qquad
  \sup_{z\in\mathcal Z}\bigl|\hat\Lambda_j(z)-\Lambda_j(z)\tau_T(z)\bigr|=o_p(1).
\]
\item\label{cond:cons-trim} (Trimming) The sequence $\varepsilon_T$ is nonrandom with
$\varepsilon_T\to0$, and $p_T=\max_j\Pr(\tau_T(z_{jt})=0)\to0$.
\end{enumerate}
\end{condition}

\begin{remark}
    Conditions~\ref{cond:regularity} and \ref{cond:consistency} are the primitives behind uniform convergence of the sample criterion to its population analogue, which is what drives consistency of $\hat\theta$.
 
Condition~\ref{cond:regularity} concerns the data and the population problem. Part~\eqref{cond:independence} treats markets as i.i.d.\ draws; dependence across products \emph{within} a market is unrestricted, since all population objects average over $j$. 
Parts~\eqref{cond:cons-reg} and \eqref{cond:cons-weight} are mild boundedness conditions: bounded product characteristics, unobserved characteristics, and instruments on a compact parameter space, and a weight bounded away from zero and infinity so that $\Lambda_j$ neither degenerates nor explodes. Part~\eqref{cond:cons-weight} does not tie the weight to the conditional variance. 
Part~\eqref{cond:cons-id} is the identification assumption \eqref{eq:id1}, restated as a unique minimizer. Since $\Lambda_j$ is bounded above and away from zero by Part~\eqref{cond:cons-weight}, $Q(\theta)=0$ if and only if $r_{jt}(\theta)=0$ almost surely for every $j$. Part~\eqref{cond:cons-bound} keeps the conditional means bounded and continuous.
 
Condition~\ref{cond:consistency} concerns the first stage. Part~\eqref{cond:cons-rate} asks only that the estimated conditional means converge in mean square, uniformly over the nonlinear parameter $\sigma$; this is a high-level requirement, verified for the kernel and series estimators in Sections~\ref{sec:theory-kernel} and \ref{sec:theory-series}. Part~\eqref{cond:cons-uniform} asks the same of the estimated weight, in sup norm, but only on the retained region $\mathcal Z_T$: this is what the trimming buys, since it is near $\partial\mathcal Z$ that a nonparametric weight is estimated from fewest effective observations and is least reliable. 
Part~\eqref{cond:cons-trim} requires only that the discarded region be asymptotically negligible; no rate is imposed. It is implied by the geometry rather than assumed on top of it: if $\mathcal Z$ is compact with a Lipschitz boundary (``sufficiently regular''), and $z_{jt}$ has a density function $f_j$ with $\sup_z f_j(z)\le\bar f<\infty$, then $Leb\{z:d(z,\partial\mathcal Z)<\varepsilon\}\le C\varepsilon$ and hence $p_T=O(\varepsilon_T)$, so Part~\eqref{cond:cons-trim} holds for any $\varepsilon_T\to0$, and in particular for the kernel choice $\varepsilon_T=Rh$; the series implementation takes $\tau_T\equiv1$ and $p_T=0$. Together, Parts~\eqref{cond:cons-uniform} and \eqref{cond:cons-trim} and Condition~\ref{cond:regularity}\eqref{cond:cons-weight} give $0\le\hat\Lambda_j(z)\lesssim C$ uniformly on $\mathcal Z$ with probability approaching one, the boundedness used throughout the proofs.
\end{remark}

\begin{theorem}\label{thm:consistency}
Under Condition~\ref{cond:regularity} and \ref{cond:consistency}, $\hat\theta\xrightarrow{p}\theta_0$.
\end{theorem}

\begin{condition}\label{cond:y-smooth} 
$\mu_{y,j}(z;\sigma)$ is bounded and three times continuously differentiable
in $\sigma$ on $\Theta_\sigma$ for each $z$, i.e., $\sup_{z\in\mathcal{Z}}|\partial^l_{\sigma} \mu_{y,j}(z;\sigma)|\le C$ for $l=0,1,2,3$ and some constant $C>0$.
\end{condition}

\begin{condition}\label{cond:normality}
\begin{enumerate}
\item\label{cond:deriv} The derivative estimators are the derivatives of the first-stage
estimator, $\hat{\dot\mu}_{y,j}(z;\sigma)=\partial_\sigma\hat\mu_{y,j}(z;\sigma)$
and $\hat{\ddot\mu}_{y,j}(z;\sigma)=\partial^2_\sigma\hat\mu_{y,j}(z;\sigma)$.
\item\label{cond:norm-deriv} (Derivative estimators consistency)
\[
  \sup_{\sigma\in\Theta_\sigma}\mathbb E_J\mathbb E_T\Big[\bigl|\hat{\dot\mu}_{y,j}(z_{jt};\sigma)-\dot\mu_{y,j}(z_{jt};\sigma)\bigr|^2\tau_T(z_{jt})\Big]
  +\sup_{\sigma\in\Theta_\sigma}\mathbb E_J\mathbb E_T\Big[\bigl|\hat{\ddot\mu}_{y,j}(z_{jt};\sigma)-\ddot\mu_{y,j}(z_{jt};\sigma)\bigr|^2\tau_T(z_{jt})\Big]=o_p(1).
\]
\item\label{cond:norm-msr} (Mean-square rate) The first-stage estimators converge in mean square
faster than $T^{-1/4}$:
\begin{align*}
    &\E_J\E_T\Big[\bigl|\hat\mu_{y,j}(z_{jt};\sigma_0)-\mu_{y,j}(z_{jt};\sigma_0)\bigr|^2\tau_T(z_{jt})\Big]
    +\E_J\E_T\Big[\bigl|\hat{\dot\mu}_{y,j}(z_{jt};\sigma_0)-\dot\mu_{y,j}(z_{jt};\sigma_0)\bigr|^2\tau_T(z_{jt})\Big]\\
    & +\;\E_J\E_T\Big[\bigl|\hat\mu_{p,j}(z_{jt})-\mu_{p,j}(z_{jt})\bigr|^2\tau_T(z_{jt})\Big]
    +\mathbb E_J\mathbb E_T\Bigl[\bigl(\hat\omega_j(z_{jt})^{-1}-\omega_j(z_{jt})^{-1}\bigr)^2\tau_T(z_{jt})\Bigr]=o_p(T^{-1/2}).
\end{align*}
The last term is equivalent to $\mathbb E_J\mathbb E_T(\hat\Lambda_j-\Lambda_j\tau_T)^2$.
\item\label{cond:norm-linear} (Asymptotic linearity) There is a market-level function $\delta(\cdot)$ with $\E[\delta(w_t)]=0$ and $\E\|\delta(w_t)\|^2<\infty$ such that
\[
  \E_J\E_T\Big[\big\{(\hat\mu_{y,j}-\mu_{y,j})-\alpha_0(\hat\mu_{p,j}-\mu_{p,j})\big\}\,
  \Lambda_j(z_{jt})\,\tau_T(z_{jt})\,\Gamma_{jt}\Big]
  =\E_T\delta(w_t)+o_p(T^{-1/2}),
\]
all first-stage objects being evaluated at $\sigma_0$.
\end{enumerate}
\end{condition}

\begin{remark}
Condition~\ref{cond:y-smooth} is mild: three-times differentiability of $\mu_{y,j}(z;\sigma)$ in $\sigma$ comes from smoothness of the inverse function $y_{j}(\sigma,w)$, which is inherited from smoothness of the random-coefficient density $f(b;\sigma)$ in $\sigma$.
 
Condition~\ref{cond:normality} strengthens Condition~\ref{cond:consistency} from consistency to $\sqrt T$-inference, and is likewise verified for the kernel and series estimators in Sections~\ref{sec:theory-kernel} and \ref{sec:theory-series}. 
Part~\eqref{cond:deriv} is a definitional requirement: the derivative estimators are the analytic $\sigma$-derivatives of the first-stage fit, not separately smoothed objects; this is what enables analysis on $\partial_\sigma \hat\mu_{y,j}$. 
Part~\eqref{cond:norm-deriv} then asks these derivative estimators to be mean-square consistent uniformly over $\sigma$; this is an analogy to the Condition~\ref{cond:consistency}\eqref{cond:cons-rate}, but for derivatives, which are concerned in asymptotic normality.
Parts~\eqref{cond:norm-msr} and \eqref{cond:norm-linear} are the heart of the normality argument, and both exploit the degeneracy of the moment at the truth, $r_{jt}(\theta_0)=0$, which holds for every admissible weight. Part~\eqref{cond:norm-msr} imposes the familiar $T^{-1/4}$ root-mean-square rate on every first-stage object, including the weight---on the retained region only, matching Condition~\ref{cond:consistency}\eqref{cond:cons-uniform}; this is precisely what makes the second-order (quadratic) remainder in the expansion of the moment negligible, since that remainder is a product of two first-stage errors each of order $o_p(T^{-1/4})$. 
Part~\eqref{cond:norm-linear} is the high-level asymptotic-linearity requirement that isolates the surviving first-order term; the function $\delta$ is the influence function of $\hat\theta$ (up to the Jacobian factor $-M^{-1}$), whose explicit form is derived for each estimator in Sections~\ref{sec:theory-kernel} and \ref{sec:theory-series}. 

\end{remark}

\begin{theorem}\label{thm:normality}
Suppose Condition~\ref{cond:regularity}, \ref{cond:consistency}, \ref{cond:y-smooth}, \ref{cond:normality} hold and that
\[
  M=\mathbb E_J\,\mathbb E\left[\Lambda_j(z_{jt})\,\Gamma_{jt}\Gamma_{jt}'\right]
   =\mathbb E_J\,\mathbb E\left[\omega_j(z_{jt})^{-1}
   \begin{pmatrix}\dot\mu_{y,j}(z_{jt};\sigma_0)\\[2pt] -\mu_{p,j}(z_{jt})\\[2pt] -x_{jt}\end{pmatrix}
   \bigl(\dot\mu_{y,j}(z_{jt};\sigma_0),\,-\mu_{p,j}(z_{jt}),\,-x_{jt}'\bigr)\right]
\]
is nonsingular. Then
\[
  \sqrt T\,(\hat\theta-\theta_0)\;\xrightarrow{\ d\ }\;N\bigl(0,\;M^{-1}V_\delta(M^{-1})'\bigr),
  \qquad V_\delta=\operatorname{Var}\bigl(\delta(w_t)\bigr),
\]
with $\delta$ as in Condition~\ref{cond:normality}\eqref{cond:norm-linear}.
\end{theorem}

\begin{remark}
The estimator is $\sqrt T$-consistent and asymptotically normal, with the
sandwich form $M^{-1}V_\delta(M^{-1})'$ standard for a two-step estimator:
$M$ is the Jacobian of the moment and $V_\delta$ the variance of its
influence function.  Nonsingularity of $M$ is the local identification
requirement that the moment be sensitive to $\theta$ in every direction at
$\theta_0$.  Two features of the statement are worth emphasizing.  First, the
weight enters only through $M$ and $V_\delta$: the argument is unchanged for
any admissible $\omega_j$. In the kernel and series implementations
$\delta(w_t)=\mathbb E_J[\xi_{jt}a_{jt}]$ with $a_{jt}=\Lambda_j(z_{jt})\Gamma_{jt}$,
giving $V_\delta=\mathbb E[(\mathbb E_J a_{jt}\xi_{jt})(\mathbb E_J a_{jt}\xi_{jt})']$.
Second, the trimming leaves no trace: both $M$ and $V_\delta$ are the
untrimmed population objects.
Efficiency is a property of the particular weight $\omega_j=\Sigma_j$: when
$\xi_{jt}$ are independent across $j$ conditional on $z_t$, a matrix
Cauchy--Schwarz inequality gives
$M^{-1}V_\delta(M^{-1})'\succeq J^{-1}\bigl(\mathbb E_J\mathbb E[\Sigma_j^{-1}\Gamma_{jt}\Gamma_{jt}']\bigr)^{-1}$
for every admissible weight, with equality at $\omega_j=\Sigma_j$, in which
case $V_\delta=M/J$ and the asymptotic variance collapses to $M^{-1}/J$, the
semiparametric efficiency bound for the conditional-moment restriction \eqref{eq:theory-id}.
Without conditional independence across $j$ the efficient weight becomes more complicated.
\end{remark}

\subsection{Theory for Kernel Estimation}\label{sec:theory-kernel}
This section provides primitive conditions for the kernel estimator in~\eqref{est:kernel} to satisfy the high-level conditions of Section~\ref{sec:theory-general}.
It is convenient to work with the kernel \emph{numerator}, the density-weighted regression function, and to keep the density estimate separate. For a generic regressand $\psi_{jt}(\sigma)$ (below, $\psi_{jt}(\sigma)\in\{y_j(\sigma,w_t),\ \partial_\sigma y_j(\sigma,w_t),\ \partial^2_\sigma y_j(\sigma,w_t),\ p_{jt}\}$, all of which are uniformly bounded by Condition~\ref{cond:regularity} \eqref{cond:cons-reg} and Condition~\ref{cond:kernel}\eqref{cond:kernel-psi}), define
\[
    m_{\psi,j}(z;\sigma)=\mu_{\psi,j}(z;\sigma)\,f_j(z),
    \qquad
    \mu_{\psi,j}(z;\sigma)=\E[\psi_{jt}(\sigma)\mid z_{jt}=z],
\]
with estimators
\[
    \hat m_{\psi,j}(z;\sigma)=\E_T\big[K_h(z_{jt}-z)\,\psi_{jt}(\sigma)\big],
    \qquad
    \hat f_j(z)=\E_T\big[K_h(z_{jt}-z)\big],
    \qquad
    \hat\mu_{\psi,j}(z;\sigma)=\frac{\hat m_{\psi,j}(z;\sigma)}{\hat f_j(z)}.
\]
Specializing $\psi$ recovers the objects of Section~\ref{sec:theory-general}: writing $m_{y,j}=\mu_{y,j}f_j$, $\dot m_{y,j}=\dot\mu_{y,j}f_j$, $\ddot m_{y,j}=\ddot\mu_{y,j}f_j$, and $m_{p,j}=\mu_{p,j}f_j$, the first-stage estimators are $\hat\mu_{y,j}=\hat m_{y,j}/\hat f_j$ and $\hat\mu_{p,j}=\hat m_{p,j}/\hat f_j$, and the derivative estimators are $\hat{\dot\mu}_{y,j}=\hat{\dot m}_{y,j}/\hat f_j=\partial_{\sigma}\hat\mu_{y,j}$ and $\hat{\ddot\mu}_{y,j}=\hat{\ddot m}_{y,j}/\hat f_j=\partial^2_{\sigma}\hat\mu_{y,j}$, obtained by replacing $\psi=y_j(\sigma)$ with its $\sigma$-derivatives in the numerator (the denominator $\hat f_j$ does not depend on $\sigma$).

\paragraph{Weighting.} Section~\ref{sec:theory-general} leaves the weight $\omega_j$ free. From here on we specialize to
\[
    \omega_j(z)=\Sigma_j(z)=\E[\xi_{jt}^2\mid z_{jt}=z],
    \qquad\text{so that}\qquad
    \Lambda_j(z)=\Sigma_j(z)^{-1},
    \qquad
    \hat\Lambda_j(z)=\hat\Sigma_j(z)^{-1}\tau_T(z),
\]
with $\hat\Sigma_j$ the kernel estimator of $\Sigma_j$ obtained by smoothing the squared preliminary residuals which is constructed using the preliminary estimator $\tilde\theta$. With an abuse of notation, the preliminary estimator $\tilde\theta$ that supplies those residuals is the case $\omega_j\equiv1$, for which $\Lambda_j\equiv1$ and $\hat\Lambda_j=\tau_T$.

\paragraph{Trimming.} The trimming introduced in Section~\ref{sec:theory-general} is instantiated here, and the kernel dictates the choice.  Let $R<\infty$ be such that $\operatorname{supp}K\subseteq[-R,R]^{d_z}$ (Condition~\ref{cond:kernel}\eqref{cond:kernel-smoothness-bw}), write $d(z,\partial\mathcal Z)=\inf_{z'\in\partial\mathcal Z}\|z-z'\|_\infty$ for the sup-norm distance from $z$ to the boundary $\partial\mathcal Z$ and set
\begin{equation}\label{eq:kernel-trim}
    \varepsilon_T=Rh,
    \qquad
    \mathcal Z_T=\big\{z\in\mathcal Z:\ d(z,\partial\mathcal Z)\ge Rh\big\},
    \qquad
    \tau_T(z)=\mathbf{1}\big\{d(z,\partial\mathcal Z)\ge Rh\big\}.
\end{equation}
Because $K_h(\cdot-z)$ is supported on $\{z':\|z'-z\|_\infty\le Rh\}$, the definition says exactly that the smoothing window of a retained point lies entirely inside $\mathcal Z$:
\begin{equation}\label{eq:window-inside}
    z\in\mathcal Z_T
    \iff
    \big\{z':\|z'-z\|_\infty\le Rh\big\}\subseteq\mathcal Z .
\end{equation}
Since $h=h_T\to0$ is nonrandom, so is $\tau_T$, as Condition~\ref{cond:consistency}\eqref{cond:cons-trim} requires.

\begin{condition}\label{cond:kernel}
    Let $\mathcal Z$ denote the support of $z_{jt}$. 
\begin{enumerate}
    \item\label{cond:kernel-boundedness} $\mathcal Z$ is a compact set with Lipschitz boundary. $\bar f\ge\sup_{z\in\mathcal Z}f_j(z)\ge\inf_{z\in\mathcal Z}f_j(z)\ge\underline f>0$, where $f_j$ is the density of $z_{jt}$. $0<\underline c\le \Sigma_j(z)^{-1}\le \bar c<\infty$ for all $z\in\mathcal Z$ and $j=1,\dots,J$.
    \item\label{cond:kernel-smoothness-bw} The kernel $K:\R^{d_z}\to\R$ is a Lipschitz function. It has compact support, $\operatorname{supp}K\subseteq[-R,R]^{d_z}$ for some $R<\infty$, and order $l$, i.e.
    \[
        \int K(u)\,du=1,\qquad \int K(u)\,u^{m}\,du=0 \ \ \text{for } 1\le m<l,
    \]
    and $K_h(u)=h^{-d_z}K(u/h)$ with bandwidth $h=h_T\to0$ and $u\in\R^{d_z}$. 
    The kernel has order $l>d_z$, and the bandwidth is $h=T^{-q}$ with $\tfrac{1}{2l}<q<\tfrac{1}{2d_z}$.
    \item\label{cond:kernel-psi} 
    For each $\psi_{jt}(\sigma) \in \bigl\{y_j(\sigma, w_t),\ \partial_\sigma y_j(\sigma, w_t),\ \partial_\sigma^2 y_j(\sigma, w_t)\bigr\}$, we have $\sup_{w_t} |\psi_{jt}(\sigma_0)| < C$, and $\psi_{jt}(\sigma)$ is Lipschitz continuous in $\sigma$; that is, there exists a constant $L > 0$ such that
\[
    |\psi_{jt}(\sigma_1) - \psi_{jt}(\sigma_2)| \le L\,|\sigma_1 - \sigma_2|, \qquad \forall\, \sigma_1, \sigma_2 \in \Theta_\sigma.
\]
Moreover, for each $\psi_{jt}(\sigma) \in \bigl\{y_j(\sigma, w_t),\ \partial_\sigma y_j(\sigma, w_t),\ \partial_\sigma^2 y_j(\sigma, w_t)\bigr\}$, the map $m_{\psi,j}(\cdot\,;\sigma)$ is $l$-times continuously differentiable in $z$, with derivatives bounded uniformly over $\Theta_\sigma \times \mathcal{Z}$:
\[
    \sup_{\sigma \in \Theta_\sigma,\ z \in \mathcal{Z},\ k \le l} \bigl\|\partial_z^k m_{\psi,j}(z; \sigma)\bigr\| \le C.
\]
Finally, for each $\psi(z) \in \{\mu_{p,j}(z),\ f_j(z),\ \Sigma_j(z)\}$, the map $\psi(\cdot)$ is $l$-times continuously differentiable in $z$, with derivatives bounded uniformly over $\mathcal{Z}$:
\[
    \sup_{z \in \mathcal{Z},\ k \le l} \bigl\|\partial_z^k \psi(z)\bigr\| \le C.
\]
\end{enumerate}
\end{condition}

\begin{remark}
Condition~\ref{cond:kernel} collects the standard primitives for uniform kernel rates, together with what the trimming requires.
 
Part~\eqref{cond:kernel-boundedness} restricts attention to a compact support on which the instrument density is bounded away from zero. This lets the ratio $\hat\mu_{\psi,j}=\hat m_{\psi,j}/\hat f_j$ inherit the numerator's rate on $\mathcal Z_T$, since there $\hat f_j\ge\underline f/2$ with probability approaching one. It also bounds the mass discarded by the trimming: for a compact set with Lipschitz boundary the collar satisfies $Leb\{z:d(z,\partial\mathcal Z)<\varepsilon\}\le C\varepsilon$, so with $f_j\le\bar f$,
\[
    p_T=\max_j\Pr\big(\tau_T(z_{jt})=0\big)\le \bar f\,C R\,h=O(h)\longrightarrow0,
\]
which verifies Condition~\ref{cond:consistency}\eqref{cond:cons-trim}.
 
Part~\eqref{cond:kernel-smoothness-bw} is the higher-order-kernel construction: a compactly supported, Lipschitz kernel of order $l$ yields a bias of order $h^{l}$. The kernel order and bandwidth are chosen so that the uniform rate $\rho_T=(\log T/Th^{d_z})^{1/2}+h^{l}$ is $o(T^{-1/4})$, in fact fast enough that the estimated weight $\hat\Sigma_{j}$ also converges at $o_p(T^{-1/4})$, and keep the bias term introduced by the first stage estimation negligible; this is what the tightened window $\tfrac{1}{2l}<q<\tfrac{1}{2d_z}$ secures. The requirement $l>d_z$ keeps this interval nonempty, so a large instrument dimension $d_z$ calls for a high-order kernel---the standard cost of dimensionality here. The support radius $R$ enters no rate; it is named only to fix the trimming in \eqref{eq:kernel-trim}, and the choice $\varepsilon_T=Rh$ is what makes the  bias legitimate. The regression functions are smooth only \emph{inside} $\mathcal Z$: by part~\eqref{cond:kernel-boundedness} the density is bounded away from zero on $\mathcal Z$ and is zero outside it, so $f_j$ and hence $m_{\psi,j}=\mu_{\psi,j}f_j$ jump at $\partial\mathcal Z$. For $z$ within $Rh$ of the boundary the change of variables in the bias step integrates the kernel over the truncated region $\{v:z+vh\in\mathcal Z\}$, on which $\int K(v)\,dv\ne1$ and the higher-order moments no longer vanish; the local-constant estimator then has larger bias, and $\rho_T=o(T^{-1/4})$ fails. By \eqref{eq:window-inside} the trimming removes exactly those points and nothing more.

Part~\eqref{cond:kernel-psi} imposes boundedness and Lipschitz continuity in $\sigma$ on each regressand, and $l$-fold smoothness in $z$ on the density-weighted targets $m_{\psi,j}$, the density $f_j$, and the variance weight $\Sigma_{j}$. The Lipschitz-in-$\sigma$ property supplies the uniform-over-$\Theta_\sigma$ boundedness, while the $z$-smoothness of $\Sigma_{j}$ is what ensures the kernel estimator of $\Sigma_{j}$ inherits the same rate. Note that this implies $\mu_{\psi,j}$ is $l$-fold smooth in $z$ as well, since $\mu_{\psi,j}=m_{\psi,j}/f_j$ and $f_j$ is bounded away from zero on $\mathcal Z$.
\end{remark}

The limiting distribution of the kernel estimator then follows from Theorem~\ref{thm:normality}.
\begin{theorem}[Asymptotic normality for kernel estimator]\label{thm:kernel-normality}
Suppose Conditions~\ref{cond:regularity}, \ref{cond:y-smooth} and \ref{cond:kernel} hold, and that $M=\mathbb E_J\,\mathbb E\left[\Sigma_j(z_{jt})^{-1}\,\Gamma_{jt}\Gamma_{jt}'\right]$ is nonsingular. Then the kernel estimator $\hat\theta$ satisfies
\[
    \sqrt T\,(\hat\theta-\theta_0)\xrightarrow{d}N\!\big(0,M^{-1}V\ (M^{-1})'\big),
\]
with $V=\E[(\E_J \xi_{jt}\Sigma_j(z_{jt})^{-1}\Gamma_{jt})(\E_J \xi_{jt}\Sigma_j(z_{jt})^{-1}\Gamma_{jt})']$.
\end{theorem}

\begin{remark}
Theorem~\ref{thm:kernel-normality} shows the kernel estimator is asymptotically normal at the rate $T^{-1/2}$, despite the nonparametric first stage. Because the weight equals the conditional variance $\Sigma_{j}(z_{jt})=\E[\xi_{jt}^2\mid z_{jt}]$, under a stronger condition that $\xi_{jt}$ are mean-zero and independent across $j$ conditional on $z_{t}$, the sandwich collapses to $V=M/J$, and the asymptotic variance simplifies to $M^{-1}/J$. Then the estimator attains the semiparametric efficiency bound for the conditional-moment restriction \eqref{eq:theory-id}. The first-stage rate requirements are entirely absorbed into the bandwidth and kernel choice of Condition~\ref{cond:kernel}\eqref{cond:kernel-smoothness-bw}: it removes the first-stage bias, so no bias term enters the limit and the distribution is centered at zero. The only price is the high-order kernel needed when $d_z$ is large.   
\end{remark}

\subsection{Theory for Series Estimation}\label{sec:theory-series}

We now verify the requirements of Section~\ref{sec:theory-general} for a series first stage. Fix a product $j$; the estimator projects the regressand onto a $K$-dimensional dictionary. Let $\phi^K(z)=(\phi_1(z),\dots,\phi_K(z))'$ be a vector of basis functions. For a generic regressand $\psi_{jt}(\sigma)\in\{y_j(\sigma,w_t),\ \partial_\sigma y_j(\sigma,w_t),\ \partial^2_\sigma y_j(\sigma,w_t),\ p_{jt}\}$, the least-squares series estimator of the conditional mean $\mu_{\psi,j}(z;\sigma)=\E[\psi_{jt}(\sigma)\mid z_{jt}=z]$ is
\[
    \hat\mu_{\psi,j}(z;\sigma)=\phi^K(z)'\hat\Pi_{\psi,j}(\sigma),
    \qquad
    \hat\Pi_{\psi,j}(\sigma)=\hat Q_j^{-1}\,\E_T\!\big[\phi^K(z_{jt})\,\psi_{jt}(\sigma)\big],
\]
where $\hat Q_j=\E_T[\phi^K(z_{jt})\phi^K(z_{jt})']$ is the empirical Gram matrix and $Q_j=\E[\phi^K(z_{jt})\phi^K(z_{jt})']$ its population counterpart. Specializing $\psi$ recovers the four first-stage objects of Section~\ref{sec:theory-general}: $\hat\mu_{y,j}$, the derivative estimators $\hat{\dot\mu}_{y,j}$ and $\hat{\ddot\mu}_{y,j}$—obtained by replacing $\psi=y_j(\sigma,w_t)$ with its $\sigma$-derivatives in the regressand, the design $\hat Q_j$ being free of $\sigma$—and $\hat\mu_{p,j}$. Write $e^{\psi}_{jt}(\sigma)=\psi_{jt}(\sigma)-\mu_{\psi,j}(z_{jt};\sigma)$ for the projection residual, which satisfies $\E[e^{\psi}_{jt}(\sigma)\mid z_{jt}]=0$. Finally, let
\[
    \zeta_0(K)=\sup_{z\in\mathcal Z}\big\|\phi^K(z)\big\|
\]
denote the envelope of the basis.

\paragraph{Trimming.} The trimming function is trivial here: $\tau_T(z)\equiv1$ for all $z\in\mathcal Z$, so that $\mathcal Z_T=\mathcal Z$. The series estimator is defined on the entire support of the instruments, and no trimming is needed.

\begin{condition}\label{cond:series}
\begin{enumerate}
    \item\label{cond:series-gram} \textup{(Design)} $\mathcal Z$ is compact, and the eigenvalues of $Q_j$ are bounded above and away from zero uniformly in $j$: $0<\underline{q}\le\lambda_{\min}(Q_j)\le\lambda_{\max}(Q_j)\le\overline{q}<\infty$. The basis satisfies the envelope bound $\zeta_0(K)<\infty$ for each $K$.  $0<\underline c\le \Sigma_j(z)^{-1}\le \bar c<\infty$ for all $z\in\mathcal Z$ and $j=1,\dots,J$.
    \item\label{cond:series-approx} \textup{(Approximation)} There is $\gamma>0$ such that, for each\\
     $\psi\in\{y_j(\sigma,w_t),\ \partial_\sigma y_j(\sigma,w_t),\ \partial^2_\sigma y_j(\sigma,w_t),\ p_{jt},\ x_{jt}',\ \xi_{jt}^2, \Sigma_{j}(z_{jt})^{-1}, \Sigma_{j}(z_{jt})^{-1}\cdot \partial_\sigma y_{j}(\sigma,w_t),\Sigma_{j}(z_{jt})^{-1}\cdot p_{jt}, \Sigma_{j}(z_{jt})^{-1}\cdot x_{jt}'\}$, there exist coefficient vectors $\Pi^{K}_{\psi,j}(\sigma)$ with
    \[
        \sup_{\sigma\in\Theta_\sigma}\ \sup_{z\in\mathcal Z}\big|\mu_{\psi,j}(z;\sigma)-\phi^K(z)'\Pi^{K}_{\psi,j}(\sigma)\big|=O(K^{-\gamma}).
    \]
    \item\label{cond:series-Lipschitz} \textup{(Lipschitz)} For each $\psi\in\{y_j(\sigma,w_t),\ \partial_\sigma y_j(\sigma,w_t),\ \partial^2_\sigma y_j(\sigma,w_t)\}$, $\sup_{w_t}|\psi_{jt}(\sigma_0)|<C$, and there exists a constant $L>0$ that does not depend on data $z_{jt}$ such that $|\psi_{jt}(\sigma_1)-\psi_{jt}(\sigma_2)|\le L|\sigma_1-\sigma_2|$ for all $\sigma_1,\sigma_2\in\Theta_\sigma$.
    \item\label{cond:series-rate} \textup{(Rates)} As $T\to\infty$, $K=K_T\to\infty$ with
    \[
        \sqrt T K^{-\gamma}=o(1),
        \qquad
        \frac{\zeta_0(K) K \log (K\vee T)}{\sqrt{T}}=o(1).
    \]
\end{enumerate}
\end{condition}

\begin{remark}
Condition~\ref{cond:series} is the series counterpart of the kernel primitives in Condition~\ref{cond:kernel} and follows the framework of \citet{newey1997convergence}.
Part~\eqref{cond:series-gram} restricts the dictionary so that the population design $Q_j$ is well-conditioned. Together with the envelope $\zeta_0(K)$, this controls the estimation error of the empirical Gram matrix, and hence ensures $\hat Q_j$ is invertible with probability approaching one, so that the linear projection is well defined. It is the series analogue of the kernel requirement that the instrument density be bounded away from zero.
Part~\eqref{cond:series-approx} is the sieve approximation requirement, imposing a common rate $K^{-\gamma}$ on the four first-stage regression functions $\mu_{y,j},\dot\mu_{y,j},\ddot\mu_{y,j},\mu_{p,j}$ and, in addition, on $x_{jt}$, the conditional variance $\Sigma_{j}(z_{jt})=\E[\xi_{jt}^2\mid z_{jt}]$, its inverse, the multiplications $\Sigma_j^{-1}\cdot \partial_\sigma y_{j}(\sigma,w_t)$, $\Sigma_j^{-1}\cdot p_{jt}$, and $\Sigma_j^{-1}\cdot x_{jt}'$. The last two are needed because $x_{jt}$ and $\Sigma$ enter the moment and because the feasible weight is itself a series fit of squared residuals. The rate $\gamma$ reflects the smoothness of these functions relative to the basis; for $s$-times differentiable functions of a $d_z$-dimensional instrument, $\gamma=s/d_z$ with polynomial or spline dictionaries. Particularly, if the function is analytic and univariate, the approximation error diminishes at $\exp(-K)$, with polynomial dictionaries.
Part~\eqref{cond:series-Lipschitz} is the series analogue of the kernel Lipschitz-in-$\sigma$ requirement, which is needed to control the uniformity of the first-stage convergence over $\Theta_\sigma$.
Part~\eqref{cond:series-rate} collects the two operative rate restrictions. The first, $\sqrt T\,K^{-\gamma}=o(1)$, requires the dictionary must be rich enough that the approximation bias is negligible relative to $T^{-1/2}$, so no bias term enters the limit—the series analogue of $h^{l}=o(T^{-1/2})$. The second, $\zeta_0(K)K\log (K\vee T)/\sqrt T=o(1)$, is the binding upper restriction on how fast $K$ may grow; it controls the estimation error of $\hat Q_j^{-1}$ and, through it, the second-order remainder in the linearization, playing the role that $q<1/(2d_z)$ plays for the kernel. As in the kernel case, the crude rate $T^{-1/4}$ suffices only because the moment is degenerate at the truth ($r_{jt}(\theta_0)=0$), so the error enters at the second order.
\end{remark}

\begin{theorem}[Asymptotic normality, series estimator]\label{thm:series-normality}
Suppose Conditions~\ref{cond:regularity}, \ref{cond:y-smooth} and \ref{cond:series} hold, and that $M=\mathbb E_J\,\mathbb E\left[\Sigma_j(z_{jt})^{-1}\,\Gamma_{jt}\Gamma_{jt}'\right]$ is nonsingular. Then the series estimator $\hat\theta$ satisfies
\[
    \sqrt T\,(\hat\theta-\theta_0)\xrightarrow{d}N\!\big(0,M^{-1}V\ (M^{-1})'\big),
\]
with $V=\E[(\E_J \xi_{jt}\Sigma_j(z_{jt})^{-1}\Gamma_{jt})(\E_J \xi_{jt}\Sigma_j(z_{jt})^{-1}\Gamma_{jt})']$.
\end{theorem}

\begin{remark}
The series estimator attains the same limiting distribution $N(0,M^{-1}V(M^{-1})')$---same Jacobian $M$ and same influence function---as the kernel estimator of Theorem~\ref{thm:kernel-normality}. The reason is that the influence function does not depend on which smoother produced it. In both implementations the projection of the first-stage error onto $a(z_{jt})$ collapses to $\E_T\E_J[\xi_{jt}a_{jt}]$---Lemma~\ref{lemma:kernel-linear} for the kernel and Lemma~\ref{lemma:series-linear} for the series---while the smoother-specific pieces, the bias $h^{l}$ or $K^{-\gamma}$ and the quadratic remainders, are held below $T^{-1/2}$ by the rate restrictions in Conditions~\ref{cond:kernel}\eqref{cond:kernel-smoothness-bw} and \ref{cond:series}\eqref{cond:series-rate}. The choice between the two is therefore governed by finite-sample and computational considerations rather than by efficiency. 
\end{remark}

\section{Monte Carlo Simulation}
In this section, we provide simulation evidence to show performance of estimators across scenarios.

\begin{table}[t]
    \centering
    \begin{threeparttable}
        \caption{Simulation Results $\lambda = 0$}
        \label{tab:simulation-results-0}

        \setlength{\tabcolsep}{6pt}
        \renewcommand{\arraystretch}{1.15}

        \begin{tabular}{lrrrrrrrr}
            \toprule
            Parameter
            & True
            & Bias
            & Abs.\ Bias
            & St.\ Err.
            & RMSE
            & Min
            & Max
            & Rej.\ Rate \\
            \midrule

            \multicolumn{9}{l}{\textit{Kernel Estimator}} \\
\cmidrule(lr){1-9}

$\sigma$
& 1.000
& -0.026
& 0.066
& 0.108
& 0.110
& 0.600
& 1.300
& 0.060 \\

$\alpha$
& -1.000
& 0.017
& 0.046
& 0.070
& 0.071
& -1.174
& -0.727
& 0.060 \\

\addlinespace

\multicolumn{9}{l}{\textit{Series Estimator (B-spline, df $= 5$)}} \\
\cmidrule(lr){1-9}

$\sigma$
& 1.000
& -0.006
& 0.050
& 0.077
& 0.076
& 0.800
& 1.100
& 0.040 \\

$\alpha$
& -1.000
& 0.003
& 0.037
& 0.051
& 0.051
& -1.090
& -0.879
& 0.040 \\

\addlinespace

\multicolumn{9}{l}{
    \textit{GMM Estimator 1: Infeasible Optimal IV}
} \\
\cmidrule(lr){1-9}

$\sigma$
& 1.000
& 2.462
& 2.514
& 2.944
& 3.815
& 0.800
& 10.000
& 0.200 \\

$\alpha$
& -1.000
& -1.738
& 1.774
& 2.077
& 2.692
& -7.401
& -0.863
& 0.180 \\

\addlinespace

\multicolumn{9}{l}{
    \textit{GMM Estimator 2: Feasible Optimal IV}
} \\
\cmidrule(lr){1-9}

$\sigma$
& 1.000
& 0.916
& 1.008
& 2.088
& 2.261
& 0.600
& 9.200
& 0.120 \\

$\alpha$
& -1.000
& -0.647
& 0.708
& 1.464
& 1.587
& -6.566
& -0.754
& 0.120 \\

            \bottomrule
        \end{tabular}

        \begin{tablenotes}
            \footnotesize
            \item \textit{Notes:}
            Abs.\ Bias denotes mean absolute bias.
            Rej.\ Rate denotes the empirical rejection rate.
            GMM Estimator 1 uses the true conditional expectation
            in the optimal instrument.
            GMM Estimator 2 uses the feasible approximation in
            \citet{berry1999voluntary}.
        \end{tablenotes}
    \end{threeparttable}
\end{table}

\begin{table}[t]
    \centering
    \begin{threeparttable}
        \caption{Simulation Results $\lambda = 0.3$}
        \label{tab:simulation-results-0.3}

        \setlength{\tabcolsep}{6pt}
        \renewcommand{\arraystretch}{1.15}

        \begin{tabular}{lrrrrrrrr}
            \toprule
            Parameter
            & True
            & Bias
            & Abs.\ Bias
            & St.\ Err.
            & RMSE
            & Min
            & Max
            & Rej.\ Rate \\
            \midrule

\multicolumn{9}{l}{\textit{Kernel Estimator}} \\
\cmidrule(lr){1-9}

$\sigma$
& 1.000
& -0.084
& 0.124
& 0.239
& 0.251
& 0.100
& 1.500
& 0.100 \\

$\alpha$
& -1.000
& 0.080
& 0.111
& 0.320
& 0.327
& -1.319
& 1.152
& 0.020 \\

\addlinespace

\multicolumn{9}{l}{\textit{Series Estimator (B-spline, df $= 5$)}} \\
\cmidrule(lr){1-9}

$\sigma$
& 1.000
& -0.004
& 0.040
& 0.073
& 0.072
& 0.800
& 1.200
& 0.060 \\

$\alpha$
& -1.000
& 0.003
& 0.032
& 0.050
& 0.049
& -1.127
& -0.861
& 0.060 \\

\addlinespace

\multicolumn{9}{l}{
    \textit{GMM Estimator 1: Infeasible Optimal IV}
} \\
\cmidrule(lr){1-9}

$\sigma$
& 1.000
& 0.396
& 0.436
& 1.560
& 1.594
& 0.800
& 8.200
& 0.060 \\

$\alpha$
& -1.000
& -0.272
& 0.303
& 1.077
& 1.100
& -6.019
& -0.863
& 0.060 \\

\addlinespace

\multicolumn{9}{l}{
    \textit{GMM Estimator 2: Feasible Optimal IV}
} \\
\cmidrule(lr){1-9}

$\sigma$
& 1.000
& 0.470
& 0.578
& 1.923
& 1.961
& 0.100
& 9.800
& 0.060 \\

$\alpha$
& -1.000
& -0.333
& 0.400
& 1.347
& 1.375
& -7.138
& -0.585
& 0.060 \\

            \bottomrule
        \end{tabular}

        \begin{tablenotes}
            \footnotesize
            \item \textit{Notes:}
            Abs.\ Bias denotes mean absolute bias.
            Rej.\ Rate denotes the empirical rejection rate.
            GMM Estimator 1 uses the true conditional expectation
            in the optimal instrument.
            GMM Estimator 2 uses the feasible approximation in
            \citet{berry1999voluntary}.
        \end{tablenotes}
    \end{threeparttable}
\end{table}

\begin{table}[t]
    \centering
    \begin{threeparttable}
        \caption{Simulation Results $\lambda = 1$}
        \label{tab:simulation-results-1}

        \setlength{\tabcolsep}{6pt}
        \renewcommand{\arraystretch}{1.15}

        \begin{tabular}{lrrrrrrrr}
            \toprule
            Parameter
            & True
            & Bias
            & Abs.\ Bias
            & St.\ Err.
            & RMSE
            & Min
            & Max
            & Rej.\ Rate \\
            \midrule

\multicolumn{9}{l}{\textit{Kernel Estimator}} \\
\cmidrule(lr){1-9}

$\sigma$
& 1.000
& -0.024
& 0.128
& 0.292
& 0.291
& 0.100
& 2.400
& 0.060 \\

$\alpha$
& -1.000
& 0.039
& 0.095
& 0.287
& 0.286
& -1.893
& 0.750
& 0.040 \\

\addlinespace

\multicolumn{9}{l}{\textit{Series Estimator (B-spline, df $= 5$)}} \\
\cmidrule(lr){1-9}

$\sigma$
& 1.000
& -0.012
& 0.032
& 0.063
& 0.063
& 0.800
& 1.100
& 0.040 \\

$\alpha$
& -1.000
& 0.007
& 0.025
& 0.038
& 0.038
& -1.075
& -0.867
& 0.060 \\

\addlinespace

\multicolumn{9}{l}{
    \textit{GMM Estimator 1: Infeasible Optimal IV}
} \\
\cmidrule(lr){1-9}

$\sigma$
& 1.000
& -0.004
& 0.040
& 0.067
& 0.066
& 0.800
& 1.100
& 0.020 \\

$\alpha$
& -1.000
& 0.002
& 0.028
& 0.040
& 0.040
& -1.076
& -0.867
& 0.020 \\

\addlinespace

\multicolumn{9}{l}{
    \textit{GMM Estimator 2: Feasible Optimal IV}
} \\
\cmidrule(lr){1-9}

$\sigma$
& 1.000
& 0.004
& 0.084
& 0.194
& 0.192
& 0.100
& 1.800
& 0.040 \\

$\alpha$
& -1.000
& -0.006
& 0.052
& 0.109
& 0.108
& -1.527
& -0.585
& 0.040 \\

            \bottomrule
        \end{tabular}

        \begin{tablenotes}
            \footnotesize
            \item \textit{Notes:}
            Abs.\ Bias denotes mean absolute bias.
            Rej.\ Rate denotes the empirical rejection rate.
            GMM Estimator 1 uses the true conditional expectation
            in the optimal instrument.
            GMM Estimator 2 uses the feasible approximation in
            \citet{berry1999voluntary}.
        \end{tablenotes}
    \end{threeparttable}
\end{table}

\subsection{Simulation Design}

In the simulations, we illustrate the finite-sample performance of different estimators when the true data-generating process resembles Example \ref{eg:counter}. We draw samples of size $T=500$ with $J=2$. The mean utility takes the form $y_{jt} = p_{jt} \alpha_0 + \xi_{jt}$. The random coefficient for $p_{jt}$ follows $N(0,\sigma_0^2)$. We set the true parameters $\alpha_0 = -1$ and $\sigma_0 = 1$. We draw $\xi_{jt} \sim \text{U}[-1,1]$, independent of the instrument $z_{jt}$, with $p_{jt} = z_{jt} + 0.3 \xi_{jt}$. The instrument $z_{j,t}$ is generated from a mixed distribution: denoting $F_z \coloneqq \frac{5}{11} N(-2,0.0625) + \frac{1}{11} N(1, 0.0625) + \frac{5}{11} N(2,0.0625)$, we draw $z_{1t} \sim F_z$ and $z_{2t} \sim \lambda U[-2,2] + (1-\lambda)F_z$. We consider three configurations with $\lambda=0$, $0.3$, and $1$. Notice that when $\lambda=1$, the distribution of $z_{jt}$ is discrete, similar to the distribution of $p_{jt}$ in Example \ref{eg:counter} with no endogeneity. Our setup therefore illustrates performance of estimators under various generalization of Example \ref{eg:counter}.

For each of the three configurations, we implement estimation of $(\sigma_0, \alpha_0)$ with a kernel estimator, a series estimator, and two GMM estimators. The kernel and series estimators are based on the conditional moment. The two GMM estimators are based on the unconditional moment with respectively the infeasible optimal instrument and a feasible optimal instrument. We implement the estimation in each of the $50$ simulations.

More specifically, we construct the kernel estimator using the sixth-order Epanechnikov kernel, and construct the series estimator using B-splines. For the first GMM estimation, we adopt the infeasible optimal instrument as
\begin{align*}
    \left( \E\left[ x_{jt} \mid z_{jt} \right], \E\left[\frac{\partial y_j}{\partial \sigma} \left( \sigma_0, x_{1t}, x_{2t}\right) \mid z_{jt}  \right]   \right)'= \left( z_{jt}, \E\left[\frac{\partial y_j}{\partial \sigma} \left( \sigma_0, x_{1t}, x_{2t}\right) \mid z_{jt}  \right]   \right)',
\end{align*}
for $j=1,2$. Here we approximate the conditional expectation by averaging over $200$ draws for each grid point of $z_{jt}$ and interpolate the conditional expectation between grid points. For the second GMM estimation, we follow the feasible approximation to the optimal instrument suggested by \citet{berry1999voluntary}

\subsection{Results}
Table \ref{tab:simulation-results-0} summarizes the results under $\lambda=0$. In this case, $z_{jt} \sim F_z$ and thus the scenario is similar to Example \ref{eg:counter}. The kernel and series estimators exhibit small biases and low root mean square errors (RMSEs), indicating good performance in estimating the parameters. In contrast, the GMM estimators show systematically larger biases and RMSEs, suggesting problem of identification of unconditional moments with the optimal instrument, both for the infeasible and the feasible optimal instrument. The empirical rejection rates for the kernel and series estimators are also close to the nominal level, while the GMM estimators display invalid inference at the true parameter values. 

Table \ref{tab:simulation-results-0.3} summarizes the results under $\lambda=0.3$, with $z_{2t}$ having a probability of $0.3$ for taking values from $U[-2,2]$. Although this setup is more different from Example \ref{eg:counter}, larger biases are still displayed in GMM estimators. The kernel and series estimators maintain smaller bias in this case. Under $T=500$ across $50$ Monte Carlo draws, the rejection rate of the kernel estimator displays unstability, but the bias and RMSE is still much smaller than GMM estimators. The series estimator achieves both smaller bias and currect rejection rate.

Table \ref{tab:simulation-results-1} summarizes the results under $\lambda=1$. Here $z_{2t}\sim U[-2,2]$. With the distribution of $z_{jt}$ far from Example \ref{eg:counter}, here the systematically larger biases in GMM estimators disappear. All four estimators display good performance in biases and rejection rates. This serves as an important baseline check, showing the source of the GMM bias is mainly the theoretical property, rather than the computational implementation.

Overall, the simulation results are consistent with our theory for the conditional moment methods, illustrating cases where the unconditional methods are limited while the conditional moment methods still apply.

\section{Conclusion}\label{sec:conclusion}
This paper revisits estimation and inference in the random-coefficients demand model of \cite{berry1995automobile} through its identifying assumption. We begin from the observation that the conditional moment restriction $\E[\xi_{jt}(\theta)\mid z_{jt}]=0$ on which the prevalent identification rests is not equivalent to the unconditional restriction $\E[\xi_{jt}(\theta)\,z_{jt}]=0$ that standard GMM exploits: the unconditional restriction may admit additional parameter values. We give a counterexample in which the model is identified by the conditional restriction yet standard GMM fails to recover $\theta_0$. This gap motivates an estimator built directly on the conditional restriction.

The proposed estimator is a two-step procedure. The first step estimates the relevant conditional expectations nonparametrically; the second chooses the structural parameters to minimize a conditional-variance-weighted sum of squared conditional-moment residuals. Standard GMM arises as the special case in which the conditional expectation is replaced by a linear projection onto a fixed, finite set of instruments.

Our main results establish that this estimator is $\sqrt T$-consistent and asymptotically normal, and we provide two concrete implementations of the first stage—kernel and series—each with primitive conditions under which the limiting distribution holds. Three features of the theory are worth emphasizing. 
First, the two implementations share the same limiting distribution and the same influence function, which depends only on the unobserved product characteristics $\xi_{jt}$. This is not because the first-stage error washes out. With the true conditional expectations the residual is identically zero and the criterion is minimized at $\theta_0$ in every sample, so the first-stage error is in fact the whole source of the limiting variance; what the theory establishes is that it enters only through its projection onto the weighted direction $a_{jt}=\Sigma_j(z_{jt})^{-1}\Gamma_{jt}$. 
Second, because the moment is degenerate at the truth, the conditional expectation may be estimated with a slow nonparametric rate without disturbing the first-order distribution, so a simple plug-in approach suffices. 
Third, with the weight set to the inverse conditional variance, the estimator is efficient for the conditional-moment problem under the additional assumption of within-market independence. The bandwidth and sieve-dimension requirements needed for these conclusions are collected in transparent rate conditions, and the same plug-in variance estimator delivers standard errors for both implementations.

Several extensions are natural. The framework accommodates any first-stage nonparametric estimator satisfying the high-level rate and linearity conditions, so alternatives such as penalized or machine-learning first stages could be analyzed within the same structure. The analysis also takes the market shares as observed; incorporating the sampling and simulation error studied by \cite{freyberger2015asymptotic} into the present conditional-moment framework is left for future work.

\bibliography{bib.bib}

@incollection{berry2021foundations,
  title={Foundations of demand estimation},
  author={Berry, Steven T and Haile, Philip A},
  booktitle={Handbook of industrial organization},
  volume={4},
  number={1},
  pages={1--62},
  year={2021},
  publisher={Elsevier}
}

@article{lavergne2013smooth,
  title={Smooth minimum distance estimation and testing with conditional estimating equations: uniform in bandwidth theory},
  author={Lavergne, Pascal and Patilea, Valentin},
  journal={Journal of Econometrics},
  volume={177},
  number={1},
  pages={47--59},
  year={2013},
  publisher={Elsevier}
}

@article{antoine2014conditional,
  title={Conditional moment models under semi-strong identification},
  author={Antoine, Bertille and Lavergne, Pascal},
  journal={Journal of Econometrics},
  volume={182},
  number={1},
  pages={59--69},
  year={2014},
  publisher={Elsevier}
}

@article{escanciano2018simple,
  title={A simple and robust estimator for linear regression models with strictly exogenous instruments},
  author={Escanciano, Juan Carlos},
  journal={The Econometrics Journal},
  volume={21},
  number={1},
  pages={36--54},
  year={2018},
  publisher={Oxford University Press Oxford, UK}
}

@article{bierens1982consistent,
  title={Consistent model specification tests},
  author={Bierens, Herman J},
  journal={Journal of Econometrics},
  volume={20},
  number={1},
  pages={105--134},
  year={1982},
  publisher={Elsevier}
}

@article{gandhi2019measuring,
  title={Measuring substitution patterns in differentiated-products industries},
  author={Gandhi, Amit and Houde, Jean-Fran{\c{c}}ois},
  journal={NBER Working paper},
  number={w26375},
  year={2019}
}

@article{dominguez2004consistent,
  title={Consistent estimation of models defined by conditional moment restrictions},
  author={Dom{\'\i}nguez, Manuel A and Lobato, Ignacio N},
  journal={Econometrica},
  volume={72},
  number={5},
  pages={1601--1615},
  year={2004},
  publisher={Wiley Online Library}
}

@article{berry1995automobile,
  title={Automobile prices in market equilibrium},
  author={Berry, Steven and Levinsohn, James and Pakes, Ariel},
  journal={Econometrica: Journal of the Econometric Society},
  pages={841--890},
  year={1995},
  publisher={JSTOR}
}

@techreport{salanie2022fast,
  title={Fast, detail-free, and approximately correct: Estimating mixed demand systems},
  author={Salani{\'e}, Bernard and Wolak, Frank A},
  year={2022},
  institution={Working paper}
}

@article{berry2014identification,
  title={Identification in differentiated products markets using market level data},
  author={Berry, Steven T and Haile, Philip A},
  journal={Econometrica},
  volume={82},
  number={5},
  pages={1749--1797},
  year={2014},
  publisher={Wiley Online Library}
}

@article{conlon2020best,
  title={Best practices for differentiated products demand estimation with pyblp},
  author={Conlon, Christopher and Gortmaker, Jeff},
  journal={The RAND Journal of Economics},
  volume={51},
  number={4},
  pages={1108--1161},
  year={2020},
  publisher={Wiley Online Library}
}

@article{reynaert2014improving,
  title={Improving the performance of random coefficients demand models: the role of optimal instruments},
  author={Reynaert, Mathias and Verboven, Frank},
  journal={Journal of econometrics},
  volume={179},
  number={1},
  pages={83--98},
  year={2014},
  publisher={Elsevier}
}

@article{fox2012random,
  title={The random coefficients logit model is identified},
  author={Fox, Jeremy T and il Kim, Kyoo and Ryan, Stephen P and Bajari, Patrick},
  journal={Journal of Econometrics},
  volume={166},
  number={2},
  pages={204--212},
  year={2012},
  publisher={Elsevier}
}

@article{fox2016nonparametric,
  title={Nonparametric identification and estimation of random coefficients in multinomial choice models},
  author={Fox, Jeremy T and Gandhi, Amit},
  journal={The RAND Journal of Economics},
  volume={47},
  number={1},
  pages={118--139},
  year={2016},
  publisher={Wiley Online Library}
}

@article{berry1999voluntary,
  title={Voluntary export restraints on automobiles: Evaluating a trade policy},
  author={Berry, Steven and Levinsohn, James and Pakes, Ariel},
  journal={American Economic Review},
  volume={89},
  number={3},
  pages={400--431},
  year={1999},
  publisher={American Economic Association}
}

@techreport{conlon2023incorporating,
  title={Incorporating Micro Data into Differentiated Products Demand Estimation with PyBLP},
  author={Conlon, Christopher and Gortmaker, Jeff},
  year={2023},
  institution={National Bureau of Economic Research}
}

@article{lee2015computationally,
  title={A computationally fast estimator for random coefficients logit demand models using aggregate data},
  author={Lee, Jinhyuk and Seo, Kyoungwon},
  journal={The RAND Journal of Economics},
  volume={46},
  number={1},
  pages={86--102},
  year={2015},
  publisher={Wiley Online Library}
}

@article{chamberlain1987asymptotic,
  title={Asymptotic efficiency in estimation with conditional moment restrictions},
  author={Chamberlain, Gary},
  journal={Journal of econometrics},
  volume={34},
  number={3},
  pages={305--334},
  year={1987},
  publisher={Elsevier}
}

@article{donald2009choosing,
  title={Choosing instrumental variables in conditional moment restriction models},
  author={Donald, Stephen G and Imbens, Guido W and Newey, Whitney K},
  journal={Journal of Econometrics},
  volume={152},
  number={1},
  pages={28--36},
  year={2009},
  publisher={Elsevier}
}

@article{dube2012improving,
  title={Improving the numerical performance of static and dynamic aggregate discrete choice random coefficients demand estimation},
  author={Dub{\'e}, Jean-Pierre and Fox, Jeremy T and Su, Che-Lin},
  journal={Econometrica},
  volume={80},
  number={5},
  pages={2231--2267},
  year={2012},
  publisher={Wiley Online Library}
}

@article{ai2003efficient,
  title={Efficient estimation of models with conditional moment restrictions containing unknown functions},
  author={Ai, Chunrong and Chen, Xiaohong},
  journal={Econometrica},
  volume={71},
  number={6},
  pages={1795--1843},
  year={2003},
  publisher={Wiley Online Library}
}

@book{li2023nonparametric,
  title={Nonparametric econometrics: theory and practice},
  author={Li, Qi and Racine, Jeffrey Scott},
  year={2023},
  publisher={Princeton University Press}
}

@article{newey1994large,
  title={Large sample estimation and hypothesis testing},
  author={Newey, Whitney K and McFadden, Daniel},
  journal={Handbook of econometrics},
  volume={4},
  pages={2111--2245},
  year={1994},
  publisher={Elsevier}
}

@article{newey1997convergence,
  title={Convergence rates and asymptotic normality for series estimators},
  author={Newey, Whitney K},
  journal={Journal of econometrics},
  volume={79},
  number={1},
  pages={147--168},
  year={1997},
  publisher={Elsevier}
}

@article{lu2023semi,
  title={Semi-nonparametric estimation of random coefficients logit model for aggregate demand},
  author={Lu, Zhentong and Shi, Xiaoxia and Tao, Jing},
  journal={Journal of Econometrics},
  volume={235},
  number={2},
  pages={2245--2265},
  year={2023},
  publisher={Elsevier}
}

@article{nevo2000mergers,
  title={Mergers with differentiated products: The case of the ready-to-eat cereal industry},
  author={Nevo, Aviv},
  journal={The Rand journal of economics},
  pages={395--421},
  year={2000},
  publisher={JSTOR}
}

@article{nevo2001measuring,
  title={Measuring market power in the ready-to-eat cereal industry},
  author={Nevo, Aviv},
  journal={Econometrica},
  volume={69},
  number={2},
  pages={307--342},
  year={2001},
  publisher={Wiley Online Library}
}

@article{petrin2002quantifying,
  title={Quantifying the benefits of new products: The case of the minivan},
  author={Petrin, Amil},
  journal={Journal of political Economy},
  volume={110},
  number={4},
  pages={705--729},
  year={2002},
  publisher={The University of Chicago Press}
}

@article{goldberg2001evolution,
  title={The evolution of price dispersion in the European car market},
  author={Goldberg, Pinelopi Koujianou and Verboven, Frank},
  journal={The Review of Economic Studies},
  volume={68},
  number={4},
  pages={811--848},
  year={2001},
  publisher={Wiley-Blackwell}
}

@article{wang2023sieve,
  title={Sieve BLP: A semi-nonparametric model of demand for differentiated products},
  author={Wang, Ao},
  journal={Journal of Econometrics},
  volume={235},
  number={2},
  pages={325--351},
  year={2023},
  publisher={Elsevier}
}

@ARTICLE{berry1994estimating,
title = {Estimating Discrete-Choice Models of Product Differentiation},
author = {Berry, Steven},
year = {1994},
journal = {RAND Journal of Economics},
volume = {25},
number = {2},
pages = {242-262},
url = {https://EconPapers.repec.org/RePEc:rje:randje:v:25:y:1994:i:summer:p:242-262}
}

@article{freyberger2015asymptotic,
  title={Asymptotic theory for differentiated products demand models with many markets},
  author={Freyberger, Joachim},
  journal={Journal of Econometrics},
  volume={185},
  number={1},
  pages={162--181},
  year={2015},
  publisher={Elsevier}
}

@article{tropp2015introduction,
  title={An introduction to matrix concentration inequalities},
  author={Tropp, Joel A},
  journal={Foundations and trends{\textregistered} in machine learning},
  volume={8},
  number={1-2},
  pages={1--230},
  year={2015},
  publisher={Emerald Publishing Limited}
}

\begin{appendices}
\section{Lemmas}

The following lemma is a matrix version of Bernstein's inequality, which implies a more usual version of Bernstein's inequality for scalar random variables.
\begin{lemma}[Matrix Bernstein's Inequality, Theorem 6.1.1 of \cite{tropp2015introduction}]\label{lem:bernstein}
Let $X_1, X_2, \ldots, X_n$ be independent random matrices of dimension $d_1 \times d_2$ satisfying
\[
    \mathbb{E}[X_i] = 0 \quad \text{and} \quad \|X_i\|_{op} \leq M \quad \text{almost surely, for all } i.
\]
Define
\[
\sigma^2 := \max\left\{ \left\| \sum_{i=1}^n \mathbb{E}[X_i X_i'] \right\|_{op}, \ \left\| \sum_{i=1}^n \mathbb{E}[X_i' X_i] \right\|_{op} \right\}.
\]
Then for any $t > 0$,
\[
    \mathbb{P}\!\left(\|\sum_{i=1}^n X_i\|_{op} \geq t\right)
    \;\leq\;
    (d_1+d_2)\exp\!\left(
        -\,\frac{\,t^2/2\,}{\sigma^2 + M t/3}
    \right).
\]
\end{lemma}

\begin{lemma}[Uniform convergence rate]\label{lemma:kernel}
Suppose Conditions \ref{cond:regularity}, and \ref{cond:kernel} hold, and let $\mathcal Z_T$ be as in \eqref{eq:kernel-trim}. Then for each $j$ and $\psi\in\{y_j(\sigma,w_t),\partial_\sigma y_j(\sigma,w_t),\partial^2_\sigma y_j(\sigma,w_t),p_{jt},\xi_{jt}^2\}$,
\[
    \sup_{\sigma\in\Theta_\sigma,\ z\in\mathcal Z_T}\big|\hat m_{\psi,j}(z;\sigma)-m_{\psi,j}(z;\sigma)\big|
    =O_p(\rho_T),
    \qquad
    \sup_{z\in\mathcal Z_T}\big|\hat f_j(z)-f_j(z)\big|=O_p(\rho_T),
\]
where $\rho_T=\big(\tfrac{\log T}{T h^{d_z}}\big)^{1/2}+h^{l}$. Consequently, since $\inf_{\mathcal Z}f_j\ge\underline f>0$,
\[
    \sup_{\sigma\in\Theta_\sigma,\ z\in\mathcal Z_T}\big|\hat\mu_{\psi,j}(z;\sigma)-\mu_{\psi,j}(z;\sigma)\big|=O_p(\rho_T).
\]
Since $h=T^{-q}$ with $\tfrac{1}{2l}<q<\tfrac{1}{2d_z}$, $\rho_T=o(T^{-1/4})$.
\end{lemma}
 
\begin{proof}[Proof of Lemma~\ref{lemma:kernel}]
The argument is a market-indexed version of the uniform kernel bound in \citet[Thm.~1.4]{li2023nonparametric}.
We prove the rate for $\hat m_{\psi,j}$; the density is the special case $\psi\equiv1$, and the ratio bound follows from
\[
    \hat\mu_{\psi,j}-\mu_{\psi,j}
    =\frac{\hat m_{\psi,j}-m_{\psi,j}}{\hat f_j}-\mu_{\psi,j}\frac{\hat f_j-f_j}{\hat f_j},
\]
together with $\hat f_j\ge\underline f/2$ w.p.a.1 uniformly on $\mathcal Z_T$ (from the density rate on $\mathcal Z_T$ and $\inf_{\mathcal Z} f_j\ge\underline f$).
 
Split into a bias and a variance part,
\[
    \sup_{\sigma,\,z\in\mathcal Z_T}\big|\hat m_{\psi,j}(z;\sigma)-m_{\psi,j}(z;\sigma)\big|
    \le\underbrace{\sup_{\sigma,\,z\in\mathcal Z}\big|\hat m_{\psi,j}(z;\sigma)-\E\hat m_{\psi,j}(z;\sigma)\big|}_{(\mathrm{I})}
    +\underbrace{\sup_{\sigma,\,z\in\mathcal Z_T}\big|\E\hat m_{\psi,j}(z;\sigma)-m_{\psi,j}(z;\sigma)\big|}_{(\mathrm{II})}.
\]
 
\textbf{\emph{Bias $(\mathrm{II})$.}} Fix $z\in\mathcal Z_T$. Since $f_j$ vanishes off $\mathcal Z$, so does $m_{\psi,j}=\mu_{\psi,j}f_j$, and the integral below is over $\R^{d_z}$ or over $\mathcal Z$ indifferently. Exploiting the $l$-th order kernel,
\begin{align*}
    \E\hat m_{\psi,j}(z;\sigma) = & \E \psi_{jt}(\sigma)\,K_h(z_{jt}-z) \\
    = & \E \mu_{\psi,j}(z_{jt};\sigma)\,K_h(z_{jt}-z)\\
    =&\int K_h(s-z)\,m_{\psi,j}(s;\sigma)\,ds\\
    =&\int K(v)\,m_{\psi,j}(z+vh;\sigma)\,dv\\
    =&m_{\psi,j}(z;\sigma)+O(h^{l}),
\end{align*}
uniformly in $(\sigma,z)\in\Theta_\sigma\times\mathcal Z_T$, using the $l$-th order Taylor expansion of $m_{\psi,j}$ and the bound on its $l$-th derivatives. The expansion is available precisely because $z\in\mathcal Z_T$: by \eqref{eq:window-inside} every point $z+vh$ with $v\in\operatorname{supp}K$ lies in $\mathcal Z$, so the change of variables integrates the full kernel---$\int K(v)\,dv=1$ and $\int K(v)v^m\,dv=0$ for $1\le m<l$---against a function that is $l$-times continuously differentiable on the whole window by Condition~\ref{cond:kernel}\eqref{cond:kernel-psi}. Hence $(\mathrm{II})=O(h^{l})$ on $\mathcal Z_T$.

\emph{Variance part $(\mathrm{I})$.} This part requires no trimming and is proved uniformly over $\mathcal Z$. Let $\tilde{\mathcal S}=\Theta_\sigma\times\mathcal Z$, a compact subset of $\R^{d_\sigma+d_z}$, and cover it by $L_T$ cubes $\{I_{k}\}_{k\le L_T}$ of side $\ell_T$ centered at $(\sigma_k,z_k)$, so $L_T\asymp \ell_T^{-(d_\sigma+d_z)}$. Decompose
\begin{align*}
    (\mathrm{I})\le& 
    \underbrace{\max_k\sup_{(\sigma,z)\in I_k}\big|\hat m_{\psi,j}(z;\sigma)-\hat m_{\psi,j}(z_k;\sigma_k)\big|}_{Q_1}
    +\underbrace{\max_k\big|\hat m_{\psi,j}(z_k;\sigma_k)-\E\hat m_{\psi,j}(z_k;\sigma_k)\big|}_{Q_2}\\
    &+\underbrace{\max_k\sup_{(\sigma,z)\in I_k}\big|\E\hat m_{\psi,j}(z;\sigma)-\E\hat m_{\psi,j}(z_k;\sigma_k)\big|}_{Q_3}.
\end{align*}
Because $K$ has compact support in $\R^{d_z}$ and is Lipschitz and $\psi_{jt}(\sigma)$ is Lipschitz in $\Theta_\sigma$, $K_h(\cdot-z)\psi(\sigma)$ is Lipschitz in $(\sigma,z)$ with constant $O(h^{-d_z-1})$, so both discretization terms satisfy $Q_1,Q_3=O\big(\ell_T\,h^{-d_z-1}\big)$ almost surely.

For $Q_2$, fix $(\sigma_k,z_k)$ and write $W_T=\hat m_{\psi,j}(z_k;\sigma_k)-\E\hat m_{\psi,j}(z_k;\sigma_k)=\sum_{t=1}^{T}\zeta_t$ with
\[
    \zeta_t=\tfrac1T\big\{K_h(z_{jt}-z_k)\psi_{jt}(\sigma_k)-\E K_h(z_{jt}-z_k)\psi_{jt}(\sigma_k)\big\}.
\]
The summands are independent across markets (Condition~\ref{cond:regularity}\eqref{cond:independence}), mean zero, bounded by $|\zeta_t|\le 2A_1/(Th^{d_z})$ for $A_1=\|\psi\|_\infty\|K\|_\infty$, and with variance
\begin{align*}
    \E\zeta_t^2\le \tfrac1{T^2}\,\E\big[K_h(z_{jt}-z_k)^2\psi_{jt}(\sigma_k)^2\big]
    \le & \tfrac1{T^2}\,\|\psi\|_\infty^2\E\big[K_h(z_{jt}-z_k)^2\big]\\
    = & \tfrac1{T^2}\,\|\psi\|_\infty^2\int\tfrac1{h^{d_z}}K(v)^2f_j(vh+z_k)dv\\
    \le & \tfrac{A_2}{T^2 h^{d_z}},\qquad
    A_2=\|\psi\|_\infty^2\Big(\sup_{\mathcal Z}f_j\Big)\!\int K(v)^2dv .
\end{align*}
Note that this bound uses only $f_j\le\bar f$, which holds on all of $\mathcal Z$, so no trimming is needed here.

By Bernstein's inequality (Lemma \ref{lem:bernstein}), for any $\eta_T>0$,
\[
    P\big(|W_T|\ge\eta_T\big)\le 2\exp\!\Big(-\tfrac{\tfrac12\eta_T^2}{A_2/(Th^{d_z})+\tfrac{2A_1}{3Th^{d_z}}\eta_T}\Big)
    = 2\exp\!\Big(-\frac{1}{2(A_2+\frac{2}{3}A_1\eta_T)}\,\eta_T^2 Th^{d_z}\Big)
\]
and a union bound over the $L_T$ centers gives $P(Q_2\ge\eta_T)\le 2L_T\exp(-C\,\eta_T^2 Th^{d_z})$ for $\eta_T$ bounded. Setting $\eta_T=C_4\big(\tfrac{\log T}{Th^{d_z}}\big)^{1/2}$, a bounded sequence, and choosing $C_4$ large,
\[
    P(Q_2\ge\eta_T)\le 2L_T\,T^{-C_4^2 C}\longrightarrow0
\]
Summarizing, we have
\[
    Q_1+Q_2+Q_3 \lesssim \ell_T\,h^{-d_z-1}+C_4\big(\tfrac{\log T}{Th^{d_z}}\big)^{1/2}\quad \text{with probability at least } 1-2L_T\,T^{-C_4^2 C}
\]
Finally take $\ell_T=(\log T)^{1/2}h^{d_z/2+1}T^{-1/2}$, which makes $\ell_T\,h^{-d_z-1}=(\log T/(Th^{d_z}))^{1/2}$ as well while keeping $\log L_T\asymp \log \ell_T^{-1}\asymp \log T-(\frac{d_z}{2}+1)\log h\asymp \log T$, because $h\asymp T^{-q}$. 
Hence $2L_T\,T^{-C_4^2 C}\to0$.
Combining, $(\mathrm{I})=O_p\big((\log T/(Th^{d_z}))^{1/2}\big)$ uniformly on $\Theta_\sigma\times\mathcal Z$, and with $(\mathrm{II})=O(h^l)$ on $\Theta_\sigma\times\mathcal Z_T$ we obtain the stated $\rho_T$.
\end{proof}

\begin{lemma}[Kernel influence function]\label{lemma:kernel-linear}
Let $a(z)=\Lambda_j(z)\Gamma_j(z)$, where $\Gamma_j(z)=\big(\dot\mu_{y,j}(z;\sigma_0),-\mu_{p,j}(z),-x(z)\big)'$ and $\Lambda_{j}$ satisfies $0<\underline c<\Lambda_{j}(z)<\bar c$ for all $z\in\mathcal Z$. Suppose $a$ is bounded, $\sup_{z\in\mathcal Z}\|a(z)\|\le\bar C$, and continuous on the interior of $\mathcal Z$; no derivatives of $a$ are required.
Under Conditions~\ref{cond:regularity}, \ref{cond:y-smooth} and \ref{cond:kernel}, define $e^{\psi}_{jt}=\psi_{jt}(\sigma_0)-\mu_{\psi,j}(z_{jt})$, with the kernel estimator $\hat\mu_{\psi,j}$ evaluated at $\sigma_0$. Then, for any $\psi_{jt}(\sigma)\in\{y_j(\sigma,w_t),\ \partial_\sigma y_j(\sigma,w_t),\ \partial^2_\sigma y_j(\sigma,w_t),\ p_{jt}\}$,
\begin{equation}\label{eq:kernel-proj}
    \E_T\big[\big(\hat\mu_{\psi,j}(z_{jt};\sigma_0)-\mu_{\psi,j}(z_{jt};\sigma_0)\big)\,\tau_T(z_{jt})\,a(z_{jt})\big]
    =\E_T\,\big[e^{\psi}_{jt}\,a(z_{jt})\big]+o_p(T^{-1/2}),
\end{equation}
where all objects are evaluated at $\sigma_0$. Note that the left-hand side is trimmed and the right-hand side is not: the smoothing of the weight by the kernel and the deletion of the boundary collar both wash out at the $\sqrt T$ scale.
\end{lemma}

\begin{proof}[Proof of Lemma~\ref{lemma:kernel-linear}]
Throughout, $\tau_T(z_{jt})=1$ forces $z_{jt}\in\mathcal Z_T$, so every uniform bound below is needed only on $\mathcal Z_T$, which is exactly what Lemma~\ref{lemma:kernel} supplies.
 
By Lemma~\ref{lemma:kernel} and $\hat f_j\ge\underline f/2$ w.p.a.1 on $\mathcal Z_T$,
\begin{align*}
    \hat\mu_{\psi,j}-\mu_{\psi,j}
    =&\frac{1}{f_j}\Big[(\hat m_{\psi,j}-m_{\psi,j})-\mu_{\psi,j}(\hat f_j-f_j)\Big]\\
    &-\underbrace{\frac{1}{f_j}\Big[(\hat\mu_{\psi,j}-\mu_{\psi,j})(\hat{f}_j-f_j)\Big]}_{R_{jt}}\qquad \sup_{\Theta_\sigma,\ \mathcal Z_T}|R_{jt}|=O_p(\rho_T^2)=o_p(T^{-1/2}),
\end{align*}
so the remainder $R_{jt}$ is negligible after averaging against the bounded $\tau_T\,a$. Evaluating the linear part at $z=z_{jt}$ and using $m_{\psi,j}=\mu_{\psi,j}f_j$,
\[
    (\hat m_{\psi,j}-m_{\psi,j})-\mu_{\psi,j}(\hat f_j-f_j)\Big|_{z=z_{jt}}=(\hat m_{\psi,j}-\mu_{\psi,j}\hat f_j)\Big|_{z=z_{jt}}
    =\E_T^{s}\big[K_h(z_{js}-z_{jt})\big(\psi_{js}-\mu_{\psi,j}(z_{jt})\big)\big],
\]
where $\E_T^s=\frac{1}{T}\sum_{s=1}^T$ averages over the smoothing index $s=1,\dots,T$. 

Hence the linear part of the left side of \eqref{eq:kernel-proj} is the double sum
\[
    U_T=\frac{1}{T^2}\sum_{t=1}^{T}\sum_{s=1}^{T}\,\underbrace{\frac{\tau_T(z_{jt})\,a(z_{jt})}{f_j(z_{jt})}\,K_h(z_{js}-z_{jt})\,\big(\psi_{js}-\mu_{\psi,j}(z_{jt})\big)}_{=:~c_T(w_{jt},w_{js})},
\]
where $w_{jt}=(z_{jt},\psi_{jt})$. This is a $V$-statistic of order two: each summand depends on a pair of observations $(t,s)$, the outer index $t$ entering through the evaluation point $\tau_T(z_{jt})a(z_{jt})/f_j(z_{jt})$, $\mu_{\psi,j}(z_{jt})$, and the smoothing kernel $K_h(z_{js}-z_{jt})$, and the inner index $s$ through the regressand $\psi_{js}$, and $K_h(z_{js}-z_{jt})$. The trimming indicator rides along on the outer variable only. The kernel $c_T$ is not symmetric in $(t,s)$ and changes with $T$ through $h$ and through $\tau_T$, but the projection argument below applies unchanged.

\emph{Idea of the projection.}
A second-order $V$-statistic $\tfrac{1}{T^2}\sum_{t,s}c_T(w_{jt},w_{js})$ is, to first order, equal to the sum of its one-observation conditional averages,
\begin{equation}\label{eq:hajek}
    U_T=\bar c_T+\E_T\big[c_{1,T}(w_{jt})\big]+\E_T\big[c_{2,T}(w_{js})\big]-2\bar c_T+(\text{degenerate remainder}),
\end{equation}
where $\bar c_T=\E[c_T(w_{jt},w_{js})]$ (both arguments independent), and the two \emph{Hájek projections}
\[
    c_{1,T}(w)=\E[c_T(w,w_{js})],\qquad c_{2,T}(w)=\E[c_T(w_{jt},w)]
\]
average out one argument while holding the other fixed. The remainder collects the part of $c_T$ orthogonal to every single observation; we will show at last that it is $o_p(T^{-1/2})$, so $U_T$ is asymptotically a plain sample average. We now compute the two projections.

\emph{Projection fixing the inner observation $w_{js}=(z_{js},\psi_{js})$.}
Averaging over the outer observation $z_{jt}$ (which is independent of $w_{js}$ and has density $f_j$),
\[
    c_{2,T}(w_{js})
    =\int \frac{\tau_T(z)a(z)}{f_j(z)}\,K_h(z_{js}-z)\,\big(\psi_{js}-\mu_{\psi,j}(z)\big)\,f_j(z)\,dz
    =\int \tau_T(z)\,a(z)\,K_h(z_{js}-z)\,\big(\psi_{js}-\mu_{\psi,j}(z)\big)\,dz,
\]
the density cancelling. Insert $\psi_{js}=\mu_{\psi,j}(z_{js})+e^{\psi}_{js}$ and split into a piece carrying the projection residual and a bias piece,
\[
    c_{2,T}(w_{js})=\tilde a_T(z_{js})\,e^{\psi}_{js}+B_T(z_{js}),
\]
\[
    \tilde a_T(y):=\int \tau_T(z)\,a(z)\,K_h(y-z)\,dz,
    \qquad
    B_T(y):=\int \tau_T(z)\,a(z)\,K_h(y-z)\big(\mu_{\psi,j}(y)-\mu_{\psi,j}(z)\big)\,dz .
\]
Here $\tilde a_T$ is the weight $\tau_T a$ smoothed by the kernel; it is a deterministic function. We show
\begin{equation}\label{eq:c2-pieces}
    \text{(a)}\quad \E_T\big[(\tilde a_T-a)(z_{js})\,e^{\psi}_{js}\big]=o_p(T^{-1/2}),
    \qquad
    \text{(b)}\quad \E_T\big[B_T(z_{js})\big]=o_p(T^{-1/2}),
\end{equation}
which together give $\E_T[c_{2,T}(w_{js})]=\E_T[a(z_{js})e^{\psi}_{js}]+o_p(T^{-1/2})$.

\emph{(a) The residual piece.} The summands $(\tilde a_T-a)(z_{js})e^{\psi}_{js}$ are i.i.d.\ across $s$; since $\tilde a_T-a$ is a deterministic function of $z$ and $\E[e^{\psi}_{js}\mid z_{js}]=0$, each has mean zero. As $e^{\psi}$ is bounded,
\[
    \operatorname{Var}\Big(\E_T\big[(\tilde a_T-a)\,e^{\psi}\big]\Big)
    =\frac1T\,\E\big[(\tilde a_T-a)^2 (e^{\psi})^2\big]
    \le \frac{C}{T}\,\big\|\tilde a_T-a\big\|_{L^2}^2 ,
\]
so all that is required is $\|\tilde a_T-a\|_{L^2}=o(1)$---no rate. This holds by dominated convergence. The family is uniformly bounded, 
\begin{align*}
    \|\tilde a_T\|_\infty=\|\int K(u)\,\tau_T(y-hu)\,a(y-hu)\,du\|_\infty\le\bar C\int|K(u)|du,
\end{align*} and for a.e.\ $y$ in the interior of $\mathcal Z$ (the boundary is Lebesgue-null by Condition~\ref{cond:kernel}\eqref{cond:kernel-boundedness}) we have $d(y,\partial\mathcal Z)>0$, so for all $T$ large enough $\tau_T(y-hu)=1$ for every $u\in\operatorname{supp}K$, whence
\[
    \tilde a_T(y)=\int K(u)\,\tau_T(y-hu)\,a(y-hu)\,du=\int K(u)\,a(y-hu)\,du\longrightarrow a(y)
\]
by continuity of $a$ at $y$ and $\int K(u)\,du=1$. Chebyshev's inequality then gives (a).
 
\emph{(b) The bias piece.} $\E_T[B_T(z_{js})]$ is an i.i.d.\ average, so $\E_T[B_T]=\E[B_T]+O_p\big(\sqrt{\E[B_T^2]/T}\big)$.
 
For the mean, Fubini and the change of variables $y=z+hu$ give
\[
    \E[B_T]=\int \tau_T(z)\,a(z)\left[\int K(u)\big(\mu_{\psi,j}(z+hu)-\mu_{\psi,j}(z)\big)f_j(z+hu)\,du\right]dz .
\]
For every $z$ with $\tau_T(z)=1$ the window $B(z,Rh)$ lies inside $\mathcal Z$ by \eqref{eq:window-inside}, so
$G_z(v):=\big(\mu_{\psi,j}(z+v)-\mu_{\psi,j}(z)\big)f_j(z+v)=m_{\psi,j}(z+v)-\mu_{\psi,j}(z)f_j(z+v)$
is $l$-times continuously differentiable in $v$ on that window with derivatives bounded uniformly in $z$ (Condition~\ref{cond:kernel}\eqref{cond:kernel-psi}), and $G_z(0)=0$. The order-$l$ kernel annihilates the Taylor terms of degrees $1,\dots,l-1$, so the bracket is $O(h^{l})$ uniformly over $\{\tau_T=1\}$, and since $a$ is bounded,
\[
    \E[B_T]=O(h^{l})=o(T^{-1/2}),
\]
the last step by the undersmoothing $q>\tfrac1{2l}$ in Condition~\ref{cond:kernel}\eqref{cond:kernel-smoothness-bw}. Only boundedness of $a$ and smoothness of $m_{\psi,j}$ and $f_j$ are used; the indicator is never differentiated, because in this Fubini form it sits on the outer variable.
 
For the fluctuation a crude bound suffices. Uniformly in $y$, using only the Lipschitz continuity of $\mu_{\psi,j}$ on $\mathcal Z$ (which follows from Condition~\ref{cond:kernel}\eqref{cond:kernel-psi} and $f_j\ge\underline f$),
\[
    |B_T(y)|\le \bar C\,L_\mu\int|K_h(y-z)|\,\|y-z\|\,dz
    =\bar C\,L_\mu\,h\int|K(u)|\,\|u\|\,du=O(h),
\]
so $\E[B_T^2]=O(h^2)$ and $\E_T[B_T]-\E[B_T]=O_p(h\,T^{-1/2})=o_p(T^{-1/2})$. This closes (b), and hence
\[
    \E_T\big[c_{2,T}(w_{js})\big]=\E_T\big[a(z_{js})\,e^{\psi}_{js}\big]+o_p(T^{-1/2}).
\]

\emph{Projection fixing the outer observation $w_{jt}=(z_{jt},\psi_{jt})$.}
Averaging over the inner observation, and using $\E[\psi_{js}\mid z_{js}]=\mu_{\psi,j}(z_{js})$,
\begin{align*}
    c_{1,T}(w_{jt})
    =&\frac{\tau_T(z_{jt})\,a(z_{jt})}{f_j(z_{jt})}\int K_h(z-z_{jt})\,\big(\mu_{\psi,j}(z)-\mu_{\psi,j}(z_{jt})\big)\,f_j(z)\,dz\\
    =&\frac{\tau_T(z_{jt})\,a(z_{jt})}{f_j(z_{jt})}\int K(u)\,\big(\mu_{\psi,j}(z_{jt}+hu)-\mu_{\psi,j}(z_{jt})\big)\,f_j(z_{jt}+hu)\,du
    =O(h^l),
\end{align*}
uniformly: when $\tau_T(z_{jt})=0$ the expression vanishes, and when $\tau_T(z_{jt})=1$ the window lies inside $\mathcal Z$ by \eqref{eq:window-inside}, the integrand vanishes at $u=0$, and the $l$-th order kernel kills the lower-order Taylor terms, exactly as in the display for $\E[B_T]$. Hence $\E_T[c_{1,T}(w_{jt})]=O(h^{l})=o(T^{-1/2})$ deterministically, and the constant $\bar c_T=\E[c_{1,T}(w_{jt})]=\E[c_{2,T}(w_{js})]=\E[B_T]$ is likewise $O(h^l)$.
 
\emph{Assembling the projection.}
Substituting the two projections into \eqref{eq:hajek}, the $c_{1,T}$ term and the centering constants are $O(h^l)=o(T^{-1/2})$, so
\[
    U_T=\E_T\big[c_{2,T}(w_{js})\big]+o(T^{-1/2})+(\text{remainder})
    =\E_T\big[a(z_{js})\,e^{\psi}_{js}\big]+o_p(T^{-1/2})+(\text{remainder}).
\]

\emph{Negligibility of the remainder.}
Recall from \eqref{eq:hajek} that the remainder is $U_T$ minus its Hájek projection,
\[
    R_T:=U_T-\Big(\bar c_T+\E_T[c_{1,T}(w_{jt})]+\E_T[c_{2,T}(w_{js})]-2\bar c_T\Big).
\]
Both terms below use only $|\tau_T|\le1$ and the boundedness of $a/f_j$, so the argument is that of the untrimmed case. Separate the diagonal and off-diagonal parts of the double sum defining $U_T$:
\[
    U_T=\frac{1}{T^2}\sum_{t\ne s}c_T(w_{jt},w_{js})+\frac{1}{T^2}\sum_{t=1}^{T}c_T(w_{jt},w_{jt})
    =:U_T^{\mathrm{off}}+U_T^{\mathrm{diag}}.
\]
 
\emph{Diagonal term.}
Since $\tau_T a/f_j$, $\psi$, and $\mu_{\psi,j}$ are bounded and $K_h(0)=h^{-d_z}K(0)$,
\[
    \big|U_T^{\mathrm{diag}}\big|
    =\Big|\frac1{T^2}\sum_{t}\frac{\tau_T(z_{jt})a(z_{jt})}{f_j(z_{jt})}K_h(0)\big(\psi_{jt}-\mu_{\psi,j}(z_{jt})\big)\Big|
    \le \frac{C}{T^2}\cdot T\cdot h^{-d_z}K(0)
    =O\!\Big(\frac{1}{Th^{d_z}}\Big)=o(T^{-1/2}).
\]
where $(Th^{d_z})^{-1}=o(T^{-1/2})$. Hence $U_T^{\mathrm{diag}}=o_p(T^{-1/2})$.

\emph{Off-diagonal term.}
Write $\mathring c_T=c_T-c_{1,T}-c_{2,T}+\bar c_T$ for the doubly-centered kernel, which satisfies, for $t\ne s$, $\E[c_{2,T}(w_{js})|w_{jt}]=\E[c_{2,T}(w_{js})]=\bar{c}_T$ and $\E[c_{1,T}(w_{jt})|w_{js}]=\E[c_{1,T}(w_{jt})]=\bar c_T$, so
\[
    \E\big[\mathring c_T(w_{jt},w_{js})\mid w_{jt}\big]=0
    \quad\text{and}\quad
    \E\big[\mathring c_T(w_{jt},w_{js})\mid w_{js}\big]=0.
\]
The degenerate ($U$-statistic) part of the off-diagonal sum is
\[
    \mathring U_T=\frac{1}{T^2}\sum_{t\ne s}\mathring c_T(w_{jt},w_{js}),
\]
and 
$$R_T=\mathring U_T+U_T^{\mathrm{diag}}-\frac{1}{T}\Big(\E_T[c_{1,T}(w_{jt})]+\E_T[c_{2,T}(w_{js})]-\bar c_T\Big)=\mathring U_T+U_T^{\mathrm{diag}}+\frac{1}{T}O_p(T^{-1/2}+h^l).$$ 
It remains to bound $\mathring U_T$. Because the summands are doubly centered, all cross-terms in the second moment vanish: for $(t,s)\ne(t',s')$ with $t\ne s$, $t'\ne s'$, at least one index is unmatched, and conditioning on the other three observations gives $\E[\mathring c_T(w_{jt},w_{js})\mathring c_T(w_{jt'},w_{js'})]=0$. 
For example, if $t=t'$ but $s\ne s'$, then 
\begin{align*}
    \E[\mathring c_T(w_{jt},w_{js})\mathring c_T(w_{jt},w_{js'})]=&\E[\E[\mathring c_T(w_{jt},w_{js})\mathring c_T(w_{jt},w_{js'})|w_{jt},w_{js}]]\\
    =&\E[\mathring c_T(w_{jt},w_{js})\E[\mathring c_T(w_{jt},w_{js'})|w_{jt},w_{js}]]\\
    =&\E[\mathring c_T(w_{jt},w_{js})\E[\mathring c_T(w_{jt},w_{js'})|w_{jt}]] = 0
\end{align*} 
Only the terms with $\{t,s\}=\{t',s'\}$ survive, of which there are $O(T^2)$, so
\[
    \E\big[\mathring U_T^2\big]
    =\frac{1}{T^4}\sum_{t\ne s}\E\big[\mathring c_T(w_{jt},w_{js})^2\big]
    +\frac{1}{T^4}\sum_{t\ne s}\E\big[\mathring c_T(w_{jt},w_{js})\mathring c_T(w_{js},w_{jt})\big]
    \le \frac{2}{T^2}\,\E\big[\mathring c_T(w_{jt},w_{js})^2\big],
\]
where in the second sum we used Cauchy--Schwarz $\mathring c_T(w_{jt},w_{js})\mathring c_T(w_{js},w_{jt})\leq \frac{1}{2}(\mathring c_T(w_{jt},w_{js})^2+\mathring c_T(w_{js},w_{jt})^2)$. Now $\mathring c_T$ is a bounded linear combination of $c_T$ and its projections, and the projections are $O(1)$ uniformly,
\begin{align*}
    \bar c_{T}=O(h^l),\qquad \sup_{w}\|c_{1,T}(w)\|=O(h^l),\qquad \E c_{2,T}(w_{js})^2\le 2\,\E[\|\tilde a_T\|_{\infty}^2\,(e^{\psi}_{js})^2]+ 2\E[B_T(z_{js})^2]\le \bar{C}+O(h^2),
\end{align*}
by the previous steps,
so $\E[\mathring c_T^2]\le C\,\E[c_T(w_{jt},w_{js})^2]+O(1)$. For the leading term, boundedness of $\tau_T a/f_j$ and $e^{\psi}$ gives
\[
    \E\big[c_T(w_{jt},w_{js})^2\big]
    \le C\,\E\big[K_h(z_{js}-z_{jt})^2\big]
    = C\int\!\!\int K_h(z'-z)^2 f_j(z)f_j(z')\,dz\,dz'
    = C\,h^{-d_z}\!\int K(u)^2\Big(\int f_j(z)f_j(z+hu)\,dz\Big)du,
\]
and the inner integral is bounded by $\sup_z f_j(z)<\infty$, so $\E[c_T^2]=O(h^{-d_z})$. Combining,
\[
    \E\big[\mathring U_T^2\big]\le \frac{2}{T^2}\,O(h^{-d_z})=O\!\Big(\frac{1}{T^2 h^{d_z}}\Big)=o(T^{-1}),
\]
the last step again because $Th^{d_z}\to\infty$. By Chebyshev, $\mathring U_T=O_p\big((T^2h^{d_z})^{-1/2}\big)=o_p(T^{-1/2})$.

\emph{Collecting.}
Therefore $R_T=\mathring U_T+U_T^{\mathrm{diag}}+o_p(T^{-1})=o_p(T^{-1/2})$, so from \eqref{eq:hajek} and the projections computed above,
\[
    U_T=\E_T\big[c_{2,T}(w_{js})\big]+O(h^{l})+R_T
    =\E_T\big[a(z_{js})\,e^{\psi}_{js}\big]+o_p(T^{-1/2}),
\]
which, after relabeling the summation index $s$ as $t$, is exactly \eqref{eq:kernel-proj}.
\end{proof}

\begin{lemma}[Rate for Kernel plugged-in weight]\label{lemma:kernel-weight}
Let $\hat\Sigma_{j}(z)=\dfrac{\E_T\!\big[K_h(z_{jt}-z)\,\hat\xi_{jt}^{\,2}\big]}{\hat f_j(z)}$, $\hat\xi_{jt}=y_{j}(\tilde\sigma,w_t)-\tilde\alpha\,p_{jt}-x_{jt}'\tilde\beta$, where $\tilde\theta=(\tilde\sigma,\tilde\alpha,\tilde\beta)$ is an estimator of $\theta_0$ satisfying $\|\tilde\theta-\theta_0\|=O_p(T^{-1/2})$.
Suppose Conditions~\ref{cond:regularity}, \ref{cond:y-smooth}, and \ref{cond:kernel} hold. Then for each $j=1,\dots,J$,
\[
    \sup_{z\in\mathcal Z_T}\big|\hat\Sigma_{j}(z)^{-1}-\Sigma_{j}(z)^{-1}\big|=O_p(\rho_T+T^{-1/2})=o_p(T^{-1/4}),
\]
and thus
\[
    \E_T\Big[\big|\hat\Sigma_{j}(z_{jt})^{-1}-\Sigma_{j}(z_{jt})^{-1}\big|^2\,\tau_T(z_{jt})\Big]=O_p(\rho_T^2+T^{-1})=o_p(T^{-1/2}).
\]
Equivalently, in the notation of Section~\ref{sec:theory-general}, where $\Lambda_j=\Sigma_j^{-1}$ and $\hat\Lambda_j=\hat\Sigma_j^{-1}\tau_T$,
\[
    \sup_{z\in\mathcal Z}\big|\hat\Lambda_{j}(z)-\Lambda_{j}(z)\tau_T(z)\big|=o_p(T^{-1/4}),
    \qquad
    \E_T\big[\hat\Lambda_{j}(z_{jt})-\Lambda_{j}(z_{jt})\tau_T(z_{jt})\big]^2=o_p(T^{-1/2}),
\]
which are Condition~\ref{cond:consistency}\eqref{cond:cons-uniform} and the weight part of Condition~\ref{cond:normality}\eqref{cond:norm-msr}.
\end{lemma}

\begin{proof}[Proof of Lemma~\ref{lemma:kernel-weight}]
Throughout, $C$ denotes a finite constant that may change from line to line, and all suprema are over $z\in\mathcal Z_T$ unless indicated otherwise. Define the \emph{infeasible} weight, formed with the true unobserved product characteristics $\xi_{jt}$ but the same denominator,
\[
    \Sigma_{j}^{*}(z)=\frac{\E_T\!\big[K_h(z_{jt}-z)\,\xi_{jt}^{2}\big]}{\hat f_j(z)}.
\]
We bound $\sup_{z\in\mathcal Z_T}|\hat\Sigma_{j}(z)-\Sigma_{j}(z)|$ through the decomposition
\begin{equation}\label{eq:kw-decomp}
    \sup_{z\in\mathcal Z_T}\big|\hat\Sigma_{j}(z)-\Sigma_{j}(z)\big|
    \le \underbrace{\sup_{z\in\mathcal Z_T}\big|\hat\Sigma_{j}(z)-\Sigma_{j}^{*}(z)\big|}_{=:A_T}
    +\underbrace{\sup_{z\in\mathcal Z_T}\big|\Sigma_{j}^{*}(z)-\Sigma_{j}(z)\big|}_{=:B_T}.
\end{equation}

\medskip
\noindent\textit{Step 0: preliminaries.}
By Lemma~\ref{lemma:kernel} applied to $\psi\equiv1$, $\sup_{z\in\mathcal Z_T}|\hat f_j(z)-f_j(z)|=O_p(\rho_T)=o_p(1)$, and since $\inf_{\mathcal Z}f_j\ge\underline f>0$ (Condition~\ref{cond:kernel}\eqref{cond:kernel-boundedness}), the event
\begin{equation}\label{eq:kw-denom}
    \mathcal E_T=\Big\{\inf_{z\in\mathcal Z_T}\hat f_j(z)\ge \underline f/2\Big\}
\end{equation}
satisfies $P(\mathcal E_T)\to1$. All bounds below are stated on $\mathcal E_T$, which suffices for $O_p(\cdot)$ statements. On $\mathcal E_T$ the maps $g\mapsto g/\hat f_j$ are Lipschitz on $\mathcal Z_T$ with constant $2/\underline f$, a fact used repeatedly.
 
Moreover, by Condition~\ref{cond:regularity}\eqref{cond:cons-reg} the $\xi_{jt}$ are uniformly bounded, $\sup_{j,t}|\xi_{jt}|\le C_\xi<\infty$, and by Condition~\ref{cond:regularity} and Condition~\ref{cond:kernel}\eqref{cond:kernel-psi} the estimated residuals are uniformly bounded in probability:
\begin{equation}\label{eq:kw-muhat-bd}
    \sup_{j,t}\big|\hat\xi_{jt}\big|
    \le L\,|\tilde\sigma-\sigma_0|+\sup_{j,t}\big|y_{j}(\sigma_0,w_t)\big|+|\tilde\alpha|\,\sup_{j,t}|p_{jt}|+\sup_{j,t}\|x_{jt}\|_{\infty}\|\tilde\beta\|_1=O_p(1).
\end{equation}

\medskip
\noindent\textit{Step 1: the infeasible term $B_T$.}
The regressand $\xi_{jt}^{2}$ is bounded by $C_\xi^2$, and its density-weighted conditional mean equals $\Sigma_{j}(z)f_j(z)$, which is $l$-times continuously differentiable in $z$ with uniformly bounded derivatives by Condition~\ref{cond:kernel}\eqref{cond:kernel-psi} (which lists $\Sigma_{j}$ and $f_j$ among the $l$-smooth targets). Lemma~\ref{lemma:kernel}, applied with $\psi=\xi_{jt}^{2}$, therefore yields
\[
    \sup_{z\in\mathcal Z_T}\Big|\E_T\!\big[K_h(z_{jt}-z)\xi_{jt}^{2}\big]-\Sigma_{j}(z)f_j(z)\Big|=O_p(\rho_T).
\]
Writing $\Sigma_{j}^{*}=\big(\E_T[K_h\xi_{jt}^{2}]\big)/\hat f_j$ and $\Sigma_{j}=(\Sigma_{j} f_j)/f_j$, adding and subtracting $(\Sigma_{j} f_j)/\hat f_j$ and using \eqref{eq:kw-denom},
\begin{align*}
    B_T
    &\le \sup_{z\in\mathcal Z_T}\frac{\big|\E_T[K_h\xi_{jt}^{2}]-\Sigma_{j}(z)f_j(z)\big|}{\hat f_j(z)}
    +\sup_{z\in\mathcal Z_T} \Sigma_{j}(z)\,\frac{\big|\hat f_j(z)-f_j(z)\big|}{\hat f_j(z)}\\
    &\le \frac{2}{\underline f}\,O_p(\rho_T)+\frac{2\,\overline\Sigma_{j}}{\underline f}\,O_p(\rho_T)
    =O_p(\rho_T),
\end{align*}
where $\overline\Sigma_{j}=\sup_{\mathcal Z}\Sigma_{j}<\infty$.

\medskip
\noindent\textit{Step 2: uniform bound on the residual perturbation.}
Decompose the estimated residual error, evaluated at the preliminary $\tilde\theta=(\tilde\sigma,\tilde\alpha,\tilde\beta)$, as
\[
    \hat\xi_{jt}-\xi_{jt}
    =\big[y_{j}(\tilde\sigma,w_t)-y_{j}(\sigma_0,w_t)\big]
    -(\tilde\alpha-\alpha_0)p_{jt}-x_{jt}'(\tilde\beta-\beta_0),
\]
using $\xi_{jt}=y_{j}(\sigma_0,w_t)-\alpha_0p_{jt}-x_{jt}'\beta_0$. By the Lipschitz property of $y_{j}(\sigma,w_t)$ in $\sigma$ (Condition~\ref{cond:kernel}\eqref{cond:kernel-psi}) and the boundedness of $p_{jt}$ and $x_{jt}$,
\begin{equation}\label{eq:kw-resid}
    \Delta_{\xi,T}:=\sup_{j,t}\big|\hat\xi_{jt}-\xi_{jt}\big|\lesssim |\tilde\sigma-\sigma_0|+|\tilde\alpha-\alpha_0|+\|\tilde\beta-\beta_0\|_1=O_p(\|\tilde\theta-\theta_0\|)=O_p(T^{-1/2}).
\end{equation}

\medskip
\noindent\textit{Step 3: the plug-in term $A_T$.}
On $\mathcal E_T$ the common denominator cancels and
\[
    A_T
    =\sup_{z\in\mathcal Z_T}\frac{\big|\E_T[K_h(z_{jt}-z)(\hat\xi_{jt}^{\,2}-\xi_{jt}^{2})]\big|}{\hat f_j(z)}
    \le \frac{2}{\underline f}\,\sup_{z\in\mathcal Z_T}\,\E_T\Big[\big|K_h(z_{jt}-z)\big|\,\big|\hat\xi_{jt}^{\,2}-\xi_{jt}^{2}\big|\Big].
\]
Factor $\hat\xi_{jt}^{\,2}-\xi_{jt}^{2}=(\hat\xi_{jt}-\xi_{jt})(\hat\xi_{jt}+\xi_{jt})$ and bound the second factor uniformly by $\sup_{j,t}(|\hat\xi_{jt}|+|\xi_{jt}|)=O_p(1)$ (Step~0). Hence, pulling the uniform residual bound \eqref{eq:kw-resid} out of the average,
\[
    A_T\le \frac{2}{\underline f}\,\Delta_{\xi,T}\cdot O_p(1)\cdot \sup_{z\in\mathcal Z_T}\,\E_T\big|K_h(z_{jt}-z)\big|=\Delta_{\xi,T}\cdot O_p(1).
\]
The last equality uses
\[
    \sup_{z\in\mathcal Z_T}\Big|\E_T\big|K_h(z_{jt}-z)\big|-C_K f_j(z)\Big|=o_p(1),
    \qquad C_K=\int|K(u)|\,du,
\]
which follows from the argument of Lemma~\ref{lemma:kernel} with $\psi\equiv1$ and $K$ replaced by $|K|$---a Lipschitz, compactly supported function, so the variance bound is unchanged, while the bias is $O(h)$ rather than $O(h^l)$ because $|K|$ is of order one; only $o_p(1)$ is needed here. Consequently $\sup_{z\in\mathcal Z_T}\E_T|K_h(z_{jt}-z)|=O_p(1)$ and
\begin{equation}\label{eq:kw-AT}
    A_T=O_p(\Delta_{\xi,T})=O_p\big(T^{-1/2}\big).
\end{equation}

\medskip
\noindent\textit{Step 4: conclusion.}
Collecting the bounds for $A_T$ and $B_T$ in \eqref{eq:kw-decomp}, and using $\|\tilde\theta-\theta_0\|=O_p(T^{-1/2})$,
\[
    \sup_{z\in\mathcal Z_T}\big|\hat\Sigma_{j}(z)-\Sigma_{j}(z)\big|=O_p(\rho_T+T^{-1/2})= o_p(T^{-1/4}),
\]
because Lemma~\ref{lemma:kernel} gives $\rho_T=o(T^{-1/4})$.
Note $\Sigma_{j}(z)\ge\underline c>0$ by Condition~\ref{cond:regularity}\eqref{cond:cons-weight}, which bounds $\Lambda_j=\Sigma_j^{-1}$ above; since $\sup_{z\in\mathcal Z_T}|\hat\Sigma_{j}(z)-\Sigma_{j}(z)|=o_p(1)$, the event $\mathcal F_T=\{\inf_{\mathcal Z_T}\hat\Sigma_{j}\ge\underline c/2\}$ has $P(\mathcal F_T)\to1$. On $\mathcal E_T\cap\mathcal F_T$,
\[
    \big|\hat\Sigma_{j}(z)^{-1}-\Sigma_{j}(z)^{-1}\big|
    =\frac{\big|\hat\Sigma_{j}(z)-\Sigma_{j}(z)\big|}{\hat\Sigma_{j}(z)\,\Sigma_{j}(z)}
    \le \frac{2}{\underline c^{2}}\,\big|\hat\Sigma_{j}(z)-\Sigma_{j}(z)\big|
    \qquad\text{for all }z\in\mathcal Z_T,
\]
so taking the supremum,
\[
    \sup_{z\in\mathcal Z_T}\big|\hat\Sigma_{j}(z)^{-1}-\Sigma_{j}(z)^{-1}\big|
    \le \frac{2}{\underline c^{2}}\,\sup_{z\in\mathcal Z_T}\big|\hat\Sigma_{j}(z)-\Sigma_{j}(z)\big|
    =o_p(T^{-1/4}).
\]
The mean-square statement follows since 
\begin{align*}
    \E_T\Big[\big|\hat\Sigma_{j}(z_{jt})^{-1}-\Sigma_{j}(z_{jt})^{-1}\big|^2\,\tau_T(z_{jt})\Big]
    \le \sup_{\mathcal Z_T}\Big[\big|\hat\Sigma_{j}(z)^{-1}-\Sigma_{j}(z)^{-1}\big|^2\Big]=o_p(T^{-1/2}).
\end{align*}
The statements for $\hat\Lambda_j$ follow from $\hat\Lambda_j-\Lambda_j\tau_T=(\hat\Sigma_j^{-1}-\Sigma_j^{-1})\tau_T$, both sides vanishing off $\mathcal Z_T$.

\end{proof}

\begin{lemma}[Design matrix]\label{lemma:series-gram}
Under Condition~\ref{cond:series}\eqref{cond:series-gram} and \eqref{cond:series-rate},
\[
    \big\|\hat Q_j-Q_j\big\|_{op}\lesssim \zeta_0(K)\sqrt{\log K/T}\quad \text{w.p.a.1},
\]
and hence, $\zeta_0(K)\sqrt{\log K/T}=o(1)$ implies the event $\{\lambda_{\min}(\hat Q_j)\ge\underline q/2\}$ has probability approaching one, on which $\|\hat Q_j^{-1}\|_{op}\le 2/\underline q$.
\end{lemma}

\begin{proof}[Proof of Lemma~\ref{lemma:series-gram}]
Each entry of $\hat Q_j-Q_j$ is a market average of the mean-zero, independent terms $Z_t:=\tfrac 1T\left[\phi^K(z_{jt})\phi^K(z_{jt})'-\E[\phi^K(z_{jt})\phi^K(z_{jt})']\right]$. $Z_t$ are independent, symmetric, and mean-zero.

\medskip\noindent\emph{Uniform bound.}
Let $\phi^K_t=\phi^K(z_{jt})$.
Since $\|\phi^K_t\|\le\zeta_0(K)$ a.s.,
$\|\phi^K_t\phi^{K\prime}_t\|_{op}=\|\phi^K_t\|^2\le\zeta_0(K)^2$, and
$\|Q\|_{op}\le tr(Q)=\E\|\phi^K_t\|^2\le\zeta_0(K)^2$. Hence
\[
  \|Z_t\|_{op}\le\tfrac1T\big(\|\phi^K_t\phi^{K\prime}_t\|_{op}+\|Q\|_{op}\big)
  \le\frac{2\zeta_0(K)^2}{T}=:M.
\]

\medskip\noindent\emph{Variance proxy.}
Using $(\phi^K_t\phi^{K\prime}_t)^2=\|\phi^K_t\|^2\phi^K_t\phi^{K\prime}_t$,
\[
  \E[(\phi^K_t\phi^{K\prime}_t-Q)^2]=\E[\|\phi^K_t\|^2\phi^K_t\phi^{K\prime}_t]-Q^2,
\]
so that
\begin{align*}
    \sigma^2:=&\Big\|\textstyle\sum_t \E[Z_t^2]\Big\|_{op}=\tfrac1T\big\|\E[(\phi^K_t\phi^{K\prime}_t-Q)^2]\big\|_{op}
    \le\tfrac1T(\big\|\E[\|\phi^K_t\|^2\phi^K_t\phi^{K\prime}_t]\big\|_{op}+\|Q\|_{op}^2)
    \le\frac{2\zeta_0(K)^2\,\bar q}{T}\quad \text{w.l.o.g. } \bar q\ge1.
\end{align*}

\medskip\noindent\emph{Matrix Bernstein.}
For every $t\ge0$,
\[
  P\big(\|\hat Q_j-Q_j\|_{op}\ge t\big)
  \le 2K\exp\!\Big(\frac{-t^2/2}{\sigma^2+Mt/3}\Big).
\]
Taking $t\asymp\sqrt{\zeta_0(K)^2\log K/T}=o(1)$ yields
\[
  P\big(\|\hat Q_j-Q_j\|_{op}\ge t\big)
  \le 2K\exp\!(\frac{-c\log K}{2\bar q+2\sqrt{\zeta_0(K)^2\log K/T}}\Big)
  \lesssim K^{-c'},
\]
which implies
\[
  \|\hat Q_j-Q_j\|_{op}
  \lesssim \zeta_0(K)\sqrt{\tfrac{\log K}{T}}\quad \text{w.p.a.1}.
\]

\medskip\noindent\emph{Eigenvalue transfer.}
Let $A_n=\{\|\hat Q_j-Q_j\|_{op}\le \underline q/2\}$; then $P(A_n)\to1$.
By Weyl's inequality, on $A_n$,
\[
  \lambda_{\min}(\hat Q_j)\ge\lambda_{\min}(Q_j)-\frac{\underline{q}}{2}\ge\frac{\underline q}{2},
  \qquad
  \lambda_{\max}(\hat Q_j)\le\lambda_{\max}(Q_j)+\frac{\underline{q}}{2}\le \frac{3\underline q}{2},
\]
so $\hat Q_j$ is invertible and $\|\hat Q_j^{-1}\|_{op}=1/\lambda_{\min}(\hat Q_j)\le 2/\underline q$.
\end{proof}

\begin{lemma}[Series convergence rate]\label{lemma:series-rate}
Under Condition~\ref{cond:regularity}, \ref{cond:y-smooth}, and \ref{cond:series}, for each $\psi$,
\[
    \sup_{\sigma\in\Theta_\sigma}\big\|\hat\Pi_{\psi,j}(\sigma)-\Pi^{K}_{\psi,j}(\sigma)\big\|
    =O_p\!\Big(\sqrt{K\log(K\vee T)/T}+K^{-\gamma}\Big)=o_p(T^{-1/4}),
\]
and consequently,
\[
    \sup_{\sigma\in\Theta_\sigma}\Big(\E_T\big[\hat\mu_{\psi,j}(z_{jt};\sigma)-\mu_{\psi,j}(z_{jt};\sigma)\big]^2\Big)^{1/2}
    =O_p\!\Big(\sqrt{K\log(K\vee T)/T}+K^{-\gamma}\Big)=o_p(T^{-1/4}).
\]
and
\[
    \sup_{\sigma\in\Theta_\sigma}\sup_{z}\big|\hat\mu_{\psi,j}(z;\sigma)-\mu_{\psi,j}(z;\sigma)\big|=O_p\!\Big(\zeta_0(K)\sqrt{K\log(K\vee T)/T}+\zeta_0(K)K^{-\gamma}\Big)=o_p(1).
\]
\end{lemma}

\begin{proof}[Proof of Lemma~\ref{lemma:series-rate}]
Work on the event $\mathcal E_T=\{\lambda_{\min}(\hat Q_j)\ge\underline q/2\}$, which has $P(\mathcal E_T)\to1$ by Lemma~\ref{lemma:series-gram}; off $\mathcal E_T$ the contribution is $o_p(1)$. Decompose $\psi_{jt}(\sigma)=\mu_{\psi,j}(z_{jt};\sigma)+e^{\psi}_{jt}(\sigma)$ and $\mu_{\psi,j}(z_{jt};\sigma)=\phi^K(z_{jt})'\Pi^{K}_{\psi,j}(\sigma)+b_{jt}(\sigma)$ with $\sup_{\sigma,t}|b_{jt}(\sigma)|=O(K^{-\gamma})$ by Condition~\ref{cond:series}\eqref{cond:series-approx}. Then
\begin{align*}
    \hat\Pi_{\psi,j}(\sigma)-\Pi^{K}_{\psi,j}(\sigma)
    &=\underbrace{\hat Q_j^{-1}\,\E_T\!\big[\phi^K(z_{jt})\,e^{\psi}_{jt}(\sigma)\big]}_{\text{(I)}}
    +\underbrace{\hat Q_j^{-1}\,\E_T\!\big[\phi^K(z_{jt})\,b_{jt}(\sigma)\big]}_{\text{(II)}}.
\end{align*}

\emph{For the first term (I),} $\|\hat Q_j^{-1}\|\le 2/\underline q$ on $\mathcal E_T$.
By the Lipschitz property of $\psi_{jt}(\sigma)$ in Condition~\ref{cond:series}\eqref{cond:series-Lipschitz}, $e^{\psi}_{jt}(\sigma)$ is also Lipschitz in $\sigma$. 
The supremum over $\sigma\in\Theta_\sigma$ is controlled by a similar concentration argument as in Lemma~\ref{lemma:kernel}, using a finite $\epsilon$-net of $\Theta_\sigma$ and the Lipschitz property to control the difference between points in the net and arbitrary points in $\Theta_\sigma$. This yields on $\mathcal E_T$,
\begin{align*}
    \sup_{\Theta_\sigma}\Big\|\hat Q_j^{-1}\,\E_T\!\big[\phi^K(z_{jt})\,(e^{\psi}_{jt}(\sigma))\big]\Big\|
    \lesssim & \sup_{\Theta_\sigma}\Big\|\E_T\!\phi^K(z_{jt})\,e^{\psi}_{jt}(\sigma)\Big\|\\
    \le & \max_{i\le (C/\delta)^{d_\sigma}} \Big\|\E_T\!\big[\phi^K(z_{jt})\,(e^{\psi}_{jt}(\sigma_i))\big]\Big\|+\sup_{\Theta_\sigma,\|\sigma-\sigma_i\|\le \delta}\Big\|\E_T\!\big[\phi^K(z_{jt})\,(e^{\psi}_{jt}(\sigma)-e^{\psi}_{jt}(\sigma_i))\big]\Big\|\\
    \le & \max_{i\le (C/\delta)^{d_\sigma}} \Big\|\E_T\!\big[\phi^K(z_{jt})\,(e^{\psi}_{jt}(\sigma_i))\big]\Big\|+\sqrt{\E_T\big\|\phi^K(z_{jt})\big\|^2}\sup_{\Theta_\sigma,\|\sigma-\sigma_i\|\le \delta}\sqrt{\E_T (e^{\psi}_{jt}(\sigma)-e^{\psi}_{jt}(\sigma_i))^2}\\
    \lesssim & \max_{i\le (C/\delta)^{d_\sigma}} \Big\|\E_T\!\big[\phi^K(z_{jt})\,(e^{\psi}_{jt}(\sigma_i))\big]\Big\|+\sqrt{tr(\hat{Q}_j)}\sup_{\Theta_\sigma,\|\sigma-\sigma_i\|\le \delta}|\sigma-\sigma_i|\\
    \lesssim & \max_{i\le (C/\delta)^{d_\sigma}} \Big\|\E_T\!\big[\phi^K(z_{jt})\,(e^{\psi}_{jt}(\sigma_i))\big]\Big\|+O_p(\sqrt{K/T})\quad \text{choose }\delta =c/\sqrt{T}
\end{align*}
Define $X_t = \frac{1}{T}\phi^K(z_{jt})e^{\psi}_{jt}(\sigma_i)$, which is a mean-zero random vector with $\|X_t\|_{op}\le \frac{1}{T}\zeta_0(K)\sup_{\Theta_\sigma}|e^{\psi}_{jt}(\sigma)|\lesssim \frac{\zeta_0(K)}{T}$. 
The variance term 
\begin{align*}
    \sigma^2 = \Big\|\sum_{t=1}^T \E[X_t X_t']\Big\|_{op} \le \frac{1}{T^2}\sum_{t=1}^T \E[\|\phi^K(z_{jt})\|^2 (e^{\psi}_{jt}(\sigma_i))^2] \lesssim \frac{1}{T} tr(Q_j)\lesssim \frac{K}{T}.
\end{align*}
Applying the matrix Bernstein inequality, with $t=C\sqrt{K\log(K\vee T)/T}$ and $\zeta_0(K)\sqrt{K\log(K\vee T)/T}\to 0$ we have
\begin{align*}
    P\Big(\Big\|\E_T\!\big[\phi^K(z_{jt})\,(e^{\psi}_{jt}(\sigma_i))\big]\Big\| \ge t\Big) \le& (K+1) \exp\Big(-\frac{t^2/2}{\sigma^2 + Mt/3}\Big)\\
    \lesssim& K \exp\Big(-\frac{t^2/2}{K/T + \zeta_0(K)t/(3T)}\Big)\\
    =& K\exp\Big(-\frac{Tt^2/2}{K + \zeta_0(K)t/3}\Big)\\
    \lesssim& K\exp\Big(-c\frac{Tt^2}{K}\Big)\\
    \le& K(K\vee T)^{-cC^2}
\end{align*}
and hence with large enough $C$,
\begin{align*}
    P(\max_{i\le C/\delta} \Big\|\E_T\!\big[\phi^K(z_{jt})\,(e^{\psi}_{jt}(\sigma_i))\big]\Big\| \ge t) \le& (\frac{C}{\delta})^{d_\sigma} K(K\vee T)^{-cC^2}\lesssim T^{d_\sigma/2} K(K\vee T)^{-cC^2} \to 0
\end{align*}
Summarizing, we have
\begin{align*}
    \sup_{\Theta_\sigma}\Big\|\hat Q_j^{-1}\,\E_T\!\big[\phi^K(z_{jt})\,(e^{\psi}_{jt}(\sigma))\big]\Big\|
    =& O_p(\sqrt{K\log(K\vee T)/T})+O_p(\sqrt{K/T})=O_p(\sqrt{K\log(K\vee T)/T}).
\end{align*}
\emph{For the second term, (II).}
\begin{align*}
    \sup_{\Theta_\sigma}\|\hat Q_j^{-1} \E_T\!\big[\phi^K(z_{jt})\,b_{jt}(\sigma)\big]\|\lesssim &\sup_{\Theta_\sigma}\,\|\E_T\!\big[\phi^K(z_{jt})\,b_{jt}(\sigma)\big]\|\\
    = &\sup_{\Theta_\sigma}\sup_{\|u\|=1}u'\E_T\big[\phi^K(z_{jt})\,b_{jt}(\sigma)\big]\\
    = &\sup_{\Theta_\sigma}\sup_{\|u\|=1}\E_T u'\phi^K(z_{jt})\,b_{jt}(\sigma)\\
    \le &\sup_{\Theta_\sigma}\sup_{\|u\|=1}\sqrt{\E_T\|u'\phi^K(z_{jt})\|^2}\sqrt{\E_T b_{jt}(\sigma)^2}\\
    = &\sup_{\Theta_\sigma}\sup_{\|u\|=1}\sqrt{u'\hat Q_j u}\sqrt{\E_T b_{jt}(\sigma)^2}\\
    \lesssim & \sup_{\Theta_\sigma}\sqrt{\E_T b_{jt}(\sigma)^2} = O_p(K^{-\gamma}).
\end{align*}
We therefore have the bound 
\[
    \sup_{\sigma\in\Theta_\sigma}\big\|\hat\Pi_{\psi,j}(\sigma)-\Pi^{K}_{\psi,j}(\sigma)\big\|
    =O_p\!\Big(\sqrt{K\log(K\vee T)/T}+K^{-\gamma}\Big).
\]
Consequently, on the event $\mathcal E_T$,
\begin{align*}
    &\sup_{\sigma\in\Theta_\sigma}\Big(\E_T\big[\hat\mu_{\psi,j}(z_{jt};\sigma)-\mu_{\psi,j}(z_{jt};\sigma)\big]^2\Big)^{1/2}\\
    =&\sup_{\sigma\in\Theta_\sigma}\Big(\E_T\big[\phi^K(z_{jt})(\hat \Pi_{\psi,j}(\sigma)-\Pi^{K}_{\psi,j}(\sigma))-b_{jt}(\sigma)\big]^2\Big)^{1/2}\\
    \le& \sup_{\sigma\in\Theta_\sigma}\Big(\E_T\big[\phi^K(z_{jt})(\hat \Pi_{\psi,j}(\sigma)-\Pi^{K}_{\psi,j}(\sigma))\big]^2\Big)^{1/2}+\sup_{\sigma\in\Theta_\sigma}\Big(\E_T\big[b_{jt}(\sigma)\big]^2\Big)^{1/2}\\
    \lesssim& \sup_{\sigma\in\Theta_\sigma}\big\|\hat \Pi_{\psi,j}(\sigma)-\Pi^{K}_{\psi,j}(\sigma)\big\|+O_p(K^{-\gamma})\\
    =&O_p\!\Big(\sqrt{K\log(K\vee T)/T}+K^{-\gamma}\Big),
\end{align*}
which gives the second claim.
Finally,
\begin{align*}
    \sup_{\sigma\in\Theta_\sigma}\sup_{z}\big|\hat\mu_{\psi,j}(z;\sigma)-\mu_{\psi,j}(z;\sigma)\big|
    \le& \sup_{\sigma\in\Theta_\sigma}\sup_{z}|\phi^K(z)'(\hat\Pi_{\psi,j}(\sigma)-\Pi^{K}_{\psi,j}(\sigma))|+O_p(K^{-\gamma})\\
    \le& \sup_{\sigma\in\Theta_\sigma}\zeta_0(K)\,\big\|\hat\Pi_{\psi,j}(\sigma)-\Pi^{K}_{\psi,j}(\sigma)\big\|+O_p(K^{-\gamma})\\
    =&O_p(\zeta_0(K)\sqrt{K\log(K\vee T)/T}+\zeta_0(K)K^{-\gamma})+O_p(K^{-\gamma})=o_p(1).
\end{align*}

\end{proof}

\begin{lemma}[Rate for series plugged-in weight]\label{lemma:series-weight}
Let $\hat\Sigma_{j}(z)=\phi^K(z)'\hat Q_j^{-1}\E_T[\phi^K(z_{jt})\hat\xi_{jt}^2]$ be the series estimator of the conditional variance, where $\hat\xi_{jt}$ is the residual from a preliminary unweighted fit $\tilde\theta$, with $\|\tilde\theta-\theta_0\|=O_p(T^{-1/2})$. Under Conditions~\ref{cond:regularity}, \ref{cond:y-smooth} and \ref{cond:series},
\[
    \Big(\E_T\big[\hat\Sigma_{j}(z_{jt})^{-1}-\Sigma_{j}(z_{jt})^{-1}\big]^2\Big)^{1/2}
    =O_p\!\Big(\sqrt{K\log(K\vee T)/T}+K^{-\gamma}+\|\tilde\theta-\theta_0\|\Big)=o_p(T^{-1/4}).
\]
and
\[
    \sup_{z\in\mathcal Z}\big|\hat\Sigma_{j}(z)^{-1}-\Sigma_{j}(z)^{-1}\big|=O_p\!\Big(\zeta_0(K)\big(\sqrt{K\log(K\vee T)/T}+K^{-\gamma}+\|\tilde\theta-\theta_0\|\big)\Big)=o_p(1).
\]
\end{lemma}

\begin{proof}[Proof of Lemma~\ref{lemma:series-weight}]
Define the \emph{infeasible} weight, formed from the true shocks but the same design,
\[
    \Sigma_{j}^{*}(z)=\phi^K(z)'\hat Q_j^{-1}\E_T\big[\phi^K(z_{jt})\,\xi_{jt}^2\big],
\]
and split
\[
    \Big(\E_T\big[\hat\Sigma_{j}(z_{jt})-\Sigma_{j}(z_{jt})\big]^2\Big)^{1/2}
    \le \underbrace{\Big(\E_T\big[\hat\Sigma_{j}(z_{jt})-\Sigma_{j}^{*}(z_{jt})\big]^2\Big)^{1/2}}_{A_T}
    +\underbrace{\Big(\E_T\big[\Sigma_{j}^{*}(z_{jt})-\Sigma_{j}(z_{jt})\big]^2\Big)^{1/2}}_{B_T}.
\]

\emph{Infeasible term $B_T$.}
The infeasible weight $\Sigma^*$ is the series regression of $\xi_{jt}^2$ on $\phi^K(z_{jt})$, i.e.\ $\Sigma^*=\hat\mu_{\psi,j}$ for the regressand $\psi_{jt}=\xi_{jt}^2$. This regressand is bounded (Condition~\ref{cond:regularity}), and its conditional mean is exactly $\Sigma_{j}(z)=\E[\xi_{jt}^2\mid z_{jt}=z]$, which is approximable at rate $K^{-\gamma}$ by the dictionary under Condition~\ref{cond:series}\eqref{cond:series-approx}. By Lemma~\ref{lemma:series-rate}, the series projection of $\xi_{jt}^2$ onto $\phi^K(z_{jt})$ converges at the rate
\[
    B_T=O_p\!\big(\sqrt{K\log(K\vee T)/T}+K^{-\gamma}\big).
\]

\emph{Residual replacement $A_T$.}
$\hat\Sigma_{j}-\Sigma_{j}^*=\phi^K(z)'\hat Q_j^{-1}\E_T[\phi^K(z_{jt})(\hat\xi_{jt}^2-\xi_{jt}^2)]$, and by Lemma~\ref{lemma:series-rate}, under the event $\{\lambda_{\min}(\hat Q_j)\ge\underline q/2\}\cap\{\lambda_{\max}(\hat{Q}_j)\le \frac{\underline q}{2}+\bar{q}\}$,
\begin{align*}
    \Big(\E_T\big[\hat\Sigma_{j}(z_{jt})-\Sigma_{j}^{*}(z_{jt})\big]^2\Big)^{1/2}=
    &\Big(\E_T\Big[\E_T[\phi^K(z_{jt})(\hat\xi_{jt}^2-\xi_{jt}^2)]'\hat Q_j^{-1}\phi^K(z_{jt})\phi^K(z_{jt})'\hat Q_{j}^{-1}\E_T[\phi^K(z_{jt})(\hat\xi_{jt}^2-\xi_{jt}^2)] \Big]\Big)^{1/2}\\
    = & \Big(\E_T[\phi^K(z_{jt})(\hat\xi_{jt}^2-\xi_{jt}^2)]'\hat Q_{j}^{-1}\E_T[\phi^K(z_{jt})(\hat\xi_{jt}^2-\xi_{jt}^2)] \Big)^{1/2}\\
    \lesssim & \|\E_T\phi^K(z_{jt})\big(\hat\xi_{jt}^2-\xi_{jt}^2\big)\|\\
    \le & \sqrt{\lambda_{\max}(\hat Q_j)}\; \sqrt{\E_T(\hat\xi_{jt}^2-\xi_{jt}^2)^2}\\
    \lesssim & \sqrt{\E_T(\hat\xi_{jt}-\xi_{jt})^2}\; \sqrt{\sup_{t}|\hat\xi_{jt}+\xi_{jt}|^2}.
\end{align*}
For the difference, use $\xi_{jt}=y_{j}(\sigma_0,w_t)-\alpha_0 p_{jt}-x_{jt}'\beta_0$ to write
\begin{align*}
    \hat\xi_{jt}-\xi_{jt}
    &=\big[y_{j}(\tilde\sigma,w_t)-y_{j}(\sigma_0,w_t)\big]
    -(\tilde\alpha-\alpha_0)p_{jt}-x_{jt}'(\tilde\beta-\beta_0).
\end{align*}
By Lipschitz properties of $y_{j}(\sigma,w_t)$ in $\sigma$ (Condition~\ref{cond:series}\eqref{cond:series-Lipschitz}),
\begin{align}\label{eq:kw-resid1}
    \Delta_{\xi,T}:=\sup\big|\hat\xi_{jt}-\xi_{jt}\big|\lesssim |\tilde\sigma-\sigma_0|+|\tilde\alpha-\alpha_0|+\|\tilde\beta-\beta_0\|_1=O_p(\|\tilde\theta-\theta_0\|)=O_p(T^{-1/2}).
\end{align}
Then
\[
    \sqrt{\E_{T}\big|\hat\xi_{jt}-\xi_{jt}\big|^2}\le \Delta_{\xi,T}=O_p\!\Big(\|\tilde{\theta}-\theta_0\|\Big),\quad \sup_{t}|\hat\xi_{jt}+\xi_{jt}|=O_p(1),
\]
and hence
\begin{align*}
    \Big(\E_T\big[\hat\Sigma_{j}(z_{jt})-\Sigma_{j}^{*}(z_{jt})\big]^2\Big)^{1/2}=O_p\!\Big(\|\tilde{\theta}-\theta_0\|\Big).
\end{align*}
so $A_T$ is of the same order.
Hence the total error is 
\[
    \Big(\E_T\big[\hat\Sigma_{j}(z_{jt})-\Sigma_{j}(z_{jt})\big]^2\Big)^{1/2}
    =O_p\!\Big(\sqrt{K\log(K\vee T)/T}+K^{-\gamma}+\|\tilde{\theta}-\theta_0\|\Big)=o_p(T^{-1/4}).
\]
\emph{Inverse transfer.} We need to first show that $\max_{1\le t \le T}|\hat\Sigma_{j}(z_{jt})-\Sigma_{j}(z_{jt})|=o_p(1)$, so that $\min_{1\le t\le T}\hat\Sigma_{j}(z_{jt})^{-1}$ is bounded away from zero with probability approaching one. 
\begin{align*}
    \sup_{z}|\hat\Sigma_{j}(z)-\Sigma_{j}(z)|
    &\le \sup_{z}|\hat\Sigma_{j}(z)-\Sigma_{j}^*(z)|+ \sup_{z}|\Sigma_{j}^*(z)-\Sigma_{j}(z)|\\
    &= \sup_{z}|\hat\Sigma_{j}(z)-\Sigma_{j}^*(z)|+O_p(\zeta_0(K)(\sqrt{K\log(K\vee T)/T}+K^{-\gamma}))\\
    &= \sup_{z}|\phi^K(z)\hat{Q}_j^{-1}\E_T[\phi^K(z_{jt})(\hat\xi_{jt}^2-\xi_{jt}^2)]|+O_p(\zeta_0(K)(\sqrt{K\log(K\vee T)/T}+K^{-\gamma}))\\
    &\lesssim \zeta_0(K)\|\E_T\big[\phi^K(z_{jt})(\hat\xi_{jt}^2-\xi_{jt}^2)\big]\|+O_p(\zeta_0(K)(\sqrt{K\log(K\vee T)/T}+K^{-\gamma}))\quad \text{w.p.a.1}\\
    &\lesssim \zeta_0(K)\sqrt{\E_T(\hat\xi_{jt}-\xi_{jt})^2}\ \sqrt{\sup_{z}|\hat\xi_{jt}+\xi_{jt}|^2}+O_p(\zeta_0(K)(\sqrt{K\log(K\vee T)/T}+K^{-\gamma}))\\
    &=O_p(\zeta_0(K)(\sqrt{K\log(K\vee T)/T}+K^{-\gamma}+\|\tilde{\theta}-\theta_0\|))\\
    &=O_p(\zeta_0(K)(\sqrt{K\log(K\vee T)/T}+K^{-\gamma}))=o_p(1),
\end{align*}
where the second inequality uses Lemma~\ref{lemma:series-gram}, the third inequality applies the same argument in the above analysis, and $\zeta_0(K)\ \|\tilde\theta-\theta_0\|=o_p(1)$ holds by Condition~\ref{cond:series}\eqref{cond:series-rate}.
Hence,
\begin{align*}
    \E_T\big[\hat\Sigma_{j}(z_{jt})^{-1}-\Sigma_{j}(z_{jt})^{-1}\big]^2=&\E_T\big[\hat\Sigma_{j}(z_{jt})^{-1}(\Sigma_{j}(z_{jt})-\hat\Sigma_{j}(z_{jt}))\Sigma_{j}(z_{jt})^{-1}\big]^2\\
    \le & \E_T\big[\hat\Sigma_{j}(z_{jt})-\Sigma_{j}(z_{jt})\big]^2\frac{1}{(\min_{t}\hat\Sigma_{j}(z_{jt}))^2 \ (\min_{t}\Sigma_{j}(z_{jt}))^2},\\
    \lesssim & \E_T\big[\hat\Sigma_{j}(z_{jt})-\Sigma_{j}(z_{jt})\big]^2\frac{1}{(\min_{t}\Sigma_{j}(z_{jt})-\max_{t}|\hat\Sigma_{j}(z_{jt})-\Sigma_{j}(z_{jt})|)^2}\\
    \le & \E_T\big[\hat\Sigma_{j}(z_{jt})-\Sigma_{j}(z_{jt})\big]^2\frac{1}{(\inf_{z}\Sigma_{j}(z)-\sup_{z}|\hat\Sigma_{j}(z)-\Sigma_{j}(z)|)^2}\\
    =&\E_T\big[\hat\Sigma_{j}(z_{jt})-\Sigma_{j}(z_{jt})\big]^2\ O_p(1)=o_p(T^{-1/2}).
\end{align*}
\end{proof}

The central result is the asymptotic-linear representation, which discharges Condition~\ref{cond:normality}\eqref{cond:norm-linear}. As in the kernel case, first-stage estimation leaves no first-order footprint: the projection of the series error onto the weight $a(z_{jt})$ collapses to the sample average of $\xi_{jt}a(z_{jt})$.

\begin{lemma}[Series influence function]\label{lemma:series-linear}
Let $a(z_{jt})=\Lambda_j\,\Gamma_{jt}$ and $\Lambda_{j}$ satisfies $0<\underline c<\Lambda_{j}(z)<\bar c$ for all $z\in\mathcal Z$, where we suppress the subscript $j$ of $a$ for simplicity. Elements of $a(z_{jt})$, $a^k(z_{jt})$ satisfies approximation condition, i.e., there exist coefficient vectors $\Pi^{a^k}_{K,j}$ with $\mu_{a^k,j}(z)=\E[a^k(z_{jt})\mid z_{jt}=z]$ such that
\[
    \sup_{z\in\mathcal Z}\big|\mu_{a^k,j}(z)-\phi^K(z)'\Pi^{a^k}_{K,j}\big|=O(K^{-\gamma}).
\]
Under Conditions~\ref{cond:regularity}, \ref{cond:y-smooth} and \ref{cond:series},
\[
    \sqrt T\,\E_J\E_T\Big[\big((\hat\mu_{y,j}-\mu_{y,j})-\alpha_0(\hat\mu_{p,j}-\mu_{p,j})\big)\,a(z_{jt})\Big]
    =\sqrt T\,\E_T\,\E_J\big[\xi_{jt}a(z_{jt})\big]+o_p(1),
\]

so Condition~\ref{cond:normality}\eqref{cond:norm-linear} holds with $\delta(w_t)=\E_J[\xi_{jt}a(z_{jt})]$.
\end{lemma}

\begin{proof}[Proof of Lemma~\ref{lemma:series-linear}]
Fix $j$ and assume $a(z_{jt})$ is a scalar (a coordinate of $a_{jt}$) w.l.o.g.; it suffices to prove, for a generic regressand $\psi$ with projection residual $e^{\psi}_{jt}=\psi_{jt}(\sigma_0)-\mu_{\psi,j}(z_{jt};\sigma_0)$ evaluated at $\sigma_0$,
\begin{equation}\label{eq:series-target}
    \sqrt T\,\E_T\big[(\hat\mu_{\psi,j}(z_{jt};\sigma_0)-\mu_{\psi,j}(z_{jt};\sigma_0))\,a(z_{jt})\big]
    =\sqrt T\,\E_T\big[e^{\psi}_{jt}\,a(z_{jt})\big]+o_p(1),
\end{equation}
and then set $\psi\in\{y_j(\sigma_0,w_t),p_{jt}\}$ and combine, using $e^{y}_{jt}-\alpha_0 e^{p}_{jt}=\xi_{jt}$. The proof is in two steps.

\emph{Step 1: the empirical process is negligible.} Split
\begin{align*}
    \sqrt T\,\E_T \big[(\hat\mu_{\psi,j}-\mu_{\psi,j})(z_{jt})\,a(z_{jt})\big]
    &=\underbrace{\sqrt T\,\E_T\big[(\hat\mu_{\psi,j}-\mu_{\psi,j})(z_{jt})\,a(z_{jt})\big]-\sqrt T\E_z[(\hat\mu_{\psi,j}-\mu_{\psi,j})(z)\,a(z)]}_{(\mathrm{i})}\\
    &\quad
    +\underbrace{\sqrt T\E_z[(\hat\mu_{\psi,j}-\mu_{\psi,j})(z)\,a(z)]}_{(\mathrm{ii})}.
\end{align*}
where $\E_z$ is the expectation over $z_{jt}$ conditional on the sample $\{z_{js}\}_{s=1}^T$. The first term is an empirical process, and the second term is a centered average of the approximation error. We show the first term is $o_p(1)$, while the second term is approximately linear.
For $(\mathrm{i})$, write $\hat\mu_{\psi,j}-\mu_{\psi,j}=\phi^K(z)'(\hat\Pi_{\psi,j}-\Pi^{K}_{\psi,j})-b(z)$ and factor the centered average,
\begin{align*}
    \big|(\mathrm{i})\big|
    &=\Big|\sqrt T\,(\E_T-\E)\big[\phi^K(z)'a(z)\big]\,(\hat\Pi_{\psi,j}-\Pi^{K}_{\psi,j})-\sqrt T\,(\E_T-\E)\big[a(z)b(z)\big]\Big|\\
    &\le \Big\|\sqrt T\,(\E_T-\E)\big[\phi^K(z)'a(z)\big]\Big\|\;\big\|\hat\Pi_{\psi,j}-\Pi^{K}_{\psi,j}\big\| + \Big\|\sqrt T\,(\E_T-\E)\big[a(z)b(z)\big]\Big\|,
\end{align*}
where $b(z)=\mu_{\psi,j}(z)-\phi^K(z)'\Pi^{K}_{\psi,j}$ is the approximation error.
For the first term, the first factor has second moment $\E\|\sqrt T(\E_T-\E)[\phi^K(z)'a(z)]\|^2\le\E[\|\phi^K\|^2 a^2]\le C\,\E\|\phi^K\|^2=tr(Q_j)=O(K)$, so it is $O_p(\sqrt K)$; the second factor is $O_p(\sqrt{K\log(K\vee T)/T}+K^{-\gamma})$ by Lemma~\ref{lemma:series-rate}. 
Similarly, the second term has second moment $\E\|\sqrt T(\E_T-\E)[a(z)\;b(z)]\|^2=\E[a^2(z) b^2(z)]\le C\,K^{-2\gamma}=O(K^{-2\gamma})$, so it is $O_p(K^{-\gamma})$.
Hence
\begin{align*}
    \big|(\mathrm{i})\big|
    =O_p\!\Big(\sqrt K\big(\sqrt{K\log(K\vee T)/T}+K^{-\gamma}\big)\Big)
    =O_p\!\Big(K\sqrt{\log(K\vee T)}/\sqrt T+\sqrt K\,K^{-\gamma}\Big)=o_p(1),
\end{align*}
the last equality by Condition~\ref{cond:series}\eqref{cond:series-rate} ($\sqrt{K}K^{-\gamma}=(K/\sqrt{T})^{1/2}T^{1/4}K^{-\gamma}\le (K/\sqrt{T})^{1/2}\sqrt{T}K^{-\gamma}=o(1)$).

\emph{Step 2: the remainder equals the target asymptotically.} Let $\psi^{K}(z)=\phi^K(z)'Q_j^{-1}\E[\phi^K(z_{jt})\psi_{jt}]$ and $a^K(z)=\phi^K(z)'Q_j^{-1}\E[\phi^K(z_{jt})a(z_{jt})]$ be the population projections of $\mu_{\psi,j}$ (or $\psi$ equivalently) and $a$, and set $\Psi^K=\E[\phi^K(z_{jt})a(z_{jt})]$. Decompose the projection minus the target,
\begin{align*}
    \sqrt T\,\E_z\big[(\hat\mu_{\psi,j}-\mu_{\psi,j})a\big]-\sqrt T\,\E_T\big[e^{\psi}_{jt}a\big]
    &=\underbrace{\sqrt T\,\E_z\big[(\hat\mu_{\psi,j}-\psi^{K})a\big]-\sqrt T\,\E_T\big[e^{\psi}_{jt}a^K\big]}_{A_1}\\
    &\quad+\underbrace{\sqrt T\,\E_z\big[(\psi^K-\mu_{\psi,j})a\big]}_{A_2}
    +\underbrace{\sqrt T\,\E_T\big[e^{\psi}_{jt}(a^K-a)\big]}_{A_3}.
\end{align*}

\emph{Term $A_2$.} By orthogonality of the population projection, $\E_z[(\psi^K-\mu_{\psi,j})a^K]=0$, so
\begin{align*}
    \|A_2\|=&\sqrt T\,\big|\E_z[(\psi^K-\mu_{\psi,j})(a-a^K)]\big|\\
    \le& \sqrt T\,\|\psi^K-\mu_{\psi,j}\|_{L^2}\,\|a-a^K\|_{L^2}\\
    \le& \sqrt T\,\|\mu_{\psi,j}-\phi^K{}'\Pi^{K}_{\psi,j}\|_{L^2}\,\|a-\phi^K{}'\Pi_{a,j}\|_{L^2}
    =O\big(\sqrt T\,K^{-2\gamma}\big)=o(1),
\end{align*}
where the second inequality uses the property of linear projection and the following equalities use Condition~\ref{cond:series}\eqref{cond:series-approx} for both factors and $K^{-\gamma}=o(T^{-1/2})$.

\emph{Term $A_3$.} This is a centered i.i.d.\ average of $e^{\psi}_{jt}(a^K-a)$, where $a$ is any element of $\Lambda_j\Gamma_{jt}$, so
\begin{align*}
    \sqrt{\E\|A_3\|^2}
    =\sqrt{\E\big[e^{\psi}_{jt}{}^2(a^K-a)^2\big]}
    \lesssim \|a^K-a\|_{L^2}
    =O(K^{-\gamma})=o_p(1),
\end{align*}
by Condition~\ref{cond:series}\eqref{cond:series-approx}.

\emph{Term $A_1$.} Using $\hat\mu_{\psi,j}(z)=\phi^K(z)'\hat Q_j^{-1}\E_T[\phi^K\psi]$, there is
\begin{align*}
    A_1=& \sqrt T\,\E_z\big[(\hat\mu_{\psi,j}-\psi^{K})a\big]-\sqrt T\,\E_T\big[e^{\psi}_{jt}a^K\big]\\
    =&\sqrt T\,\Psi^K{}'\Big(\hat Q_j^{-1}\E_T[\phi^K\psi]-Q_j^{-1}\E[\phi^K\psi]\Big)
    -\sqrt T\,\Psi^K{}'Q_j^{-1}\E_T\big[\phi^K(\psi-\psi^K)\big]-\sqrt T\,\Psi^K{}'Q_j^{-1}\E_T\big[\phi^K(\psi^K-\mu_{\psi,j})\big]\\
    =&\sqrt T\,\Psi^K{}'\big(\hat Q_j^{-1}-Q_j^{-1}\big)\E_T\big[\phi^K(\psi-\psi^K)\big]-\sqrt T\,\Psi^K{}'Q_j^{-1}\E_T\big[\phi^K(\psi^K-\mu_{\psi,j})\big],
\end{align*}
Bounding with the $l_2$ norm,
\begin{align*}
    \|A_1\|
    \le& \big\|\Psi^{K\prime}(\hat Q_j^{-1}-Q_j^{-1})\big\|\;
    \Big\|\sqrt T\,\E_T\big[\phi^K(z_{jt})(\psi(z_{jt})-\psi^K(z_{jt}))\big]\Big\|+\big\|\sqrt T\,\Psi^{K\prime}Q_j^{-1}\E_T\big[\phi^K(\psi^K-\mu_{\psi,j})\big]\big\|.\\
    \le& \big\|(\hat Q_j^{-1}-Q_j^{-1})\big\|_{op}\ \|\Psi^{K\prime}\|_2\;\Big\|\sqrt T\,\E_T\big[\phi^K(z_{jt})(\psi(z_{jt})-\psi^K(z_{jt}))\big]\Big\|+\big\|\sqrt T\,\Psi^{K\prime}Q_j^{-1}\E_T\big[\phi^K(\psi^K-\mu_{\psi,j})\big]\big\|.
\end{align*}
For the first factor, $\|\hat Q_j^{-1}-Q_j^{-1}\|_{op}=\|Q_j^{-1}(\hat Q_j-Q_j)\hat Q_j^{-1}\|_{op}\lesssim \zeta_0(K)\sqrt{\log K/T}\;w.p.a.1$ (Lemma~\ref{lemma:series-gram}); and $\|\Psi^{K}\|=\sup_{\|u\|=1}u'\E[\phi^Ka]\le \sup_{\|u\|=1}\sqrt{\E u'\phi^K\phi^K{}'u}\sqrt{\E a^2}\lesssim \sqrt{\E a^2}=O(1)$. For the third factor, observe that $\E[\phi^K(\psi-\psi^K)]=0$ by the property of linear projection, so $\sqrt T\,\E_T[\phi^K(\psi-\psi^K)]$ is a centered i.i.d.\ average with second moment
\begin{align*}
    \E\Big\|\sqrt T\,\E_T[\phi^K(z_{jt})(\psi(z_{jt})-\psi^K(z_{jt}))]\Big\|^2
    &=\E\big[(\psi(z_{jt})-\psi^K(z_{jt}))^2\,\phi^K(z_{jt})'\phi^K(z_{jt})\big]\\
    &\le \zeta_0(K)^2\,\E\big[(\psi(z_{jt})-\psi^K(z_{jt}))^2\big]\\
    &\le \zeta_0(K)^2\,\E\big[(\psi(z_{jt})-\phi^K(z_{jt})'\Pi^{K}_{\psi,j})^2\big]\\
    &\lesssim \zeta_0(K)^2\,K^{-2\gamma} = o(1),
\end{align*}
by noticing the fact $\zeta_0(K)\,K^{-\gamma}=\zeta_0(K)/\sqrt{T}\ \sqrt{T}K^{-\gamma}=o(1)$. Hence the first term is $O_p(\zeta_0(K)\sqrt{\log K/T})$.
For the last term,
\begin{align*}
    \Big\|\sqrt T\,\Psi^{K\prime}Q_j^{-1}\E_T[\phi^K(\psi^K-\mu_{\psi,j})]\Big\|
    \le&\|\Psi^{K\prime}Q_j^{-1}\|\,\|\sqrt T \E_T[\phi^K(\psi^K-\mu_{\psi,j})]\|\\
    \lesssim& \,\|\Psi^{K}\|\,O_p(\E[\|\phi^K\|^2(\psi^K-\mu_{\psi,j})^2]^{1/2})\quad \text{because } \E[\phi^K(\psi^K-\mu_{\psi,j})]=0\\
    \lesssim& \zeta_0(K)\|\psi^K-\mu_{\psi,j}\|_{L^2}\\
    \le& \zeta_0(K)\|\phi^K{}'\Pi^{K}_{\psi,j}(\sigma_0)-\mu_{\psi,j}\|_{L^2}\lesssim \zeta_0(K)\,K^{-\gamma},
\end{align*}
Combining,
\begin{align*}
    \|A_1\|
    =O_p\!\Big(\big(\zeta_0(K)\sqrt{K\log K/T}+\zeta_0(K)K^{-\gamma}\big)\Big)=o_p(1),
\end{align*}
by Condition~\ref{cond:series}\eqref{cond:series-rate}.

Collecting $A_1,A_2,A_3=o_p(1)$ establishes \eqref{eq:series-target}.
\end{proof}

\section{Proofs of Main Theorems}

\begin{proof}[Proof of Theorem~\ref{thm:consistency}]
Define the population and feasible sample criteria
\begin{align*}
    Q(\theta)&=\E_J\,\E\big[\big(\mu_{y,j}(z;\sigma)-\alpha\,\mu_{p,j}(z)-x'\beta\big)^2\,\Lambda_{j}(z)\big],\\
    \hat Q_n(\theta)&=\E_J\,\E_T\big[\big(\hat\mu_{y,j}(z_{jt};\sigma)-\alpha\,\hat\mu_{p,j}(z_{jt})-x_{jt}'\beta\big)^2\,\hat\Lambda_{j}(z_{jt})\big],
\end{align*}
so that $\hat\theta=\arg\min_{\theta\in\Theta}\hat Q_n(\theta)$.
We verify the conditions of the consistency theorem for extremum estimators (\citealp{newey1994large}, Theorem~2.1): (i) $Q$ is uniquely minimized at $\theta_0$; (ii) $\Theta$ is compact; (iii) $Q$ is continuous; and (iv) $\sup_{\theta\in\Theta}|\hat Q_n(\theta)-Q(\theta)|\xrightarrow{p}0$.
 
Condition~(i) holds by Condition~\ref{cond:regularity}\eqref{cond:cons-id}: $Q(\theta)\ge0$ with $Q(\theta)=0$ if and only if $\theta=\theta_0$. Condition~(ii) is assumed in Condition~\ref{cond:regularity}\eqref{cond:cons-reg}. Condition~(iii) follows because the integrand is continuous in $\theta$ by Condition~\ref{cond:regularity}\eqref{cond:cons-bound} (By uniform boundedness and Dominated convergence theorem). It remains to establish the uniform convergence in~(iv).
 
Before splitting, record two consequences of the conditions on the weight that are used repeatedly. First, since $\hat\Lambda_j$ and $\Lambda_j\tau_T$ both vanish off $\mathcal Z_T$,
\begin{equation}\label{eq:weight-id}
    \sup_{z\in\mathcal Z}\big|\hat\Lambda_{j}(z)-\Lambda_{j}(z)\tau_T(z)\big|
    =\sup_{z\in\mathcal Z_T}\big|\hat\omega_{j}(z)^{-1}-\omega_{j}(z)^{-1}\big|=o_p(1)
\end{equation}
by Condition~\ref{cond:consistency}\eqref{cond:cons-uniform}; together with $\Lambda_j\le \overline C$ (Condition~\ref{cond:regularity}\eqref{cond:cons-weight}) this gives $0\le\sup_{z\in\mathcal{Z}}|\hat\Lambda_{j}(z)|\lesssim\overline C$ with probability approaching one. Second, by Markov's inequality and Condition~\ref{cond:consistency}\eqref{cond:cons-trim},
\begin{equation}\label{eq:trim-mass}
    \E_J\E_T\big[1-\tau_T(z_{jt})\big]=O_p(p_T)=o_p(1).
\end{equation}
 
To separate the first-stage estimation error from the sampling variation, introduce the \emph{infeasible} criterion that replaces the estimated nuisance functions and weight by their population counterparts,
\[
    \tilde Q_n(\theta)=\E_J\,\E_T\big[\big(\mu_{y,j}(z_{jt};\sigma)-\alpha\,\mu_{p,j}(z_{jt})-x_{jt}'\beta\big)^2\,\Lambda_{j}(z_{jt})\big],
\]
and split
\[
    \sup_{\theta\in\Theta}\big|\hat Q_n(\theta)-Q(\theta)\big|
    \le
    \underbrace{\sup_{\theta\in\Theta}\big|\hat Q_n(\theta)-\tilde Q_n(\theta)\big|}_{(\mathrm{A})}
    +
    \underbrace{\sup_{\theta\in\Theta}\big|\tilde Q_n(\theta)-Q(\theta)\big|}_{(\mathrm{B})}.
\]
 
\textbf{Term (B).}
The infeasible criterion $\tilde Q_n$ is a sample average, over markets and products, of the function
$q(\theta;z)=\big(\mu_{y,j}(z;\sigma)-\alpha\mu_{p,j}(z)-x'\beta\big)^2\Lambda_{j}(z)$,
which is continuous in $\theta$ for each $z$ and, by Conditions~\ref{cond:regularity}\eqref{cond:cons-bound}, bounded by a constant uniformly over the compact set $\Theta$. A uniform law of large numbers (\citealp{newey1994large}, Lemma~2.4) therefore gives $\sup_{\theta\in\Theta}|\tilde Q_n(\theta)-Q(\theta)|=o_p(1)$.
 
\textbf{Term (A).}
Write the true and estimated residuals as
\[
    r_{jt}(\theta)=\mu_{y,j}(z_{jt};\sigma)-\alpha\,\mu_{p,j}(z_{jt})-x_{jt}'\beta,
    \qquad
    \hat r_{jt}(\theta)=\hat\mu_{y,j}(z_{jt};\sigma)-\alpha\,\hat\mu_{p,j}(z_{jt})-x_{jt}'\beta,
\]
and let $\Delta_{jt}(\theta)=\hat r_{jt}(\theta)-r_{jt}(\theta)=\big(\hat\mu_{y,j}(z_{jt};\sigma)-\mu_{y,j}(z_{jt};\sigma)\big)-\alpha\big(\hat\mu_{p,j}(z_{jt})-\mu_{p,j}(z_{jt})\big)$.
We have
\begin{align*}
    \sup_{\theta\in\Theta}\E_T\Big[|\Delta_{jt}(\theta)|^2\tau_T(z_{jt})\Big]\lesssim& \sup_{\sigma\in\Theta_\sigma}\E_T\Big[|\hat\mu_{y,j}(z_{jt};\sigma)-\mu_{y,j}(z_{jt};\sigma)|^2\tau_T(z_{jt})\Big] \\
    &+ \E_T\Big[|\hat\mu_{p,j}(z_{jt})-\mu_{p,j}(z_{jt})|^2\tau_T(z_{jt})\Big]\,\sup_{\theta\in\Theta}|\alpha|^2=o_p(1),
\end{align*}
Hence by uniform boundedness of $r_{jt}(\theta)$ over $\Theta$ (Condition~\ref{cond:regularity}),
\begin{align*}
    \sup_{\theta\in\Theta}\E_T\Big[|\hat r^2_{jt}(\theta)-r^2_{jt}(\theta)|\tau_T(z_{jt})\Big]=&\sup_{\theta\in\Theta}\E_T\Big[|\Delta_{jt}(\theta)(\Delta_{jt}(\theta)+2r_{jt}(\theta))|\tau_T(z_{jt})\Big]\\
    \lesssim & \sup_{\theta\in\Theta}\E_T\Big[|\Delta_{jt}(\theta)|^2\tau_T(z_{jt})\Big] + \sup_{\theta\in\Theta}\E_T\Big[|\Delta_{jt}(\theta)|\tau_T(z_{jt})\Big]\\
    \lesssim & \sup_{\theta\in\Theta}\E_T\Big[|\Delta_{jt}(\theta)|^2\tau_T(z_{jt})\Big] + \sup_{\theta\in\Theta}\sqrt{\E_T\Big[|\Delta_{jt}(\theta)|^2\tau_T(z_{jt})\Big]}\\
    =& o_p(1).
\end{align*}
Decompose the summand difference, adding and subtracting the \emph{trimmed} population weight $\Lambda_{j}\tau_T$,
\begin{align*}
    \hat r_{jt}^2(\theta)\,\hat\Lambda_{j}(z_{jt})-r_{jt}^2(\theta)\,\Lambda_{j}(z_{jt})
    =&\big(\hat r_{jt}^2(\theta)-r_{jt}^2(\theta)\big)\,\big(\hat\Lambda_{j}(z_{jt})-\Lambda_{j}(z_{jt})\tau_T(z_{jt})\big)
    + \big(\hat r_{jt}^2(\theta)-r_{jt}^2(\theta)\big)\,\Lambda_{j}(z_{jt})\tau_T(z_{jt})\\
    &+r_{jt}^2(\theta)\big(\hat\Lambda_{j}(z_{jt})-\Lambda_{j}(z_{jt})\tau_T(z_{jt})\big)
    -r_{jt}^2(\theta)\,\Lambda_{j}(z_{jt})\big(1-\tau_T(z_{jt})\big).
\end{align*}
The first three terms are the analogues of the decomposition with an untrimmed weight; the fourth is the price of the trimming.
 
For the second and third terms, $|r_{jt}(\theta)|$ and $|\Lambda_{j}(z_{jt})\tau_T|$ are bounded uniformly, and $\sup_{\theta}\E_J\E_T\big|\hat r_{jt}^2(\theta)-r_{jt}^2(\theta)\big|\tau_T(z_{jt})=o_p(1)$, $\E_J\E_T\big|\hat\Lambda_{j}(z_{jt})-\Lambda_{j}(z_{jt})\tau_T(z_{jt})\big|=o_p(1)$ by \eqref{eq:weight-id}, so
\[
    \sup_{\theta\in\Theta}\Big|\E_J\E_T\big[r_{jt}^2(\theta)\big(\hat\Lambda_{j}(z_{jt})-\Lambda_{j}(z_{jt})\tau_T(z_{jt})\big)\big]\Big|
    \lesssim \E_J\E_T\big|\hat\Lambda_{j}(z_{jt})-\Lambda_{j}(z_{jt})\tau_T(z_{jt})\big|=o_p(1).
\]
and
\[
    \sup_{\theta\in\Theta}\Big|\E_J\E_T\big[\big(\hat r_{jt}^2(\theta)-r_{jt}^2(\theta)\big)\,\Lambda_{j}(z_{jt})\tau_T(z_{jt})\big]\Big|
    \lesssim \E_J\E_T\big|\hat r_{jt}^2(\theta)-r_{jt}^2(\theta)\big|\tau_T(z_{jt})=o_p(1).
\]
For the first term,
\begin{align*}
    \sup_{\theta}\Big|\E_J\E_T\big[(\hat r_{jt}^2-r_{jt}^2)\big(\hat\Lambda_{j}(z_{jt})-\Lambda_{j}(z_{jt})\tau_T(z_{jt})\big)\big]\Big|
    \le& \sup_{\theta}\E_J\E_T\big|\hat r_{jt}^2-r_{jt}^2\big|\tau_T(z_{jt})\,\sup_{z\in\mathcal Z_T}\big|\hat\omega_{j}(z)^{-1}-\omega_{j}(z)^{-1}\big|=o_p(1).
\end{align*}
For the fourth term, $r_{jt}^2(\theta)\Lambda_{j}(z_{jt})$ is bounded uniformly over $\Theta$ by Conditions~\ref{cond:regularity}\eqref{cond:cons-reg}, \eqref{cond:cons-weight} and \eqref{cond:cons-bound}, so by \eqref{eq:trim-mass}
\[
    \sup_{\theta\in\Theta}\Big|\E_J\E_T\big[r_{jt}^2(\theta)\,\Lambda_{j}(z_{jt})\big(1-\tau_T(z_{jt})\big)\big]\Big|
    \lesssim \E_J\E_T\big[1-\tau_T(z_{jt})\big]=O_p(p_T)=o_p(1).
\]
So $(\mathrm{A})=o_p(1)$. Note that consistency uses only $p_T\to0$; no rate on the trimming is required.
 
\textbf{Conclusion.}
Terms (A) and (B) together yield $\sup_{\theta\in\Theta}|\hat Q_n(\theta)-Q(\theta)|=o_p(1)$, which is condition~(iv). Since conditions~(i)--(iv) hold, \citet{newey1994large}, Theorem~2.1, gives $\hat\theta\xrightarrow{p}\theta_0$.
\end{proof}

\begin{proof}[Proof of Theorem~\ref{thm:normality}]
By Theorem~\ref{thm:consistency}, $\hat\theta\xrightarrow{p}\theta_0$. A mean-value expansion of $\E_T\,g(w_t,\cdot,\hat\eta)$ about $\theta_0$ gives
\[
    0=\E_T\,g(w_t,\hat\theta,\hat\eta)=\E_T\,g(w_t,\theta_0,\hat\eta)+\hat M\,(\hat\theta-\theta_0),
    \qquad \hat M=\E_T\,\nabla_{\theta'}g(w_t,\bar\theta,\hat\eta),
\]
with $\bar\theta\xrightarrow{p}\theta_0$, so that
\begin{equation}\label{eq:linearization}
    \sqrt T\,(\hat\theta-\theta_0)=-\hat M^{-1}\,\sqrt T\,\E_T\,g(w_t,\theta_0,\hat\eta).
\end{equation}
The claim follows from Step~1, $\hat M\xrightarrow{p}M$, and Step~2, $\sqrt T\,\E_T\,g(w_t,\theta_0,\hat\eta)\xrightarrow{d}N(0,V_\delta)$. Throughout we use \eqref{eq:weight-id} and \eqref{eq:trim-mass} from the proof of Theorem~\ref{thm:consistency}.
 
\medskip
\noindent\textbf{Step 1: Jacobian.}
Since $\hat\eta=(\hat\mu_{p},\hat\mu_{y},\hat{\dot\mu}_{y},\hat \Lambda)$, $\hat{r}_{jt}(\theta)=\hat\mu_{y,j}(z_{jt};\sigma)-\alpha\hat\mu_{p,j}(z_{jt})-x_{jt}'\beta$, we have $g(w_t,\theta,\hat{\eta})=\E_J[\hat{r}_{jt}(\theta)\,\hat\Lambda_{j}\,\hat\Gamma_{jt}]$ with $\hat\Gamma_{jt}=(\hat{\dot\mu}_{y,j},-\hat\mu_{p,j},-x_{jt})'$ and $\nabla_{\theta'}\hat r=\hat\Gamma_{jt}'$,
\[
    \nabla_{\theta'}g(w_t,\bar\theta,\hat\eta)=\E_J\big[\hat\Lambda_{j}\,\hat\Gamma_{jt}\,\hat\Gamma_{jt}'\big]
    +\E_J\big[\hat r_{jt}(\theta)\,\hat\Lambda_{j}\,\nabla_{\theta'}\hat\Gamma_{jt}\big],
\]
and $\nabla_{\theta'}\hat\Gamma_{jt}$ contributes only the curvature entry $\ddot\mu_{y,j}$ in the $(\sigma,\sigma)$ block. Evaluate at $(\bar\theta,\hat\eta)$. Each averaged product of first-stage objects converges to its population value at $(\theta_0,\eta_0)$: for a representative term, adding and subtracting the trimmed weight $\Lambda_j\tau_T$ and using Conditions~\ref{cond:normality}\eqref{cond:norm-deriv} together with $\underline c\le\Lambda_{j}\le\overline C$,
\begin{align*}
    &\big|\E_J\E_T[\hat\Lambda_{j}\hat{\dot\mu}_{y,j}(\bar\sigma)^2]-\E_J\E[\Lambda_{j}\dot\mu_{y,j}(\sigma_0)^2]\big|\\
    \le & \big|\E_J\E_T[(\hat\Lambda_{j}-\Lambda_{j}\tau_T)(\hat{\dot\mu}_{y,j}(\bar\sigma)^2-\dot\mu_{y,j}(\sigma_0)^2)]\big|+\big|\E_J\E_T[\Lambda_{j}\tau_T(\hat{\dot\mu}_{y,j}(\bar\sigma)^2-\dot\mu_{y,j}(\sigma_0)^2)]\big|\\
    &+\big|\E_J\E_T[(\hat{\Lambda}_{j}-\Lambda_{j}\tau_T)\dot\mu_{y,j}(\sigma_0)^2]\big|
    +\big|\E_J\E_T[\Lambda_{j}(1-\tau_T)\dot\mu_{y,j}(\sigma_0)^2]\big|
    +\E_J(\E_T-\E)[\Lambda_{j}\dot\mu_{y,j}(\sigma_0)^2]\\
    \le &\E_J\E_T\Big[|\hat{\dot\mu}_{y,j}(\bar\sigma)^2-\dot\mu_{y,j}(\sigma_0)^2|\tau_T(z_{jt})\Big]\,\sup_{z\in\mathcal Z_T}|\hat \omega_{j}(z)^{-1}-\omega_{j}(z)^{-1}|+\E_J\E_T\Big[|\hat{\dot\mu}_{y,j}(\bar\sigma)^2-\dot\mu_{y,j}(\sigma_0)^2|\tau_T(z_{jt})\Big]\\
    &+\E_J\E_T|\hat{\Lambda}_{j}-\Lambda_{j}\tau_T|+ \overline C\,\E_J\E_T|1-\tau_T|+o_p(1)=o_p(1),
\end{align*}
the fourth term being $O_p(p_T)=o_p(1)$ by \eqref{eq:trim-mass} and the last $o_p(1)$ by the uniform law of large numbers. For the first and second term, by (\ref{eq:weight-id}), together with the following bound,
\begin{align*}
    &\E_T\Big[|\hat{\dot\mu}_{y,j}(\bar\sigma)^2-\dot\mu_{y,j}(\sigma_0)^2|\tau_T(z_{jt})\Big]\\
    \le &
    \sup_{\sigma}\E_T\Big[|\hat{\dot\mu}_{y,j}(\sigma)^2-\dot\mu_{y,j}(\sigma)^2|\tau_T(z_{jt})\Big]+\E_T\Big[|\dot\mu_{y,j}(\bar \sigma)^2-\dot\mu_{y,j}(\sigma_0)^2|\tau_T\Big]\\
    \le& \sup_{\sigma}\E_T\big[|\hat{\dot\mu}_{y,j}(\sigma)-\dot\mu_{y,j}(\sigma)|^2\tau_T+2|\hat{\dot\mu}_{y,j}(\sigma)-\dot\mu_{y,j}(\sigma)|\;|\dot\mu_{y,j}(\sigma)|\tau_T\big]+\E_T|\dot\mu_{y,j}(\bar \sigma)-\dot\mu_{y,j}(\sigma_0)|\ |\dot\mu_{y,j}(\bar \sigma)+\dot\mu_{y,j}(\sigma_0)|\\
    \lesssim& \sup_{\sigma}\E_T\big[|\hat{\dot\mu}_{y,j}(\sigma)-\dot\mu_{y,j}(\sigma)|^2\tau_T+|\hat{\dot\mu}_{y,j}(\sigma)-\dot\mu_{y,j}(\sigma)|\tau_T\big]+\E_T|\dot\mu_{y,j}(\bar \sigma)-\dot\mu_{y,j}(\sigma_0)|\\
    \lesssim& \sup_{\sigma}\E_T|\hat{\dot\mu}_{y,j}(\sigma)-\dot\mu_{y,j}(\sigma)|^2\tau_T + \sqrt{\sup_{\sigma}\E_T|\hat{\dot\mu}_{y,j}(\sigma)-\dot\mu_{y,j}(\sigma)|^2\tau_T}+|\bar\sigma-\sigma_0|=o_p(1),
\end{align*}
by Condition~\ref{cond:normality}\eqref{cond:norm-deriv} and Condition~\ref{cond:y-smooth}.
The same argument applies to the $\hat\Lambda_{j}\hat{\dot\mu}_{y,j}\hat\mu_{p,j}$, $\hat\Lambda_{j}\hat{\dot\mu}_{y,j}x$, $\hat\Lambda_{j}\hat\mu_{p,j}^2$, $\hat\Lambda_{j}\hat\mu_{p,j}x$, $\hat\Lambda_{j}xx'$ blocks and $r_{jt}(\bar\theta)\hat{\Lambda}_{j}\hat{\ddot\mu}_{y,j}(\bar\sigma)$. 
The curvature term satisfies\\
$\E_J\E_T[\hat\Lambda_{j}\hat{\ddot\mu}_{y,j}(\bar\sigma)\,\hat r_{jt}(\bar\theta)]\xrightarrow{p}\E_J\E[\Lambda_{j}\ddot\mu_{y,j}(\sigma_0)\,r_{jt}(\theta_0)]=0$,
the limit vanishing because $r_{jt}(\theta_0)=0$, which holds for every admissible weight. Collecting blocks yields $\hat M\xrightarrow{p}M$ in the stated outer-product form, with no curvature contribution. As in Term~(A) above, the trimming enters only through $\E_J\E_T[1-\tau_T]=O_p(p_T)$, so Step~1 requires no rate on $p_T$.
 
\medskip
\noindent\textbf{Step 2: Moment.}
Because $g(w_t,\theta_0,\eta_0)=0$,
\[
    \E_T\,g(w_t,\theta_0,\hat\eta)=\E_T\big[g(w_t,\theta_0,\hat\eta)-g(w_t,\theta_0,\eta_0)\big].
\]
With $\hat r_{jt}(\theta_0)=(\hat\mu_{y,j}-\mu_{y,j})-\alpha_0(\hat\mu_{p,j}-\mu_{p,j})$ (using $r_{jt}(\theta_0)=0$), expand the weighted product into a part linear in $\hat\eta-\eta_0$ and a quadratic remainder,
\[
    g(w_t,\theta_0,\hat\eta)-g(w_t,\theta_0,\eta_0)
    =\underbrace{\E_J\Big[\hat r_{jt}(\theta_0)\,\hat{\Lambda}_{j}\!\begin{bmatrix}\dot\mu_{y,j}\\-\mu_{p,j}\\-x_{jt}\end{bmatrix}\Big]}_{\text{linear}}
    +\underbrace{\E_J\Big[\hat r_{jt}(\theta_0)\,\hat{\Lambda}_{j}\!\begin{bmatrix}\hat{\dot\mu}_{y,j}-\dot\mu_{y,j}\\-(\hat\mu_{p,j}-\mu_{p,j})\\0\end{bmatrix}\Big]}_{\text{quadratic}}.
\]
\textbf{Quadratic term.}
By Cauchy--Schwarz and the boundedness of $\hat\Lambda_{j}$, the quadratic term is dominated by the mean-square errors in Condition~\ref{cond:normality}\eqref{cond:norm-msr}, so $\sqrt T\,\E_T[\text{quadratic}]=o_p(1)$. It is here that the trimming is used: $\hat\Lambda_{j}$ is uniformly bounded with probability approaching one by \eqref{eq:weight-id}, whereas an untrimmed $\hat\omega_{j}^{-1}$ need not be, being a nonparametric fit near $\partial\mathcal Z$.
Specifically, introduce the shorthand, all evaluated at $\sigma_0$,
\[
    \Delta y_{jt}=\hat\mu_{y,j}(z_{jt};\sigma_0)-\mu_{y,j}(z_{jt};\sigma_0),\quad
    \Delta p_{jt}=\hat\mu_{p,j}(z_{jt})-\mu_{p,j}(z_{jt}),\quad
    \Delta \dot y_{jt}=\hat{\dot\mu}_{y,j}(z_{jt};\sigma_0)-\dot\mu_{y,j}(z_{jt};\sigma_0),
\]
so that $\hat r_{jt}(\theta_0)=\Delta y_{jt}-\alpha_0\,\Delta p_{jt}$ and the quadratic term reads
\[
    \E_J\Big[\hat r_{jt}(\theta_0)\,\hat{\Lambda}_{j}(z_{jt})\!\begin{bmatrix}\Delta\dot y_{jt}\\-\Delta p_{jt}\\0\end{bmatrix}\Big].
\]
Taking the Euclidean norm, using the triangle inequality across the two nonzero components, and bounding the weight by $\Lambda_{j}(z_{jt})\tau_T(z_{jt})\le\overline C\tau_T(z_{jt})$ (Condition~\ref{cond:regularity}\eqref{cond:cons-weight}),
\begin{align*}
    &\Big\|\E_J\E_T\big[\hat r_{jt}(\theta_0)\,\hat\Lambda_{j}(\Delta\dot y_{jt},\,-\Delta p_{jt},\,0)'\big]\Big\|\\
    \le &\Big\|\E_J\E_T\big[\hat r_{jt}(\theta_0)\,(\hat\Lambda_{j}-\Lambda_{j}\tau_T)(\Delta\dot y_{jt},\,-\Delta p_{jt},\,0)'\big]\Big\| + \Big\|\E_J\E_T\big[\hat r_{jt}(\theta_0)\,\Lambda_{j}\tau_T(\Delta\dot y_{jt},\,-\Delta p_{jt},\,0)'\big]\Big\|\\
    \le &\E_J\E_T\big[\,|\hat r_{jt}(\theta_0)|\tau_T\,|\hat\Lambda_{j}-\Lambda_{j}\tau_T|\,(|\Delta\dot y_{jt}|+|\Delta p_{jt}|)\,\big]
    +\,\overline C\,\E_J\E_T\big[\,|\hat r_{jt}(\theta_0)|\tau_T\,(|\Delta p_{jt}|+|\Delta\dot y_{jt}|)\,\big]
\end{align*}
Observe that $|\hat r_{jt}(\theta_0)|\le|\Delta y_{jt}|+|\alpha_0|\,|\Delta p_{jt}|$, $\E_J\E_T|\hat r_{jt}(\theta_0)|^2\tau_T=o_p(T^{-1/2})$, $\E_J\E_T[\Delta\dot y_{jt}^2\tau_T]=o_p(T^{-1/2})$, and $\E_J\E_T[|\hat\Lambda_{j}-\Lambda_{j}\tau_T|^2]=\E_J\E_T[(\hat\omega_{j}^{-1}-\omega_{j}^{-1})^2\tau_T]=o_p(T^{-1/2})$ by Condition~\ref{cond:normality}\eqref{cond:norm-msr}. Each summand is a cross-product of two first-stage errors. Applying the Cauchy--Schwarz inequality to the joint average $\E_J\E_T[\,\cdot\,]$ term by term,
\begin{align*}
    \E_J\E_T\big[|\hat r_{jt}(\theta_0)|\tau_T\,|\Delta\dot y_{jt}|\big]
    &\le \big(\E_J\E_T[\hat r_{jt}(\theta_0)]^2\tau_T\big)^{1/2}\big(\E_J\E_T[\Delta\dot y_{jt}]^2\tau_T\big)^{1/2},\\
    \E_J\E_T\big[|\Delta p_{jt}|\tau_T\,|\hat r_{jt}(\theta_0)|\big]
    &\le \big(\E_J\E_T[\Delta p_{jt}]^2\tau_T\big)^{1/2}\big(\E_J\E_T[\hat r_{jt}(\theta_0)]^2\tau_T\big)^{1/2},\\
    \E_J\E_T\big[|\hat r_{jt}(\theta_0)|\tau_T\,|\hat\Lambda_{j}-\Lambda_{j}\tau_T|\,|\Delta p_{jt}|\big]
    &\le \E_J\big(\E_T[\hat r_{jt}(\theta_0)]^2\tau_T\big)^{1/2}\big(\E_T[\Delta p_{jt}]^2\tau_T\big)^{1/2}\sup_{z\in\mathcal Z_T}|\hat\omega_{j}(z)^{-1}-\omega_{j}(z)^{-1}|.
\end{align*}
Hence right-hand side above is bounded by $o_p(T^{-1/2})^{1/2}\cdot o_p(T^{-1/2})^{1/2}=o_p(T^{-1/2})$,
\[
    \Big\|\E_J\E_T\big[\hat r_{jt}(\theta_0)\,\hat\Lambda_{j}(\Delta\dot y_{jt},\,-\Delta p_{jt},\,0)'\big]\Big\|=o_p(T^{-1/2}),
\]
and therefore $\sqrt T\,\E_T[\text{quadratic}]=o_p(1)$.
 
\noindent\textbf{Linear term.}
The linear term is the left-hand side of Condition~\ref{cond:normality}\eqref{cond:norm-linear} up to one correction for the estimation of the weight, giving $\E_T\,\delta(w_t)+o_p(T^{-1/2})$.
Specifically, writing $\Gamma_{jt}=(\dot\mu_{y,j},-\mu_{p,j},-x_{jt})'$ and adding and subtracting $\Lambda_{j}\Gamma_{jt}$,
\begin{align*}
    &\E_J\E_T\big[\hat r_{jt}(\theta_0)\,\hat{\Lambda}_{j}\,\Gamma_{jt}\big]\\
    = & \E_J\E_T\big[\hat r_{jt}(\theta_0)\,\Lambda_{j}\tau_T\,\Gamma_{jt}\big]
    +\E_J\E_T\big[\hat r_{jt}(\theta_0)\,\big(\hat{\Lambda}_{j}-\Lambda_{j}\tau_T\big)\,\Gamma_{jt}\big]\\
    = & \E_T\,\delta(w_t)+o_p(T^{-1/2}) + \E_J\E_T\big[\hat r_{jt}(\theta_0)\,\big(\hat{\Lambda}_{j}-\Lambda_{j}\tau_T\big)\,\Gamma_{jt}\big]\\
    = & \E_T\,\delta(w_t)+o_p(T^{-1/2}),
\end{align*}
where the first equality is Condition~\ref{cond:normality}\eqref{cond:norm-linear} for the leading term, and the last uses Cauchy--Schwarz on the second, together with the boundedness of $\Gamma_{jt}$, to give
\[
    \Big\|\E_J\E_T\big[\hat r_{jt}(\theta_0)\,\big(\hat{\Lambda}_{j}-\Lambda_{j}\tau_T\big)\,\Gamma_{jt}\big]\Big\|
    \lesssim \big(\E_J\E_T[\hat r_{jt}(\theta_0)]^2\tau_T\big)^{1/2}
    \big(\E_J\E_T[\hat{\Lambda}_{j}-\Lambda_{j}\tau_T]^2\big)^{1/2}
    =o_p(T^{-1/4})\,o_p(T^{-1/4})=o_p(T^{-1/2}),
\]
both factors being controlled by Condition~\ref{cond:normality}\eqref{cond:norm-msr}.
The markets are independent (Condition~\ref{cond:regularity}\eqref{cond:independence}) and $\delta(w_t)$ is mean-zero and square-integrable, so the Lindeberg--L\'evy central limit theorem gives $\sqrt T\,\E_T\,\delta(w_t)\xrightarrow{d}N(0,V_\delta)$.
 
\medskip
\noindent\textbf{Conclusion.}
Substituting Steps~1--2 into \eqref{eq:linearization} and applying Slutsky's theorem,
\(
    \sqrt T\,(\hat\theta-\theta_0)\xrightarrow{d}N\big(0,\ M^{-1}V_\delta(M^{-1})'\big).
\)
\end{proof}

\begin{proof}[Proof of Theorem~\ref{thm:kernel-normality}]
It suffices to verify the hypotheses of Theorem~\ref{thm:normality}; the limiting distribution then follows.
 
First consider the preliminary estimator $\tilde\theta$, computed with equal weights $\omega_j\equiv1$, so that $\Lambda_j\equiv1$ and $\hat\Lambda_j=\tau_T$. Condition~\ref{cond:regularity}\eqref{cond:cons-weight}\eqref{cond:cons-id}, Condition~\ref{cond:consistency}\eqref{cond:cons-uniform} and the weight term of Condition~\ref{cond:normality}\eqref{cond:norm-msr} hold trivially, since $\hat\Lambda_j-\Lambda_j\tau_T\equiv0$; Condition~\ref{cond:consistency}\eqref{cond:cons-trim} holds with $p_T=O(h)$ by the remark to Condition~\ref{cond:kernel}; and Condition~\ref{cond:normality}\eqref{cond:deriv} holds by construction of the kernel derivative estimators. Lemma~\ref{lemma:kernel} discharges the remaining first-stage requirements: applying it to each $\psi\in\{y_j(\sigma),\partial_\sigma y_j(\sigma),\partial^2_\sigma y_j(\sigma),p_{jt}\}$ gives, uniformly over $\Theta_\sigma\times\mathcal Z_T$,
\begin{align*}
    \max_{j}\sup_{\Theta_\sigma\times\mathcal Z_T}\Big\{
    &\big|\hat\mu_{y,j}(z;\sigma)-\mu_{y,j}(z;\sigma)\big|
    +\big|\hat{\dot\mu}_{y,j}(z;\sigma)-\dot\mu_{y,j}(z;\sigma)\big|\\
    &+\big|\hat{\ddot\mu}_{y,j}(z;\sigma)-\ddot\mu_{y,j}(z;\sigma)\big|
    +\big|\hat\mu_{p,j}(z)-\mu_{p,j}(z)\big|\Big\}
    =O_p(\rho_T),
\end{align*}
so Conditions~\ref{cond:consistency}\eqref{cond:cons-rate} and \ref{cond:normality}\eqref{cond:norm-deriv} hold, and squaring gives the mean-square part of \eqref{cond:norm-msr},
\begin{align*}
    &\E_J\E_T\big(\hat\mu_{y,j}(z_{jt};\sigma_0)-\mu_{y,j}(z_{jt};\sigma_0)\big)^2\tau_T(z_{jt})
    +\E_J\E_T\big(\hat{\dot\mu}_{y,j}(z_{jt};\sigma_0)-\dot\mu_{y,j}(z_{jt};\sigma_0)\big)^2\tau_T(z_{jt})\\
    &+\E_J\E_T\big(\hat\mu_{p,j}(z_{jt})-\mu_{p,j}(z_{jt})\big)^2\tau_T(z_{jt})=O_p(\rho_T^2)=o_p(T^{-1/2}).
\end{align*}

Finally Lemma~\ref{lemma:kernel-linear} supplies the asymptotic-linearity Condition~\ref{cond:normality}\eqref{cond:norm-linear} with $\delta(w_t)=\E_J[\xi_{jt}\tilde a_{jt}]$, $\tilde a_{jt}=\Gamma_{jt}$; the weight $\tilde a_{jt}$ is bounded and continuous, which is all the lemma requires. All hypotheses of Theorem~\ref{thm:normality} are met, so
\[
    \sqrt{T}\,(\tilde\theta-\theta_0) \xrightarrow{d}N\!\big(0,\ \tilde M^{-1}\tilde V_\delta(\tilde M^{-1})'\big),
    \quad \tilde V_\delta = \E\big[(\E_J\xi_{jt}\tilde a_{jt})(\E_J\xi_{jt}\tilde a_{jt})'\big],
    \quad \tilde M = \E_J\E\big[\Gamma_{jt}\Gamma_{jt}'\big].
\]
Consequently $\|\tilde\theta-\theta_0\|=O_p(T^{-1/2})$, and Lemma~\ref{lemma:kernel-weight} gives the uniform rate for the estimated weight,
\begin{equation}\label{eq:kw-weight-rate}
    \sup_{z\in\mathcal Z_T}\big|\hat\Sigma_{j}(z)^{-1}-\Sigma_{j}(z)^{-1}\big|=o_p(T^{-1/4}).
\end{equation} 

Next consider the final estimator $\hat\theta$, computed with $\omega_j=\Sigma_j$ and the plug-in weight $\hat\Lambda_j=\hat\Sigma_j^{-1}\tau_T$. Condition~\ref{cond:regularity}\eqref{cond:cons-weight} holds by assumption on $\Sigma_j$. The first-stage requirements of Conditions~\ref{cond:consistency}\eqref{cond:cons-rate} and \ref{cond:normality}\eqref{cond:deriv}\eqref{cond:norm-deriv} do not involve the weight and are unchanged, as is the mean-square part of \eqref{cond:norm-msr}; the only new ingredient is the weight term of \eqref{cond:norm-msr}, which is \eqref{eq:kw-weight-rate}, and \eqref{eq:kw-weight-rate} also delivers Condition~\ref{cond:consistency}\eqref{cond:cons-uniform}. It remains to verify Condition~\ref{cond:normality}\eqref{cond:norm-linear}, which follows from Lemma~\ref{lemma:kernel-linear} applied with $a=\Sigma_j^{-1}\Gamma_j$: this weight is bounded and continuous on $\mathcal Z$, since $\Sigma_j$ is continuous and bounded away from zero and infinity and $\mu_{p,j},\dot\mu_{y,j}$ are continuous by Condition~\ref{cond:kernel}\eqref{cond:kernel-psi}. Combining $\psi=y_j(\sigma_0,w_t)$ and $\psi=p_{jt}$ through $e^{y}_{jt}-\alpha_0e^{p}_{jt}=\xi_{jt}$ gives $\delta(w_t)=\E_J[\xi_{jt}a_{jt}]$ with $a_{jt}=\Sigma_j(z_{jt})^{-1}\Gamma_{jt}$. Hence all hypotheses of Theorem~\ref{thm:normality} are satisfied for $\hat\theta$, and the conclusion follows with $V=V_\delta=\E[(\E_J a_{jt}\xi_{jt})(\E_J a_{jt}\xi_{jt})']$.
\end{proof}

\begin{proof}[Proof of Theorem~\ref{thm:series-normality}]
The series first stage is well behaved on all of $\mathcal Z$, since Condition~\ref{cond:series}\eqref{cond:series-gram} bounds the eigenvalues of $Q_j$ globally, so we take $\tau_T\equiv1$ and $p_T=0$; Condition~\ref{cond:consistency}\eqref{cond:cons-trim} holds trivially and $\hat\Lambda_j=\hat\omega_j^{-1}$ throughout. It suffices to verify the hypotheses of Theorem~\ref{thm:normality} twice, once for each stage.
 
\emph{Preliminary estimator $\tilde\theta$, with $\omega_j\equiv1$.} Then $\Lambda_j\equiv1$ and $\hat\Lambda_j\equiv1$, so Condition~\ref{cond:regularity}\eqref{cond:cons-weight}, Condition~\ref{cond:consistency}\eqref{cond:cons-uniform} and the weight term of Condition~\ref{cond:normality}\eqref{cond:norm-msr} hold trivially, and Condition~\ref{cond:normality}\eqref{cond:deriv} holds by construction, the derivative estimators being the $\sigma$-derivatives of the series fit with the design free of $\sigma$. Lemma~\ref{lemma:series-rate} verifies Condition~\ref{cond:consistency}\eqref{cond:cons-rate} and Condition~\ref{cond:normality}\eqref{cond:norm-deriv} and the mean-square part of \eqref{cond:norm-msr}, and Lemma~\ref{lemma:series-linear} verifies Condition~\ref{cond:normality}\eqref{cond:norm-linear} with $\delta(w_t)=\E_J[\xi_{jt}\tilde a_{jt}]$, $\tilde a_{jt}=\Gamma_{jt}$. Theorem~\ref{thm:normality} therefore applies and gives, in particular, $\|\tilde\theta-\theta_0\|=O_p(T^{-1/2})$.
 
\emph{Final estimator $\hat\theta$, with $\omega_j=\Sigma_j$.} Condition~\ref{cond:regularity}\eqref{cond:cons-weight} holds by assumption on $\Sigma_j$. The first-stage requirements just verified do not involve the weight and are unchanged. With $\|\tilde\theta-\theta_0\|=O_p(T^{-1/2})$ in hand, Lemma~\ref{lemma:series-weight} gives
\[
    \sup_{z\in\mathcal Z}\big|\hat\Sigma_j(z)^{-1}-\Sigma_j(z)^{-1}\big|=o_p(1),
    \qquad
    \E_T\big[\hat\Sigma_j(z_{jt})^{-1}-\Sigma_j(z_{jt})^{-1}\big]^2=o_p(T^{-1/2}),
\]
which are Condition~\ref{cond:consistency}\eqref{cond:cons-uniform} and the weight term of Condition~\ref{cond:normality}\eqref{cond:norm-msr}. Condition~\ref{cond:normality}\eqref{cond:norm-linear} follows from Lemma~\ref{lemma:series-linear} applied with $a_{jt}=\Sigma_j(z_{jt})^{-1}\Gamma_{jt}$, whose coordinates satisfy the approximation requirement of Condition~\ref{cond:series}\eqref{cond:series-approx}, giving $\delta(w_t)=\E_J[\xi_{jt}a_{jt}]$. All hypotheses of Theorem~\ref{thm:normality} are met, and the conclusion follows with $V=\operatorname{Var}(\delta(w_t))$ as stated.
\end{proof}

\end{appendices}

\noindent\textbf{Declaration of generative AI and AI-assisted technologies in the manuscript preparation process}

\noindent During the preparation of this work, the authors used Claude for structuring the paper's writing framework and supporting code writing and debugging.
The authors reviewed and edited the output as needed and take full responsibility for the content of the published article. 

\end{document}